\documentclass[pdflatex,sn-mathphys]{sn-jnl}

\usepackage{pdfpages}

\usepackage{subfigure} 
\usepackage{bm}
\usepackage{graphicx}
\usepackage{booktabs}
\usepackage{soul}

\usepackage[shortlabels]{enumitem} 

\jyear{2021}%

\theoremstyle{thmstyleone}%
\newtheorem{theorem}{Theorem}
\newtheorem{lemma}{Lemma}
\theoremstyle{thmstyletwo}%

\theoremstyle{thmstylethree}%
\newtheorem{definition}{Definition}%

\begin{document}

\title[POLYLLA: Polygonal meshing algorithm based on terminal-edge regions]{POLYLLA: Polygonal meshing algorithm based on terminal-edge regions}


\author*[1]{\fnm{Sergio} \sur{Salinas}}\email{ssalinas@dcc.uchile.cl}
\equalcont{These authors contributed equally to this work.} 
\author*[1]{\fnm{Nancy} \sur{Hitschfeld-Kahler} \dgr{} }\email{nancy@dcc.uchile.cl}
\equalcont{These authors contributed equally to this work.}
\author[2]{\fnm{Alejandro} \sur{Ortiz-Bernardin} \dgr{}}\email{aortizb@uchile.cl}
\author[3]{\fnm{Hang} \sur{Si} \dgr{}}\email{si@wias-berlin.de}

\affil*[1]{\orgdiv{Department of Computer Sciences}, \orgname{Universidad de Chile}, \orgaddress{\\\street{Av. Beauchef 851}, \city{Santiago}, \postcode{8370456}, \state{RM}, \country{Chile}}}

\affil[2]{\orgdiv{Computational and Applied Mechanics Laboratory, Department of Mechanical Engineering}, \orgname{Universidad de Chile}, \orgaddress{\street{Av. Beauchef 851}, \city{Santiago}, \postcode{8370456}, \state{RM}, \country{Chile}}}

\affil[3]{\orgdiv{Numerical Mathematics and Scientific Computing}, \orgname{Weierstrass Institute for Applied Analysis and Stochastics}, \orgaddress{\street{Mohrenstr. 39}, \city{Berlin}, \postcode{10117}, \state{Berlin}, \country{Germany}}}


\abstract{This paper presents an algorithm to generate a new kind of polygonal mesh obtained  from triangulations. Each polygon is built from a terminal-edge region surrounded by edges that are not the longest-edge of any of the two triangles that share them. The algorithm is termed Polylla and is divided into three phases.  The first phase consists of labeling each edge of the input triangulation according to its size; the second phase builds polygons (simple or not) from terminal-edges regions using the label system; and the third phase transforms each non simple polygon into simple ones. The final mesh contains polygons with convex and non convex shape. Since Voronoi based meshes are currently the most used polygonal meshes, we compare some geometric properties of our meshes against constrained Voronoi meshes. Several experiments were run to compare the shape and size of polygons, the number of final mesh points and polygons. For the same input, Polylla meshes contain less polygons than Voronoi meshes and the algorithm is simpler and faster than the algorithm to generate constrained Voronoi meshes. Finally, we have validated  Polylla meshes by solving the Laplace equation on an L-shaped domain using the Virtual Element Method (VEM). We show that the numerical performance of the VEM using  Polylla meshes and Voronoi meshes is similar.}

\keywords{Polygonal mesh,Terminal-edge region, Virtual element method, Delaunay triangulations}



\maketitle

\section*{Article Highlights}

\begin{itemize}
\item A simple and automatic tool to generate polygonal meshes 
composed of convex or/and non-convex polygons. 
\item The polygonal meshes are composed of less polygons and points than constrained Voronoi meshes for the same input 
\item A new kind of polygonal meshes  for the virtual element method. 

\end{itemize}

\section{Introduction}\label{sec1}

Meshes composed of  triangles and quadrilaterals are common in simulations using the Finite Element Method (FEM). One of the main requirements is that polygons (elements)  need to obey specific quality criteria such as  angles not too large or small, or reasonable aspect ratio and  area,  among others. To fulfill these criteria, sometimes the insertion of a large number of points and elements is required in order to properly model a domain, increasing the time needed to make a simulation. New methods such as the Virtual Element Method (VEM)~\cite{Basisprinciples,Brezzi2015} can use any polygon as basic cell. Our main research interest is to explore how far the VEM can handle non-convex and convex polygons and still be able to compute accurate simulation results. Our goal is to build a new tool that allows the VEM community~\cite{Wriggers2019} to model and simulate more complex problems than before, both in 2D and 3D.

In this context, the VEM presents an opportunity to explore new kind of polygonal meshes and new algorithms to generate them. 
Our main research questions are: Can terminal-edge regions~\footnote{A terminal-edge region is formed by all triangles whose Longest-Edge Propagation Path (Lepp)~\cite{Rivara97} share the same terminal-edge.} be adapted to be used as good basic cells for polygonal numerical methods such as the VEM? Do this kind of meshes require less polygons to model the same problem than polygonal meshes based on the Voronoi diagram? Our hypotheses are: (i) Terminal-edge regions can be transformed into simple polygons and used  as basic cells, (ii) the domain geometry can be fitted using less elements than constrained Voronoi meshes, (iii) this kind of polygonal meshes can be used with the VEM. 

In this paper, we propose an algorithm to generate meshes that adapt to a geometric domain specified through  Planar line straight graph (PLSG) using  polygons of any shape, and respecting the  point distribution given as input. The algorithm reads as input a triangulation and uses the concept of terminal-edge region as basis to build polygons. 
Terminal-edge regions can generate non-simple polygons, so we also propose an algorithm to divide  them into simple ones. We  run several experiments to show  properties of these polygons and compare the generated meshes against constrained Voronoi meshes. Moreover, we  validate the polygonal meshes over a classical problem to show that these meshes can be used to solve problems with the VEM.   As an example, Fig. \ref{fig:logo} shows a polygonal mesh generated  with the algorithm proposed in this paper.

\begin{figure}
\centering
\includegraphics[width=0.8\linewidth]{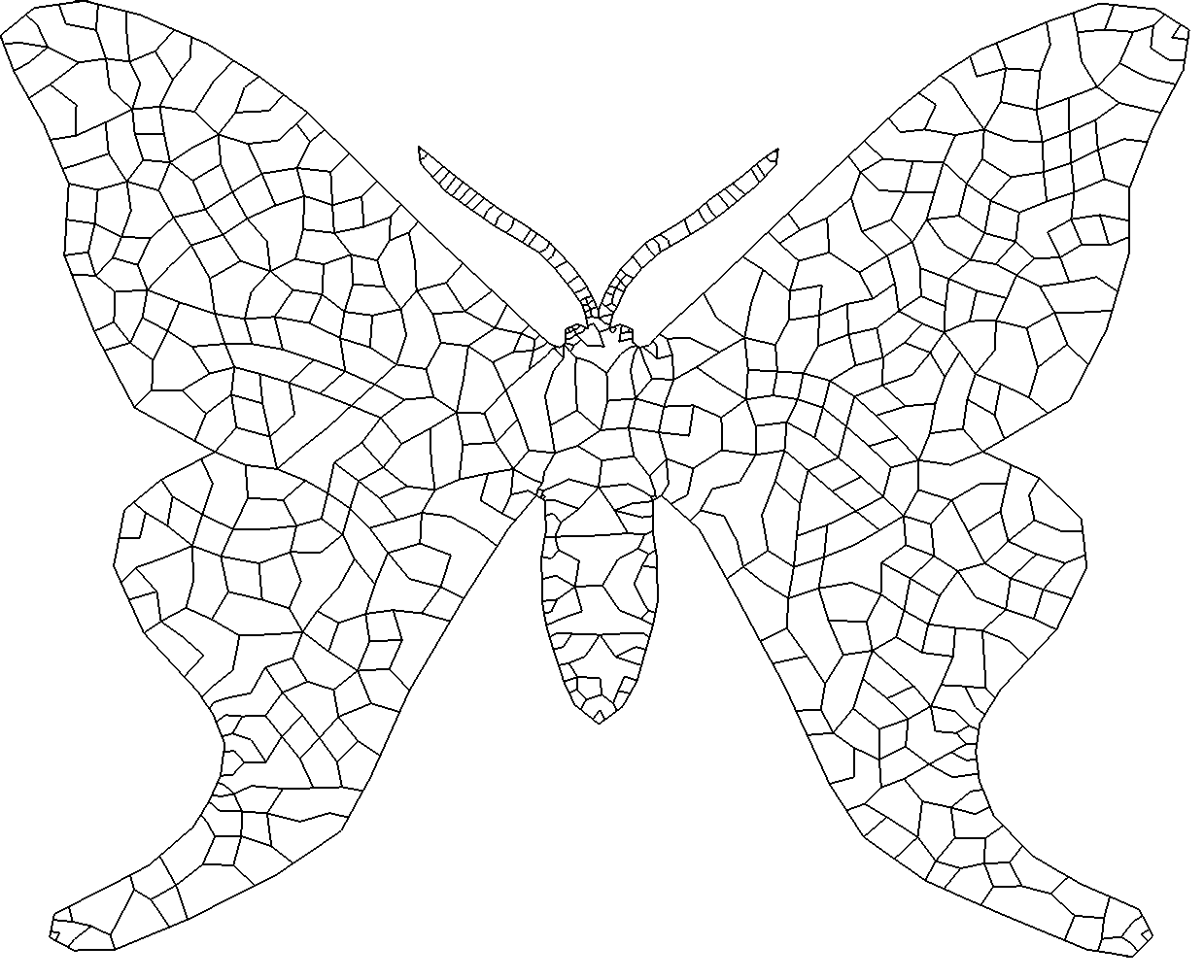} 
\caption{Polygonal mesh of a Chilean moth (Polilla Venusta) composed of 513 polygons, 1351 vertices and 2455 edges. The original PLSG contains 616 points and segments. A conforming Delaunay triangulation with  maximum edge size of 160 was generated using Pygalmesh~\cite{Schlomer_pygalmesh_Python_interface}. The resulting triangulation contains 1351 vertices (the same of the output), 2083 triangles and 3432 edges.}
\label{fig:logo}    
\end{figure}

Any triangulation can be used as input but through this paper the polygonal meshes are generated from a Delaunay triangulation. It is known that  Delaunay triangulations are the ones that maximize the smallest angle among all the triangulations of a  point set. Since the proposed algorithm does not divide any input triangle, the smallest angle of the triangulation is a lower bound for the minimum interior angle of the polygonal mesh.

%

The main contributions of this paper are:
\begin{itemize}
    \item A simple and automatic tool to generate polygonal meshes composed of convex or/and non-convex polygons, fitting exactly the input domain and respecting the initial point distribution. 
    \item A new kind of polygonal meshes composed of less polygons and points than constrained Voronoi  meshes for the same input.
    \item The algorithm benefits from robust tools such as Detri2d or Triangle to generate the initial constrained Delaunay triangulation  in similar way some constrained Voronoi meshing algorithms do.  Voronoi meshing algorithms require to compute new points (Voronoi points),  cut non-bounded  Voronoi regions to fit the domain  and  insert new points at the domain boundary  to create the constrained Voronoi mesh. So the proposed algorithm is faster than the algorithm to generate constrained Voronoi meshes.
    
\end{itemize}
\noindent
This paper is organized as follows: Section~\ref{sec:relatedwork} presents and discusses the-state-of-art;  Section~\ref{sec:basic} introduces the basic concepts  that explains the algorithm;  Section~\ref{sec:the_algorithm} describes the main steps of the proposed algorithm  and the used data structure; in Section~\ref{sec:experimental_evaluation} we analyse some geometric features of the generated meshes and gives a comparison against the constrained  Voronoi meshes; Section~\ref{sec:simulation_results} shows a preliminary assessment of the meshes in the Virtual Element Method (VEM) and Section~\ref{sec:conclusions} presents our  conclusions and ongoing work.

\section{Related work}
\label{sec:relatedwork}



Mesh generation refers to the methods used to discretize a geometric domain into smaller elements without overlap. Those methods has been widely studied due to their importance in science and engineering. Meshes are used in geographic information systems~\cite{PrinciplesGIS}, computer vision~\cite{JOHNSON1998261} and  numerical methods~\cite{HOLE198827}, among other applications. Common methods used to generate unstructured polygon meshes are the Delaunay methods \cite{cheng2013delaunay}, Voronoi diagram methods \cite{yan:hal-00647979, 2dcentroidalvoro, PolyMesher2012}, advancing front method \cite{lohner1996progress, Schberl1997NETGENAA}, quadtree based methods \cite{BERN1994384, bommes2013quad} and hybrids methods \cite{owen1999q, ito2002unstructured}, among others. In general, meshing algorithms can be classified into two groups \cite{owen1998survey, johnen2016indirect}: (i) direct algorithms:  meshes are generated from the input geometry, and (ii) indirect algorithms: meshes are generated starting from an input mesh, typically an initial triangle mesh. Indirect methods is a common approach to generate quadrilateral meshes by mixing triangles of an initial triangulation ~\cite{LeetreetoQuad, BlossonQuad, Merhof2007Aniso-5662}. The advantage of using indirect methods is that currently triangular meshes are easy to generate because several robust open source and free tools~\cite{triangle2d, Detri2, qhull,  cgal:y-t2-21b} are available. Polylla mesh generation is an indirect method as it is based  on mixing triangles from an initial triangulation. 

In standard FEM, the most used 2D meshes are triangulations~\cite{triangle2d,Detri2, chew1} and quadrilateral meshes~\cite{Canann98, Lee20031055, Owen98advancingfront}. 2D Mixed meshes composed of triangles and quadrilaterals have also been used, but are not so common as the previous ones~\cite{jaillet2021}. Other numerical methods such as  the Voronoi Cell Finite Element Method (VCFEM) \cite{ReviewFEM} and Polygonal Finite Element method (PFEM)\cite{Chi2015PolygonalFE}  use the constrained  Voronoi diagram as the polygonal mesh, where the Voronoi cells are the mesh elements~\cite{YanWLL10,UniformRandomVoronoiMeshes,SiegerAB10}. 
 
The generalization of finite element methods to include polygons as part of the mesh elements is not recent~\cite{zbMATH03504364, Sukumar2006,tabarraei2008extended}.  Polygonal elements are usually generated by using a quadtree approach and from Voronoi based algorithms. In particular, the {\sc vem} was introduced in the last decade  ~\cite{Basisprinciples,Brezzi2015} and since then, several research groups have been developing computational frameworks
in order to explore  new problems. To mention some  applications, the {\sc vem} has been formulated and applied to solve linear elastic and inelastic solid mechanics problems~\cite{da2015virtual}, in fluid mechanics~\cite{ErnestoBRinkmanFluid, ErnestoQuasiNewtwon}, in the optimization
of a fluid problem through a discrete network~\cite{benedetto2014virtual}, for compressible and incompressible finite deformations~\cite{Wriggers2017,Wriggers2017a} and brittle crack propagation~\cite{HUSSEIN201915}. 

VEM has a flexibility in dealing with complex cell shapes that can even be non convex and have an arbitrary number of vertices~\cite{Aldakheel_2019}. As an example, meshes composed of animal shape polygons were used for  crack propagation to show that  the {\sc vem}  can deal with non-convex element shapes~\cite{Wriggers2019}. Similar experiments using the {\sc vem} in  meshes composed of high irregular shapes are described in \cite{PARK2019669, CHI2017148}. 

To our knowledge, no algorithm for the automatic generation of 2D meshes formed by polygons (convex or not) has been published for the {\sc vem} in arbitrary domains. We recently developed a tool to  generate a particular kind of polygonal meshes in the context of modeling the rock packing problem  inside a square container \cite{PackingJoaquin}. The rocks were modeled as convex polygons and the space  with general polygons. We ran preliminary simulations using the {\sc vem} and also experienced that the {\sc vem} is very robust under  irregular polygonal shapes~\cite{PackingJoaquin}.

%


\section{Basic concepts}
\label{sec:basic}

The proposed polygonal meshing algorithm is based on two concepts: Longest-edge propagation path (Lepp) introduced in~\cite{Rivara97} and terminal-edge regions defined in~\cite{Ascom209, Ojeda2018ANA}. We briefly review these concepts and some related properties in this section.

\subsection{Terminal-edge regions}

In any triangulation, triangles can be grouped under the Longest-edge propagation path concept defined as follows:

\begin{definition}{\textbf{Longest-edge propagation path }~\cite{Rivara97}:}\label{d:lepp}  
For any triangle $t_0$ in any triangulation $\Omega$, the Lepp($t_0$)
is the ordered list of all the triangles $t_0$, $t_1$, $t_2$, ..., $t_{l-1}$, $t_{l}$, such that $t_{i}$ is the neighbor triangle of $t_{i-1}$  by the longest-edge of $t_{i-1}$, for $i = 1,2,\dots,l$. The longest-edge shared by $t_{l-1}$ and $t_l$ is a terminal-edge and $t_{l-1}$ and  $t_l$ are terminal-triangles. As an example, in Fig.~\ref{fig:leppex1}  the Lepp($t_0$) is shown in blue, its terminal edge is shown in red  and its terminal triangles labeled as $t_1$ and $t_2$.
\end{definition}

\begin{figure}[h]
\centering     
\subfigure[
]{\label{fig:leppex1}\includegraphics[width=0.3\textwidth]{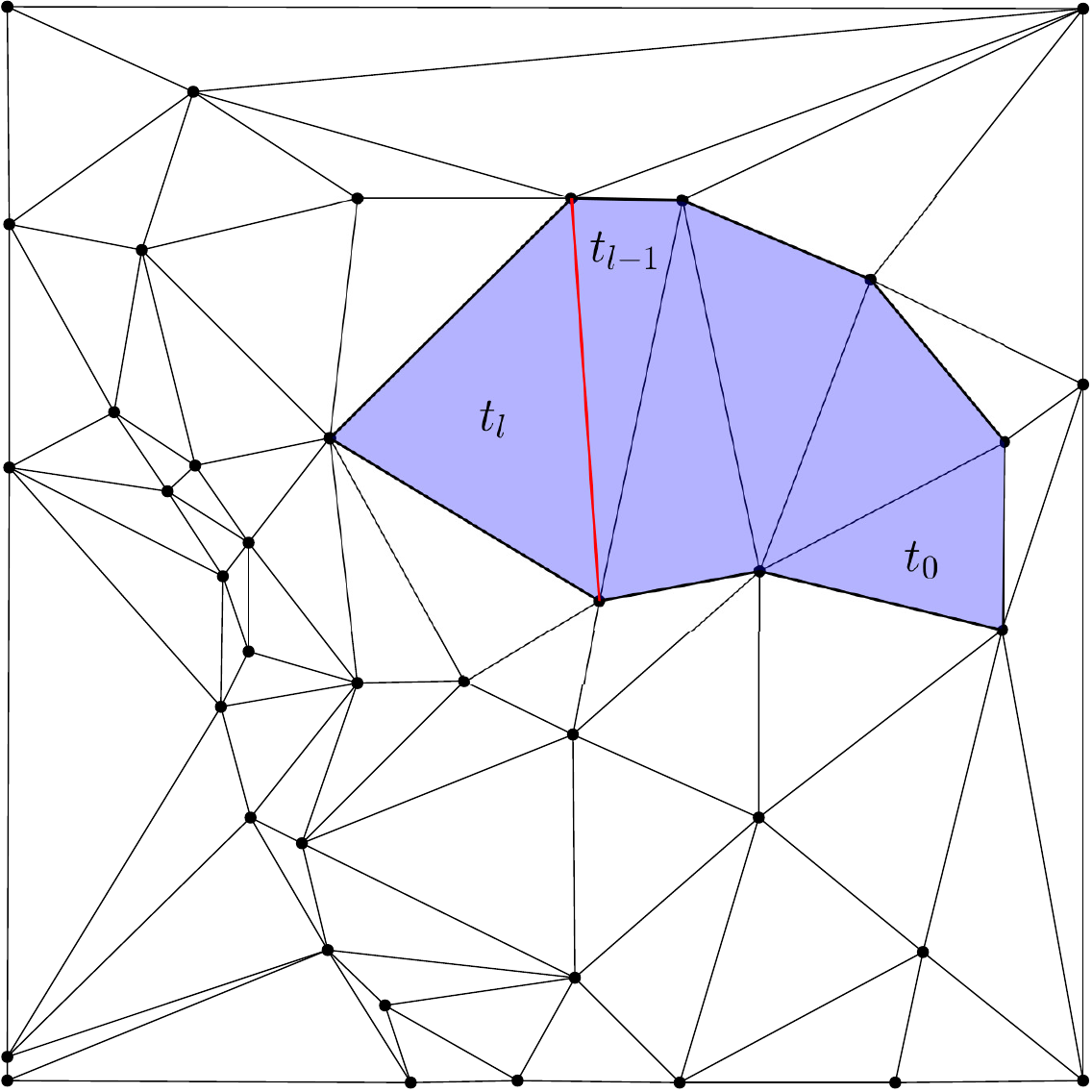}} 
\subfigure[
]{\label{fig:leppex2}\includegraphics[width=0.3\textwidth]{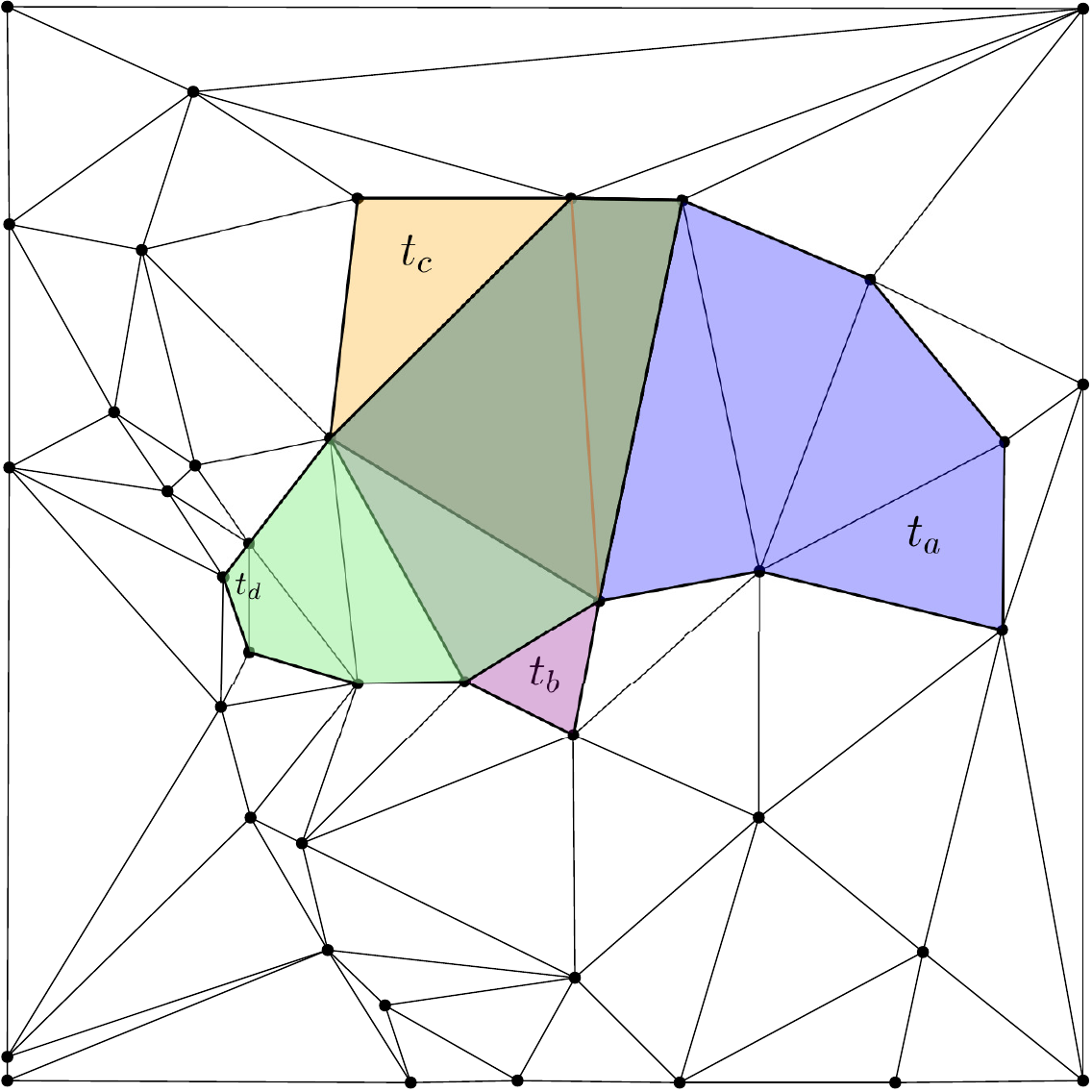}}
\subfigure[
]{\label{fig:leppex3}\includegraphics[width=0.3\textwidth]{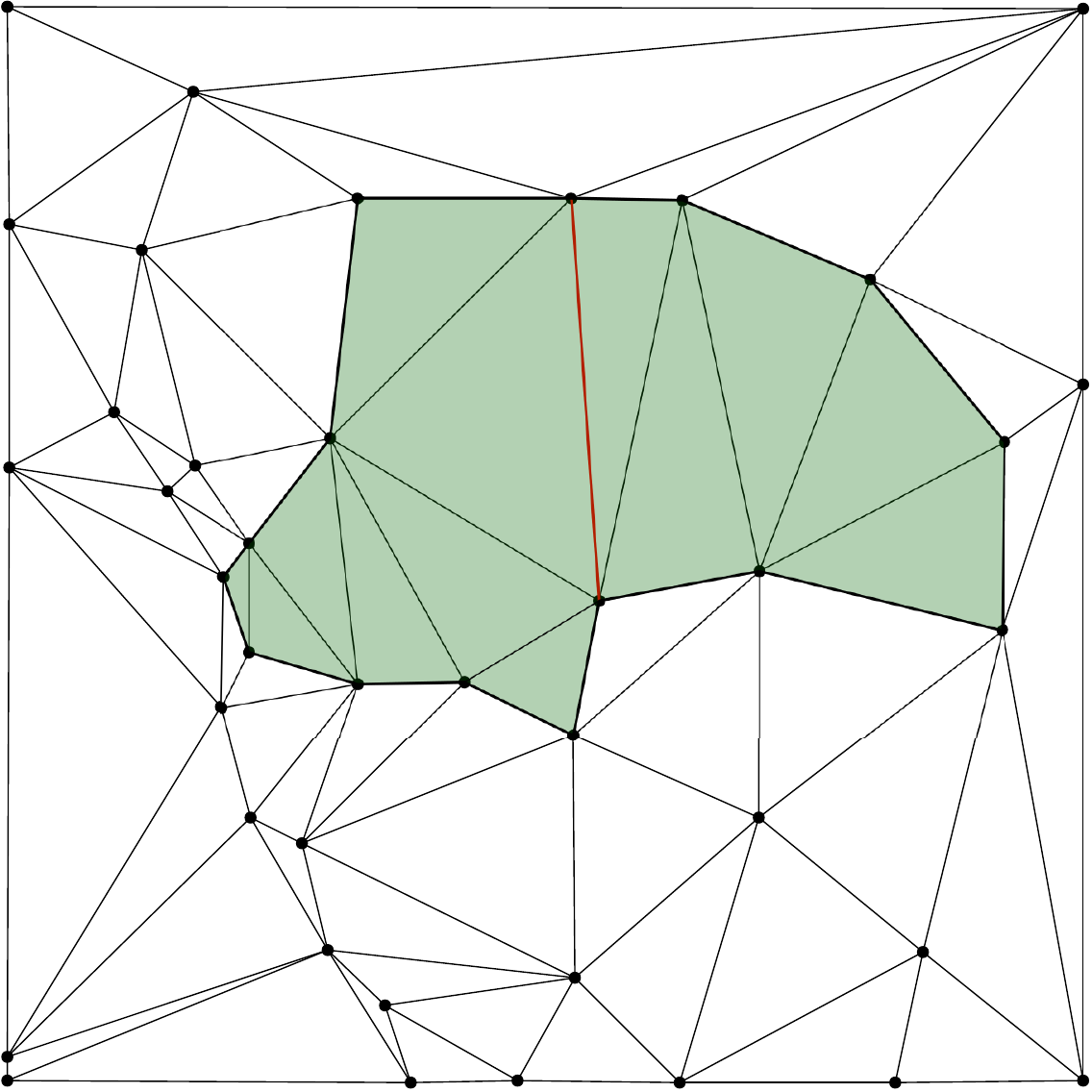}}
\caption{Terminal-edge region. \textbf{(a)} Longest-edge propagation of $t_0$ where the red edge is the terminal-edge. \textbf{(b)} Four Lepps: Lepp($t_a$), Lepp($t_b$), Lepp($t_c$) and Lepp($t_d$), with the same terminal-edge. \textbf{(c)} Terminal-edge region generated by the union of Lepp($t_a$), Lepp($t_b$), Lepp($t_c$) and Lepp($t_d$).}
\label{fig:leppexample} 
\end{figure}

\noindent
A triangle edge can be classified according to its  length inside the two triangles that share it. Therefore, given an edge $e$ and two triangles $t_1$, $t_2$ that share $e$, we can label  $e$ as: 

\begin{itemize}
    \item Terminal-edge~\cite{Rivara97}: $e$ is the longest-edge of $t_1$ and $t_2$.
    \item Frontier-edge~\cite{Ascom209}: $e$ is neither the longest-edge of $t_1$ nor $t_2$. 
    \item Internal-edge:  $e$ is the longest edge of $t_1$ but not of $t_2$ or vice-versa.
    \item Boundary edge: $e$ belongs to one triangle. In the proposed algorithm boundary edges are handled as frontier edges.
\end{itemize}

\noindent
It is worth mentioning that edges of equal length can appear. Therefore, in case there are equilateral  triangles, each edge is chosen arbitrarily as the longest-, the middle- and smallest-edge. When there are isosceles triangles, something similar is done for equal size edges.

\begin{definition}{\textbf{Terminal-edge region~\cite{Ascom209}:}}
\label{d:terminaledgeregion}
A {\em terminal-edge region} $R$ is a region formed by the union of all triangles $t_i$ such that Lepp($t_i$) has the same terminal-edge.  In case the terminal-edge region  is partially delimited by boundary edges the region will be called {\em boundary terminal-edge region}. Fig. \ref{fig:leppexample}(c) shows the terminal-edge region formed by the union of  Lepp($t_a$), Lepp($t_b$), Lepp($t_c$) and Lepp($t_d$). 
\end{definition}

\noindent
We have already used  the concept of terminal-edge region for finding polygonal voids (holes) inside a point set~\cite{HerviasHCF13,Ascom209} and proven some geometric properties. In order to facilitate understanding of the algorithm we are recalling the most important here:

\begin{itemize}
    \item Terminal-edge regions are surrounded by frontier-edges~\cite{Ascom209}.
    \item Terminal-edge regions cover the whole domain without overlapping~\cite{Ascom209,Ojeda2018ANA}.
    \item Terminal-edge-regions might include frontier-edges in their interior. We have called this kind of frontier-edge a \textbf{barrier-edge}~\cite{Ascom209,Ojeda2018ANA}.
  
\end{itemize}
\noindent



\noindent
Let us use Fig. \ref{fig:general_example} to illustrate an example of a partition generated by terminal-edge regions. Fig.~\ref{fig:initialpoinset} shows an input point set; Fig.~\ref{fig:labelsystem} shows the Delaunay triangulation of this point set where terminal-edges are drawn using red dashed lines, internal-edges using black dashed lines and frontier-edges using  solid lines. Fig.~\ref{fig:terminalregion} shows the polygons defined by  terminal-edge regions using a different color, where each polygon is delimited by  frontier edges. It can be noticed that the green polygon is a non-simple polygon because it includes a barrier-edge. \noindent
For simplicity, in the case  $e$ is a domain boundary edge, $e$ will be considered a frontier-edge too (see the edges belonging to the square in Figs. \ref{fig:labelsystem} and \ref{fig:terminalregion}). 

\begin{figure}[h]
\centering     
\subfigure[
]{\label{fig:initialpoinset}\includegraphics[width=0.3\textwidth]{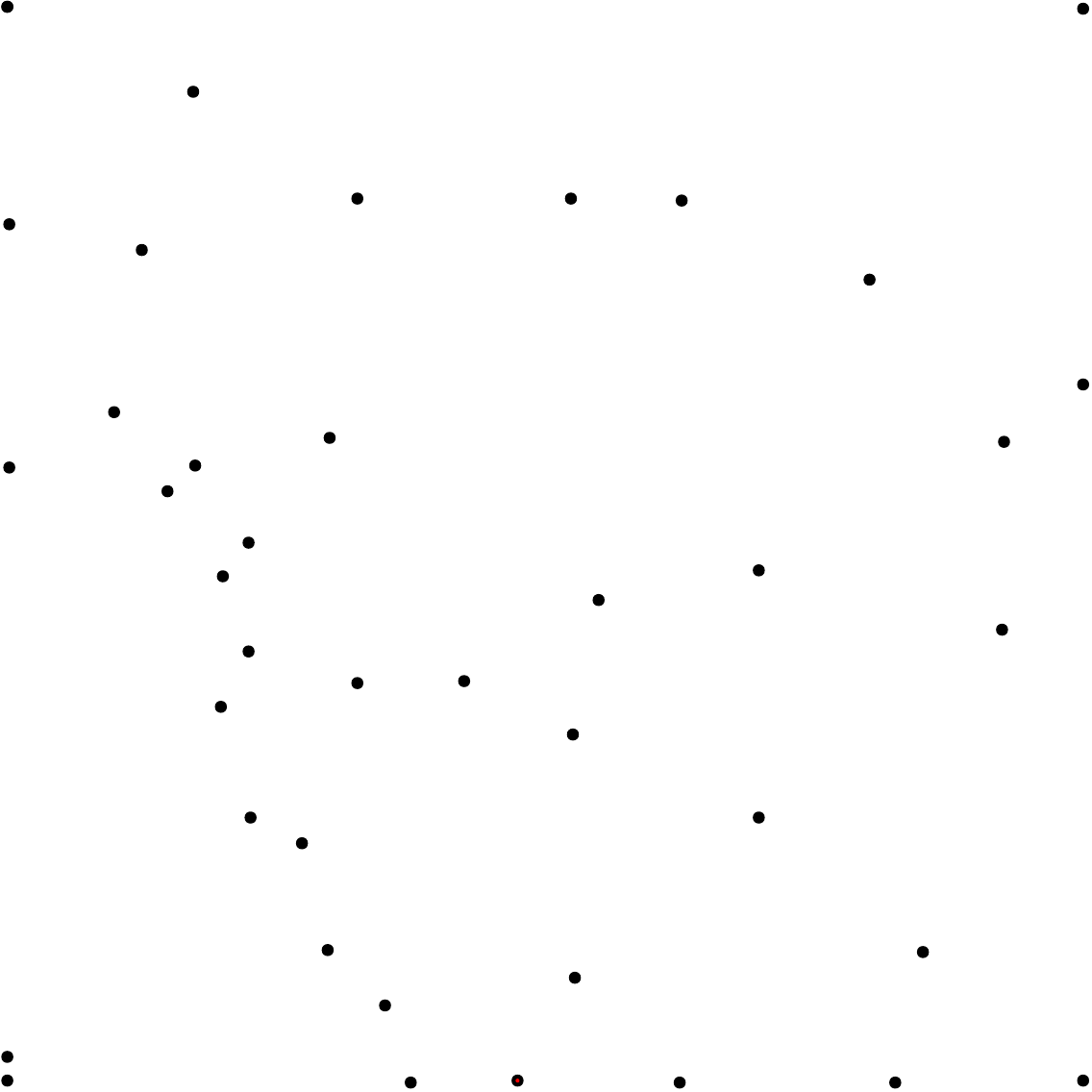}} 
\subfigure[
]{\label{fig:labelsystem}
\includegraphics[width=0.3\textwidth]{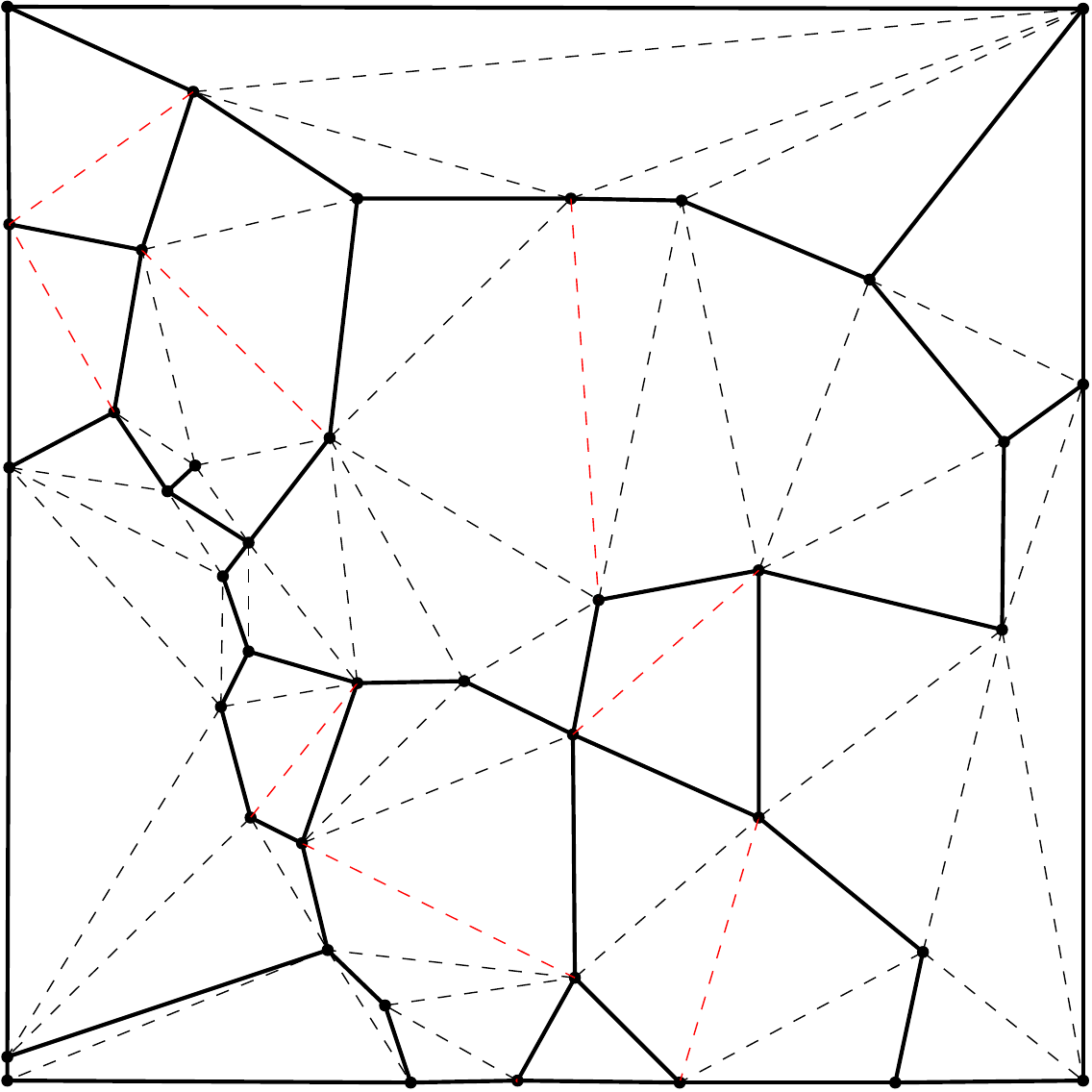}}\subfigure[
]{\label{fig:terminalregion}\includegraphics[width=0.3\textwidth]{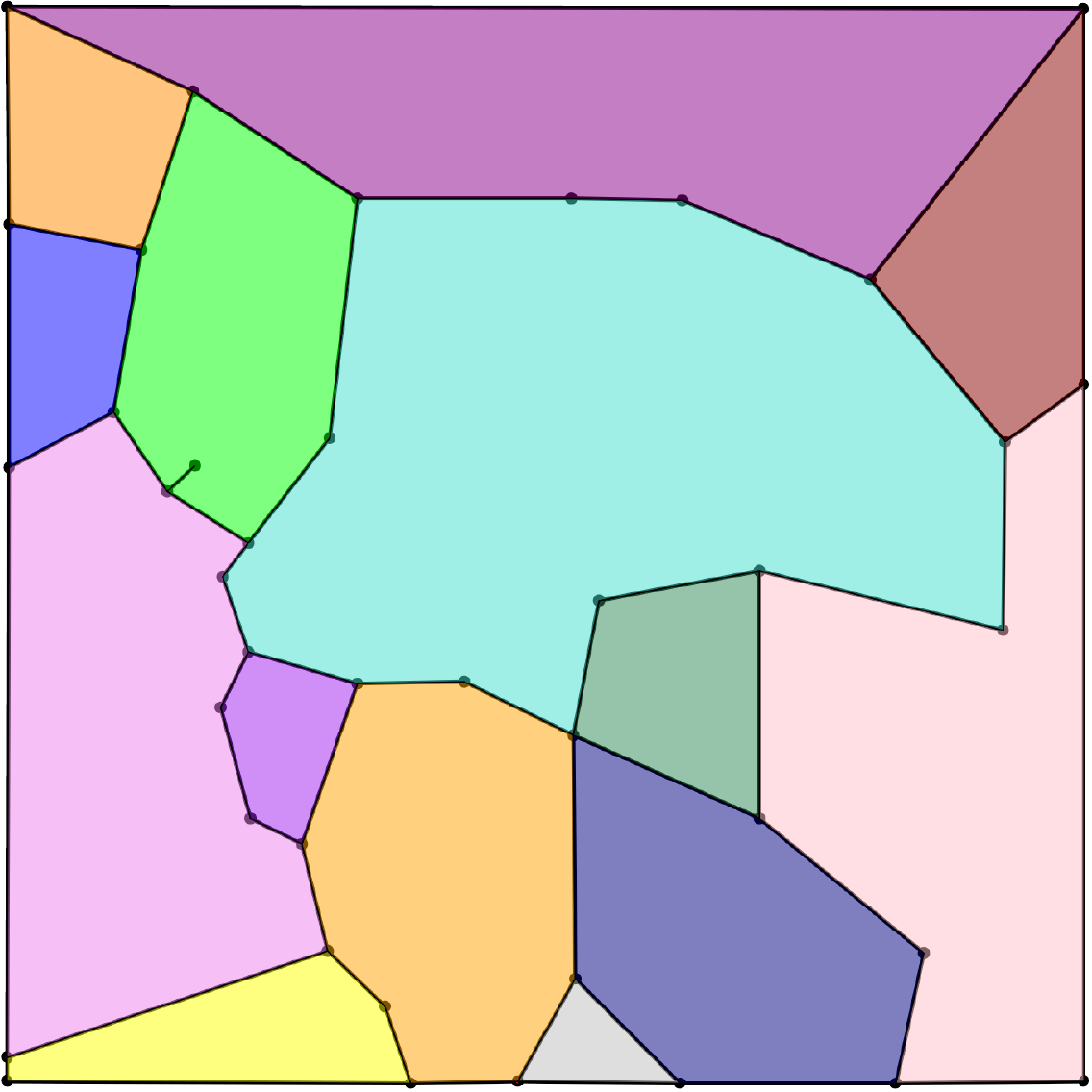}}
\caption{(a) Initial random point set. (b) Delaunay triangulation: Solid lines are frontier-edges, dashed black lines are internal-edges and red dashed lines are terminal-edges. (c) Polygon partition from terminal-edge regions.}
\label{fig:general_example} 
\end{figure}

\subsection{Terminal-edge regions as  polygons}


As we have said in the previous section, terminal-edge regions generate a polygon partition of the domain without overlapping. Depending on the point distribution,  polygons can be simple or  non-simple polygons. Non-simple polygons appear when terminal-edge regions include barrier-edges.  As observed in Fig.~\ref{fig:general_example}, the domain was tessellated into 14 polygons where only the green region is represented by a  non-simple polygon because it includes one barrier-edge. 


In order to build a conforming polygon tessellation from the partition generated by the terminal-edge regions, non-simple polygons must be divided into simple polygons. That requires the integration of barrier-edges as part of the boundary of new simple polygons.  Since this requirement is part of the meshing algorithm we are proposing in Section~\ref{sec:the_algorithm},  beforehand, we need to introduce some definitions and prove some properties that will sustain the correctness of the algorithm.

\begin{definition}{\textbf{Barrier-edge tip:}}
\label{d:barrier-edge}
A barrier-edge tip in a terminal-edge region $R$ is a barrier-edge endpoint shared by no other barrier-edge/frontier-edge.
\end{definition}

 
Fig. \ref{figs:kindofbarrieredges} includes  two polygons with barrier-edges; the barrier-edge tips are shown in  green. Fig. \ref{fig:1be} shows a case of a polygon with one barrier-edge tip and Fig. \ref{fig:2be} a case with two barrier-edge tips.
We have observed that each point of the input data is an endpoint of a frontier-edge or barrier-edge. This means that terminal-edge regions have none isolated interior points.
Theorem~\ref{D:theoreemvertices} demonstrates this property.

\begin{figure}[h]
\centering  
\subfigure[
] {\label{fig:1be}\includegraphics[width=0.15\textwidth, ]{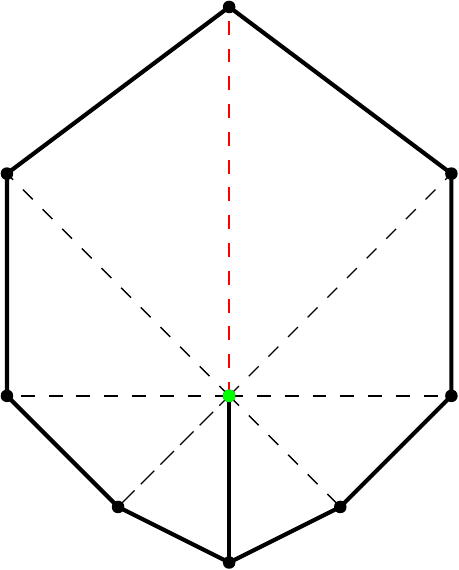}}\hspace{2cm}
\subfigure[
]{\label{fig:2be}\includegraphics[width=0.15\textwidth, angle =60]{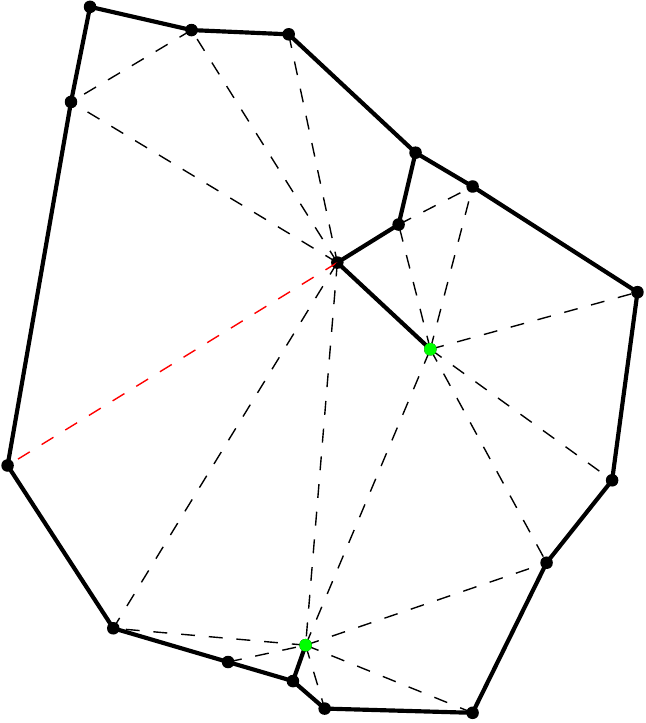}}\hspace{0.5cm}

\caption{Examples of non-simple polygons. Black lines are frontier-edges, dashed black lines are internal-edges and red edges are terminal-edges. \textbf{(a)} One barrier-edge and one tip \textbf{(b)} Four barrier-edges and two tips tip }
\label{figs:kindofbarrieredges} 
\end{figure}


\begin{theorem}
\label{D:theoreemvertices}
Let   $\Omega$ be a triangulation with the set of vertices $V$ in general position. Then, each vertex $v$ is an end-point of at least one  of  the  frontier- or barrier-edges, i.e. there is no isolated interior points (vertices incident only to internal-edges).
\end{theorem}

\begin{proof}: Let $v$ be  a vertex  associated to the terminal-edge region $R$ generated by the terminal-edge $e$ and  $T$ the set of the triangles that share $v$. By contradiction, let us assume that $v$ is an interior point of $R$ as shown  in Fig. \ref{fig:lemmaquitadomingos}. Since the triangles in $T$ are part of $R$, they must share their longest-edge around $v$. 
Given that $T$ is finite, there should exist a triangle $t_0$ (see Fig. \ref{fig:lemmaquitadomingos}) that shares their longest-edge with two triangles of $R$ in order to maintain $v$ interior point in $R$. This is not possible because a triangle has just one edge labeled as its  longest-edge. This contradicts our assumption, so $v$ has to be an endpoint of  at least one frontier- or barrier-edge in $R$.  Since triangles in $\Omega$ are distributed into terminal-edge regions without overlap~\cite{Ojeda2018ANA},  isolated points can not exist.
\end{proof}

\begin{figure}[h]
\centering
\includegraphics[width=0.25\linewidth]{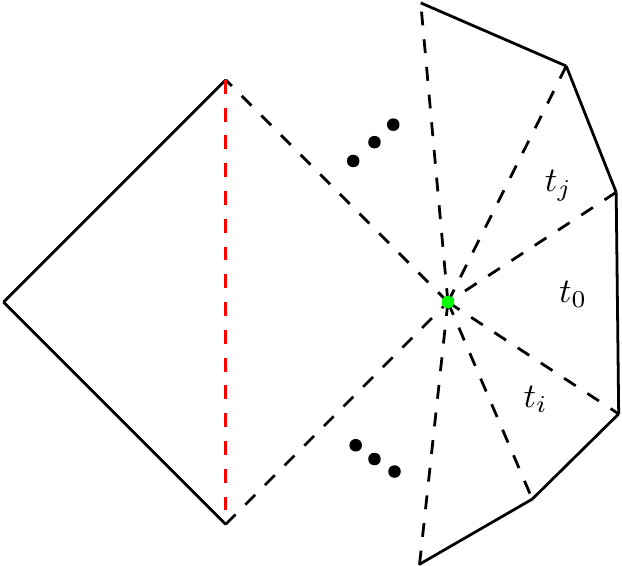} 
\caption{The  vertex in green is an interior point.}
\label{fig:lemmaquitadomingos}    
\end{figure}

It is worth mentioning Theorem \ref{D:theoreemvertices} guarantees that  the initial points used to represent the geometric domain and the ones inside the domain to fulfill point density requirements  are vertices delimiting one or more polygons. Moreover, each  internal-edge is a diagonal of one polygon. Therefore,  interior-edges that contain barrier-edge tips as endpoints  can be used to split a non-simple polygon into simple polygons.


Note that if the input points are not in general position,  a region without a terminal-edge might appear. This  degenerate region $R$ occur when the algorithm  labels all edges that share a vertex $v$ as internal-edges. This case might happen 
if all triangles in $R$ are equilateral and the algorithm labels each edge $e$ sharing $v$ as the longest-edge in only one of the triangles that share $e$. Then all edges that share $v$ are internal-edges. This rare kind of case can be solved  by  assigning one edge with  $v$ as endpoint as longest-edge in  the two triangles that share it. In this way, $R$ is now a terminal-edge region.

\section{The algorithm}
\label{sec:the_algorithm}


In this section, we describe  the main steps of the proposed algorithm, its computational cost  and the data structure used for its implementation.
The algorithm receives an initial triangulation as input, that can be generated by any known triangulator. The triangulation can be Delaunay or not, but we are using Delaunay triangulations because, several polygons of the mesh keep some angles of the triangles of the input triangulation. Since a Delaunay triangulation maximizes the minimum angle over all the possible triangulations for a point set, this angle will be a lower bound for the minimum angle of the generated polygonal mesh.

There are  several correct and robust triangulators such as Detri2~\cite{Detri2} and Triangle~\cite{triangle2d} available for free.  We are using Detri2 \cite{Detri2} to generate the constrained Delaunay triangulation needed as initial mesh. The whole process applied to the initial triangulation is divided in next three phases:

\begin{enumerate}[label=\roman*)]
    \item Label phase: Each edge is labeled as  longest-, terminal- or frontier-edges to  build terminal-edge regions. The algorithm also labels one  triangle in each terminal-edge region as seed triangle used in the next phase to build each polygon.
    \item Traversal phase: Generation of polygons from seed triangles. In this phase the frontier-edge vertices of a terminal-edge region are stored in counter-clock-wise (ccw) order to  build a polygon. Non-simple polygons generated in this phase are sent to the reparation phase.
    \item Reparation phase: Polygons with barrier-edges are partitioned into simple polygons.
\end{enumerate}

\subsection{Data structures}
\label{sec:datastructrue}
In this subsection we describe the data structures used in our implementation. We have decided to use an indexed data structure with adjacencies as  described in \cite{MeshDatastructSurvey}.  The purpose of the following representation is to have an easy and compact way to travel through all the faces of the triangulation and, in an implicit way, by using a special order in the representation, to access their edges and neighbor triangles in ccw. 

The triangulation is represented using three one-dimensional arrays for the vertices, triangles and neighbor triangles, respectively.  
Vertices are stored in pairs $(x,y)$, where each two consecutive elements of the Vertex array are the coordinates $x$ and $y$ of a point. The Triangle array is a set of indices to the Vertex array. Each 3 consecutive values is a triangle. Since the algorithm needs to know the neighbor through each triangle edge, the Neighbor array  stores the three indices of the neighbor triangles of  triangle $i$ at the locations $3i + 0$, $3i + 1$, $3i + 2$. 

To facilitate the implementation of the algorithm, vertex indices in the Triangle array are ordered  in such a manner that is easy to find out which is the neighbor triangle through  each triangle edge. Fig. \ref{figs:data_and_triangle}(a) illustrates the connections of each array element. The first two point indices $3i+0$, $3i+1$ in the Triangle array is an edge of triangle $i$ and this edge is shared with the triangle stored in the position $3i+2$ of the Neighbor array; the triangle edge defined by $3i+1$, $3i+2$ in the Triangle array is shared with the triangle stored  in $3i+0$ in the Neighbor array and the triangle edge $3i+2$, $3i+0$ in the Triangle array with the triangle stored at $3i+1$ in Neighbor array. Triangle edges are stored implicitly: the green edge in Fig. \ref{figs:data_and_triangle}(b) corresponds to the first one,  the  red  edge is  the second one and the blue  the third one.  These neighborhood relations can currently be generated as output by several tools such as  Qhull \cite{qhull}, Triangle \cite{triangle2d} and Detri2qt \cite{Detri2}.

\begin{figure}
\centering     
\subfigure[]{\label{fig:datastruc}\includegraphics[width=0.3\textwidth]{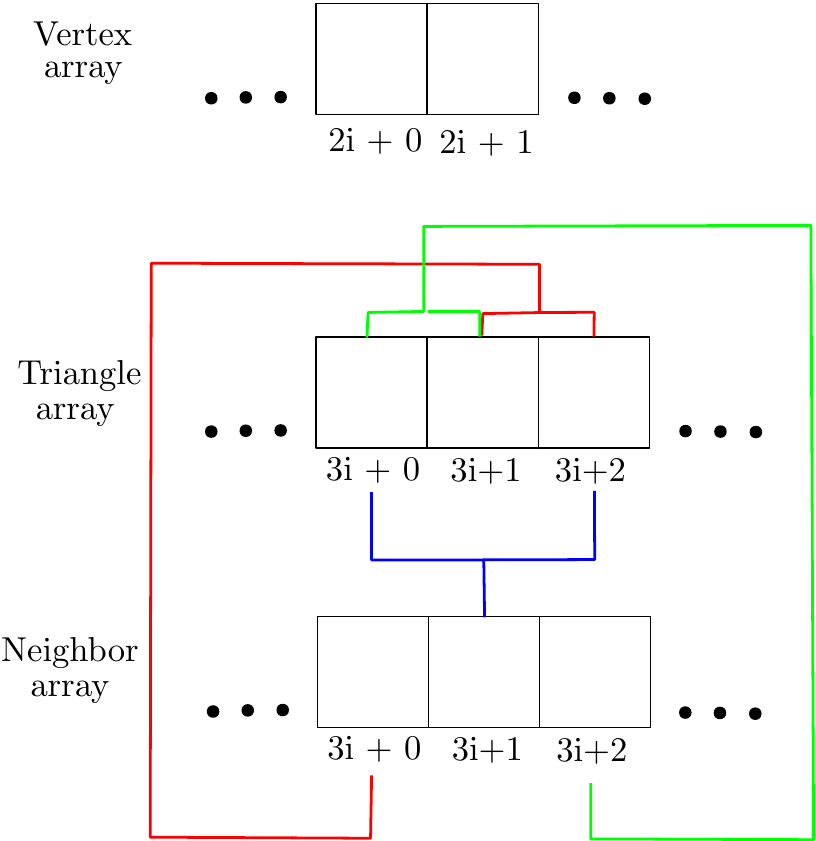}}\hspace{2cm}
\subfigure[]{\label{fig:datastructriangle}\includegraphics[width=0.3\textwidth]{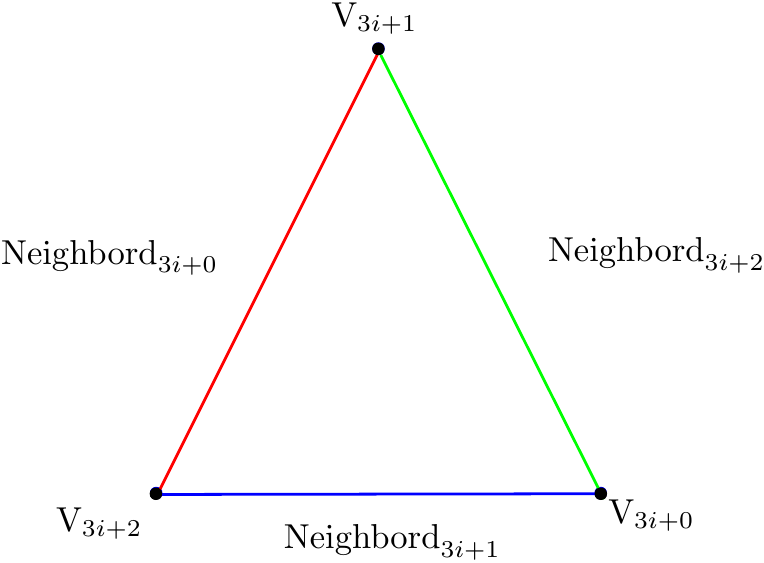}}
\caption{Data structure:  \textbf{(a)} Relation between the  vertex indices in the Triangle array and the triangle indices in the Neighbor array. \textbf{(b)} Triangle information stored  in the data structure shown in (a)
.}
\label{figs:data_and_triangle} 
\end{figure}


The final mesh composed of simple polygons is also stored in  one array, where each polygon includes first its number of vertices and then the vertex indices in ccw  as shown in Fig. \ref{fig:meshdatastruct}. 

Data structures to store additional information needed only in one of the steps of the algorithm are mentioned in the respective subsection. It is worth mentioning here the representation of non-simple polygons since this information is used in two phases. Fig. \ref{fig:poly_array_representation}  shows an example of a polygon with two barrier edge tips. This kind of polygon is stored in an index array  in the same way as a simple polygon in order ccw, and it is recognized as non-simple because there are repeated vertex indices. The same occurs with the seed list $L$, which is initialized in the Label phase with one triangle index per each terminal-edge region, and used in the Traversal phase to build polygons.  

\begin{figure}[]
    \centering
\includegraphics[width=0.6\textwidth]{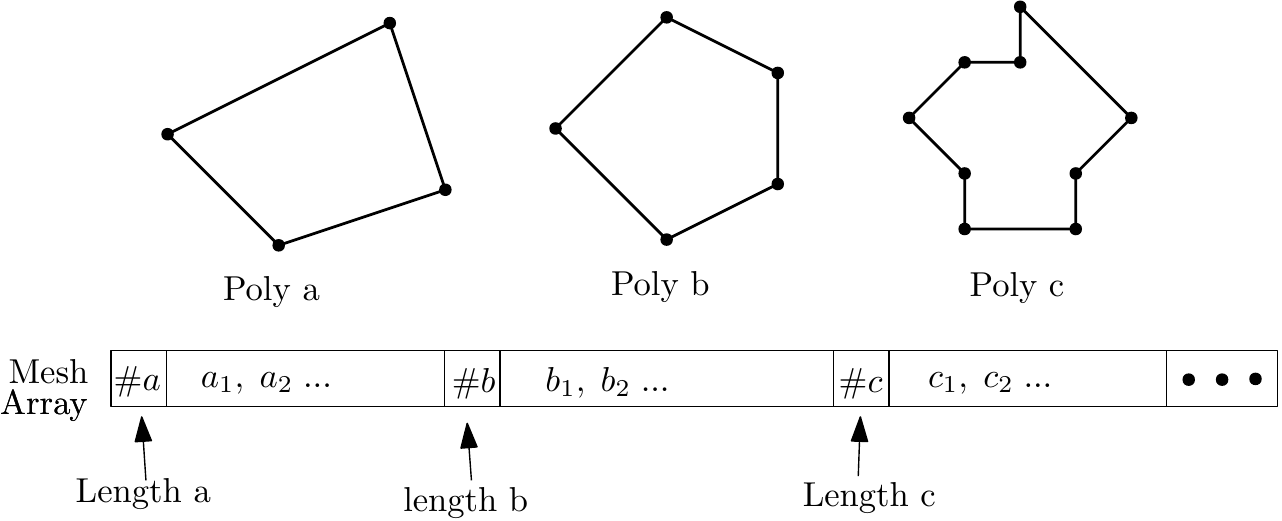}
    \caption{Mesh array with three polygons.} 
    \label{fig:meshdatastruct}
\end{figure}

\begin{figure}
\centering     
\includegraphics[width=0.4\textwidth]{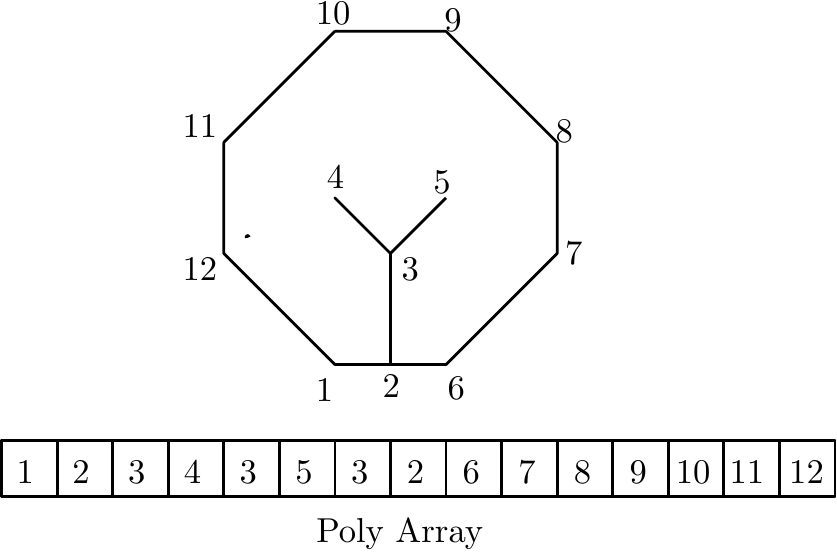} 
\caption{Representation of a  non-simple polygon as a set of vertex indices in ccw. Sequences $3 - 4 - 3$ and $3 - 5 - 3$ indicate the existence of barrier-edge tips.}
\label{fig:poly_array_representation} 
\end{figure}




\subsection{Label Phase}
\label{subsec:Labelpohase}

The goal of the Label phase  is to find  frontier-edges, terminal-edges and seed  triangles to build polygons in the next phase. The pseudo-code of this process is shown in Algorithm \ref{algo:labelphase}.
The algorithm first cycles over each triangle  of the initial triangulation (line \ref{algolabel:TriangleIteration}) to compute which edge is the longest-edge and stores this edge index (one per each triangle)  in a temporal array of size equal to the number of triangles; in the case of equilateral and isosceles triangles, the algorithm assigns  randomly a size order to the equal length edges  to avoid having a triangle that belongs to two terminal-edge regions at the same time. 


Afterward, the algorithm does a second iteration over the edges (line \ref{algolabel:Edgeiteration}). In case  an edge $e$ is not the longest-edge of any of the two triangles that share it (line \ref{algolabel:label}), $e$ is labeled as a frontier-edge. In case $e$ is a terminal-edge (line \ref{algolabel:seedtrianglelabel}), the algorithm stores the  smallest index of the triangles that share it in the Seed list $L$. In the case of a boundary terminal-edge, the index of the unique triangle is  stored. 

\begin{figure}[h]
\centering     
\subfigure[
]{\label{fig:delunolabel}\includegraphics[width=0.4\textwidth]{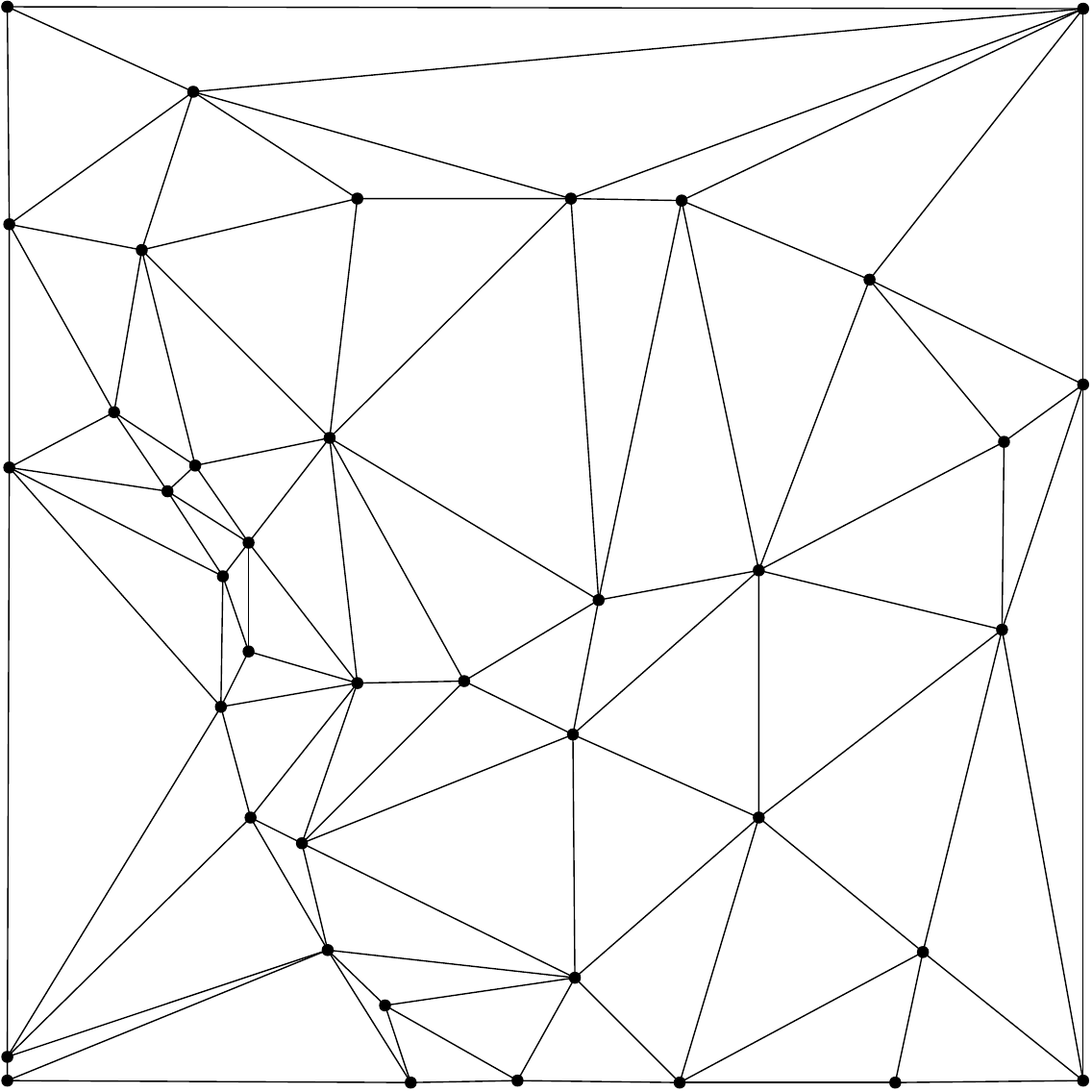}} \hspace{0.5cm} 
\subfigure[
]{\label{fig:Labelphaseafter}\includegraphics[width=0.4\textwidth]{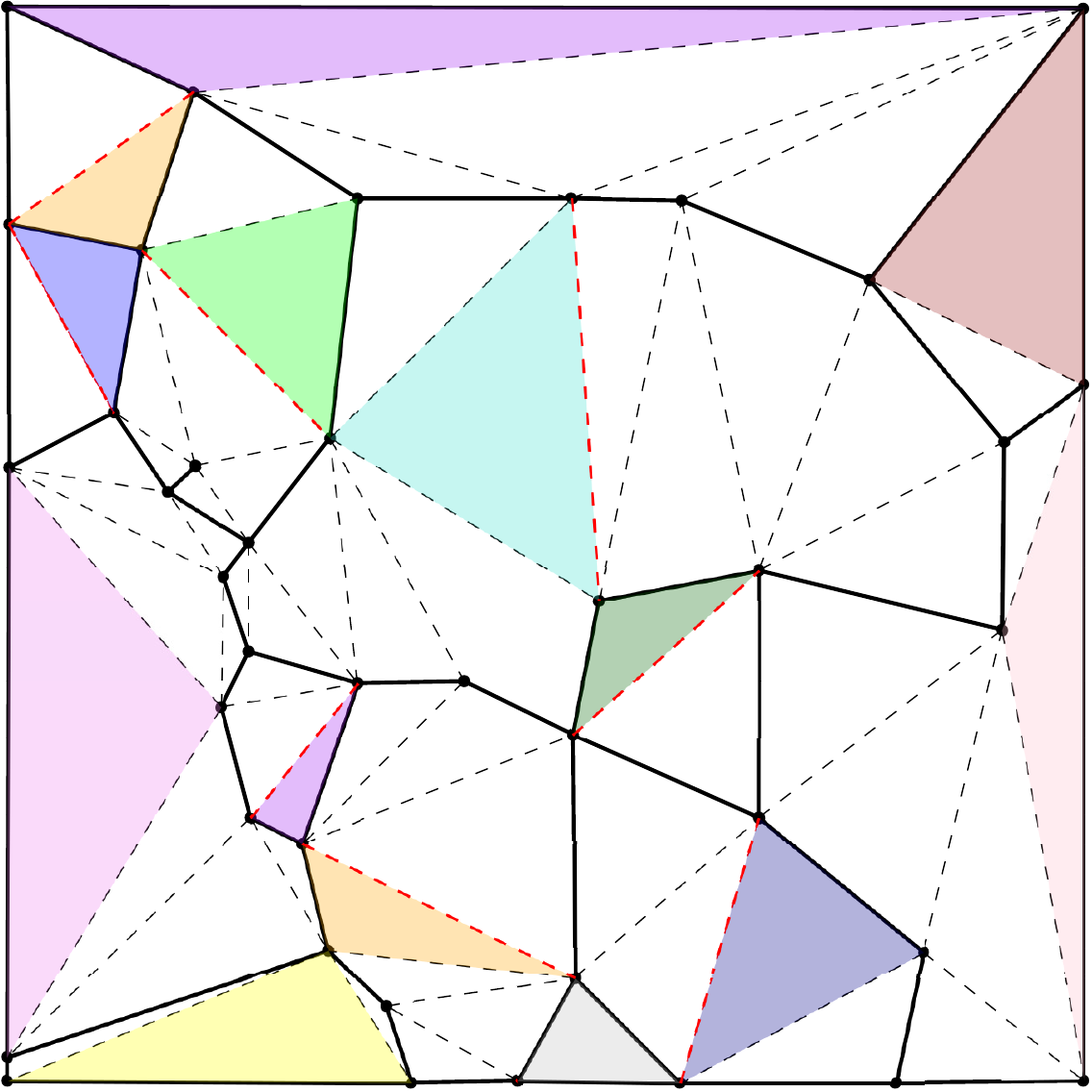}}
\caption{Label phase: (a) Input Delaunay triangulation. (b) Solid lines are the frontier-edges, dashed lines are the internal-edges and colorful triangles are seed triangles. }
\label{figs:label_phase} 
\end{figure}

The final result of this step is shown in Fig. \ref{figs:label_phase}(b). The colorful triangles are the seeds used in the next phase to generate the polygons. It can also be  observed that terminal-edge regions are  delimited by frontier-edges. 

\begin{algorithm}
    \caption{Label phase}\label{algo:labelphase}
    \begin{algorithmic}[1]
    \Require Initial Triangulation $\Omega$
    \Ensure Labeled edges and seed triangles of $\Omega$ 
    \ForAll{triangle $t_i$ in $\Omega$} \label{algolabel:TriangleIteration}
        \State Label the longest-edge of $t_i$
    \EndFor
    \ForAll{edge $e_i$ in $\Omega$} \label{algolabel:Edgeiteration}
        \State Be $t_1$ and $t_2$ triangles that share $e_i$
        \If{$e_i$ neither the longest-edge of $t_1$ nor $t_2$ or $e_i$ is border-edge} \label{algolabel:label}
            \State Label $e$ as frontier-edge
        \EndIf    
        \If{$e_i$ is terminal-edge or border terminal-edge} \label{algolabel:seedtrianglelabel}
            \State Choose one triangle sharing $e_i$ and store its index in the Seed list $L$
        \EndIf    
    \EndFor
    \end{algorithmic}
\end{algorithm}


\subsection{Traversal Phase}
\label{subsec:Travelpohase}

In this phase, polygons are built and represented as a closed polyline as shown in Fig. \ref{fig:poly_array_representation}. The main idea behind this phase is to travel through neighbor triangles, neighbors by  internal-edges inside a terminal-edge region and save their frontier-edges as edges of the new polygon $P$ in ccw. For this reason, the algorithm uses each triangle $t$ inserted in the Seed list $L$ as the starting triangle to build each polygon. The pseudo-code of this phase is shown in Algorithm \ref{algo:travelphase}. 

From line \ref{algotravel:init} until line \ref{algotravel:initend}, the algorithm gets the first triangle $t$ from the Seed list and  initializes the variables according to  the number of frontier-edges that $t$ has.  In case $t$ has 3 frontier-edge edges, $t$ is stored in the polygon $P$ as a whole polygon. In case $t$ has less than $3$ frontier-edges, their endpoints are saved as the first points of $P$ in ccw. In addition, the first vertex  of $P$ and $t$ are saved in $v_{init}$ and $t_{init}$, respectively, to check when the polygon boundary is ready. It is also necessary to store the last added vertex  in $P$ in $v_{end}$ to look for the e endpoint to be added. In case $t$ has no frontier-edge (Line \ref{algotravel:initial0edgescase}), the algorithm saves any vertex of $t$ as $v_{init}$ and $v_{end}$, and  $t$ in $t_{init}$. In the end, all vertices of $t$ will be part of polygon boundary; then it does not matter which one is taken as the initial point. After the initialization, the algorithm travels to the next internal-edge neighboring triangle $t'$ that shares $v_{end}$ with $t$. Three cases for $t'$ must be  considered:

\begin{enumerate}[label=\roman*)]
    \item Line \ref{algotravel:twoedgecase}: The triangle $t'$ has just  one frontier-edge $e$ and this edge contains $v_{end}$. Then $e$ is stored in $P$ and $v_{end}$ is updated with the other $e$ endpoint. The next triangle $t'$ is the non-visited internal-edge neighbor triangle that contains the new $v_{end}$. 
    
    \item Line \ref{algotravel:oneedgecase}: The triangle $t'$ is an ear triangle (a triangle with 2 frontier-edges). In this case, both edges are stored in ccw in $P$ and $v_{end}$ is updated with last stored endpoint. The previously visited triangle is the new $t'$ because the other two triangles are neighbor by frontier-edges (i.e., they belong to other terminal-edge regions).
    
    \item Line \ref{algotravel:noedgecase}: The triangle $t'$ has no frontier-edge. $t'$ shares an internal-edge with the last visited triangle and with the other two neighbor triangles that can be visited. Just one of these two triangles contains the endpoint $v_{end}$ as vertex, so that triangle is the new $t'$. 
\end{enumerate}


\begin{figure}[h]
    \centering
\includegraphics[width=0.5\textwidth]{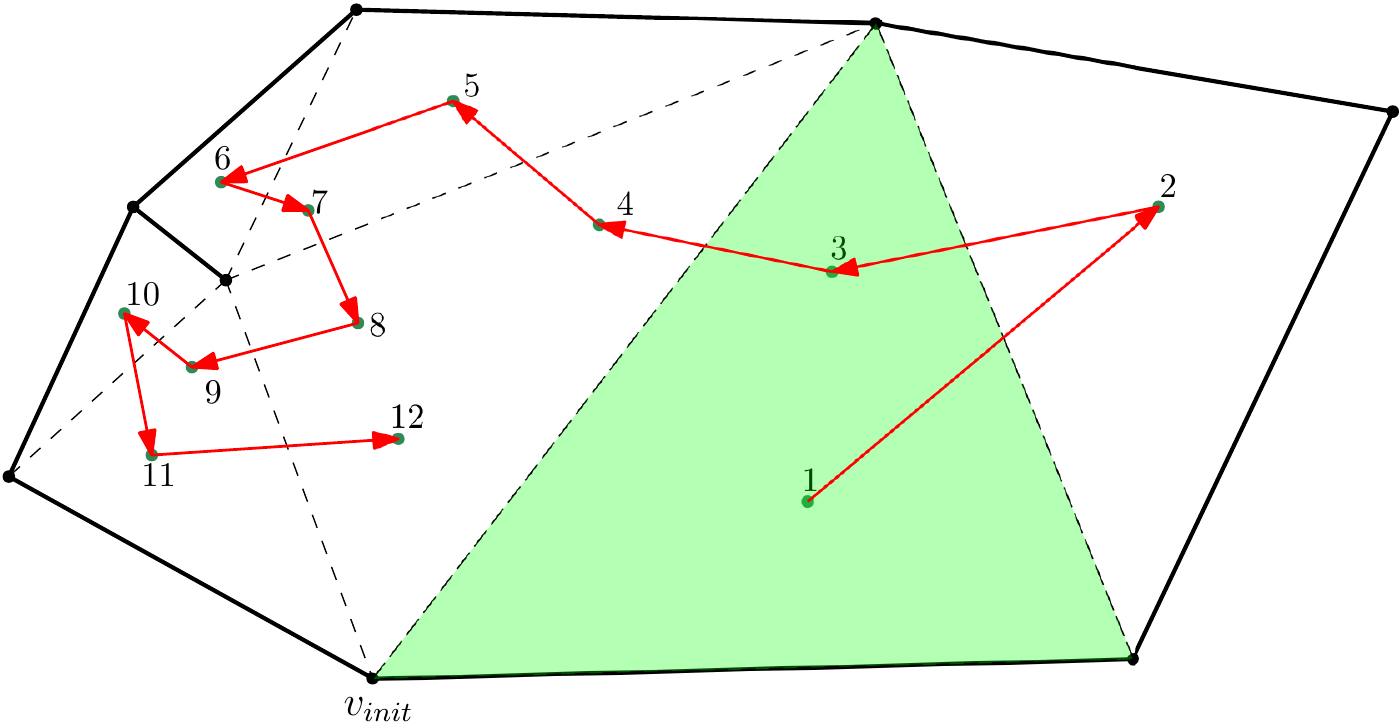}

    \caption{Traversal phase. The green triangle is the seed triangle of this terminal-edge region. The labels indicate in which order  each triangle was visited. } 
    \label{fig:travelphase}
\end{figure}

\begin{algorithm}
    \caption{Polygon construction}\label{algo:travelphase}
    \begin{algorithmic}[1]
    \Require Seed triangle $t$ of a terminal-edge region
    \Ensure Arbitrary shape polygon
    \State Polygon $\leftarrow$ $\emptyset$ \label{algotravel:init}
        \If{ $t$ has $3$ frontier-edges $e_1$, $e_2$ and $e_3$} \label{algotravel:initial3edgescase}
            \State \Return Polygon $\leftarrow$ $e_1 \cup	 e_2 \cup	e_3$
        \ElsIf{ $t$ has $2$ frontier-edges $e_1$ and $e_2$} \label{algotravel:initial2edgescase}
            \State $v_{init} \leftarrow e_1.v_{first}$,  $v_{end} \leftarrow e_2.v_{second}$
            \State Polygon $\leftarrow$ $e_1 \cup e_2$
        \ElsIf{ $t$ has $1$ frontier-edge $e_1$} \label{algotravel:initial1edgescase}
            \State $v_{init} \leftarrow e_1.v_{first}$, $v_{end} \leftarrow e_1.v_{second}$
            \State Polygon $\leftarrow$ $e_1$
        \ElsIf{$t$ has no frontier-edges} \label{algotravel:initial0edgescase}
            \State $v_{init}, v_{end} \leftarrow t.v_{0}$
        \EndIf    
    \State $t' \leftarrow$ internal-edge neighbor triangle of $t'$ that shares $v_{end}$  \label{algotravel:initend}
    \While {$v_{init} \not= v_{end}$ and $t' \not= t_{init}
    $} \label{algotravel:end}
        \If{ $t'$ has $2$ frontier-edges $e_1$ and $e_2$ with $e_1$ continuous to $v_{end}$} \label{algotravel:twoedgecase}
            \State Polygon $\leftarrow$ Polygon $\cup$ $e_1 \cup e_2$
            \State $v_{end} \leftarrow e_2.v_{second}$
            \State $t' \leftarrow$ last visited triangle
        \ElsIf{$t'$ has $1$ frontier-edge $e_1$ continuous to $v_{end}$} \label{algotravel:oneedgecase}
            \State $v_{end} \leftarrow e_1.v_{second}$
            \State Polygon $\leftarrow$ Polygon $\cup$ $e_1$
            \State $t' \leftarrow$ internal-edge neighbor triangle not visited in previous iteration that contains $v_{end}$
        \Else \label{algotravel:noedgecase}
            \State $t' \leftarrow$ internal-edge neighbor triangle in ccw that contains $v_{end}$
        \EndIf 
    \EndWhile
    \State \Return Polygon $P$
    \end{algorithmic}
\end{algorithm}

Algorithm~\ref{algo:travelphase} is applied to each seed triangle and returns a polygon $P$. Note that triangles with no frontier-edges are visited 3 times, with one frontier-edge two times and with two frontier-edges one time; this is demostrated in Lemma \ref{l:lemmatravelphase}. In case $P$ has barrier-edge tips, the algorithm described in Section~\ref{subsec:nonsimplereparation} is used to divide it into simple-polygons. A non-simple polygon can be easily detected by just checking the repetition of consecutive vertices in $P$: given three consecutive vertices of $P$, $v_i$, $v_j$ and $v_k$, if $v_i$ and $v_k$ are equal, then $v_j$ is a barrier-edge tip. In case $P$ does not have barrier-edge tips,  the polygon is saved as part of the mesh in the Mesh array.

\begin{lemma}{Each triangle is visited at most $3$ times during the polygon construction from a terminal-edge region.} \label{l:lemmatravelphase}
\end{lemma}

\begin{proof}: By construction,  the algorithm can only  visit triangles through their internal edges. The worst case  is when a triangle has 3 internal edges (no frontier-edges); then this triangle is visited once per internal edge, i.e. 3 times. In case of 2-internal edges, the triangle is visited 2 times and in case of 1 internal-edge just 1 time.


\end{proof}

\subsection{Non-simple polygon reparation}
\label{subsec:nonsimplereparation}

As previously mentioned, some polygons generated in the Traversal phase are non-simple. In this section, we describe an algorithm to divide non-simple polygons into simple-ones by using internal-edges containing barrier-edge tips as one of their end-points. 
Notice that a simple strategy could just be to remove barrier-edges and keep as polygon boundary only the frontier-edges. However, this approach would require to delete barrier-edge endpoints, which are vertices defined/inserted in  the initial triangulation by the user/triangulation tool. We assume that all points are important. That is why, we propose  a method to repair non-simple  polygons without removing initial vertices.

The reparation process works similarly to the Label and Traversal phases but inside a non-simple polygon $P$. In short, the algorithm labels so many internal edges as frontier-edges until all barrier-edge tips are part of a polygon boundary and stores one seed triangle per each new polygon in a new Seed list $L_p$. After, those triangles are used as seed to generate the new polygons with the Algorithm \ref{algo:travelphase}. The internal-edges chosen to be new frontier edges are the ones that contribute to generate polygons of similar size.

The reparation algorithm is shown in Algorithm \ref{algo:reparationphase}. For each barrier-edge tip $b_i$ in a polygon $P$ (line \ref{algorepa:foreachbet}), the algorithm computes first $deg(b_i)$, which is the number of incident internal-edges to $b_i$ inside the terminal-edge region. To facilitate this calculation the algorithm uses an auxiliary array that associates to each vertex a triangle. If $deg(b_i)$ is odd then the algorithm labels the middle edge as frontier-edge; otherwise, the algorithm chooses any of the two middle edges as a new frontier-edge. In both cases, the two triangles that share the new frontier-edge are saved in the seed list $L_p$. From line~\ref{algorepa:foreachseedtriangle} to line ~\ref{algorepa:removeseeds}, each new simple polygon is built. A Bit array $A$, of size equal to the number of triangles of the input triangulation is used to avoid the generation of the same polygon more than once. The insertion of one or more seed triangles of the same new polygon in the seed list $L_p$ occurs when a terminal-edge region contains more than one barrier-edge tip. Each seed triangle stored in $L_p$ is set to 1 in $A$; the others are 0. When a polygon is built, each used triangle is set to $0$ in $A$. Fig. \ref{fig:betosplit} shows an example of a non-simple polygon with three barrier-edge tips (vertices in green color). Fig. \ref{fig:besplited} shows the internal-edges converted in new frontier-edges to delimit the new simple polygons and the six seed triangles stored in $L_p$.  Fig. \ref{fig:benewpolygons} shows the four new polygons after the split and the seed triangles used to generate them.

It is worth mentioning that the number of polygons generated after the split is at  most $(\lvert B\rvert + 1)$, with $\lvert B\rvert $ the number of barrier-edge tips in a polygon $P$.

\begin{lemma} \label{lemma:degReparation}
Let   $\Omega$ be the input triangulation and $B$ the set of barrier edge tips in the polygon $P$. In the computation of $deg(b_i)$,   $\forall b_i, b_i \in B$, each triangle $t$ inside the terminal-edge region that generated $P$ can be visited at most 3 times.
\end{lemma}
\begin{proof}: By construction  
 if a triangle $t$ is not incident to a barrier-edge tip $b_i$, Algorithm~\ref{algo:reparationphase} does not  visit $t$ during the computation of  $deg(b_i)$. If $t$ contains as vertex one, two or three barrier-edge tips, then $t$ is visited one time per each barrier-edge tip. Since $t$ has at most   $3$ barrier-edge tips then $t$ is visited in the worst case $3$ times for the computation of $deg(b_i)$.
\end{proof}

\begin{algorithm}
    \caption{Non-simple polygon reparation}\label{algo:reparationphase}
    \begin{algorithmic}[1]
    \Require Non-simple polygon P
    \Ensure Set of simple polygons $S$
    \State Initialize seed list $L_p$ and bit array $A$ \label{algorepa:inithash}
    \State $B \leftarrow$ set of barrier-edge tips 
    \State $S$ $\leftarrow$ $\emptyset$ 
    \ForAll{barrier-edge tip $b_i$ in $B$} \label{algorepa:foreachbet}
        \State Calculate $deg(b_i)$, the number of incident internal-edges  to $b_i$
        \State Label middle internal-edge $e$ incident  to $b_i$ as frontier-edge
        \State Save triangles $t_1$ and $t_2$ that share $e$ in $L_p$
        \State $A[t_1] \leftarrow $ True, $A[t_2] \leftarrow $ True  
    \EndFor
    \ForAll{seed triangle $t_i$ in $L_P$} \label{algorepa:foreachseedtriangle}
        \If{$A[t_i]$ is \texttt{True}} \label{algorepa:true}
            \State $A[t_i] \leftarrow False$
            \State Generate new polygon $P'$ starting from $t_i$  \label{algorepa:generationpoly} using Algorithm \ref{algo:travelphase}. 
        \State Set as \texttt{False} all indices of triangles in $A$ used to generate $P'$ \label{algorepa:removeseeds}
        \EndIf  
        
    \State $S$ $\leftarrow$ $S$ $\cup$ $P'$ 
    \EndFor 
    \State \Return $S$
    \end{algorithmic}
\end{algorithm}

\begin{figure}[h]
\centering  
\resizebox{.8\linewidth}{!}{
\subfigure[] {\label{fig:betosplit}\includegraphics[width=0.3\textwidth, ]{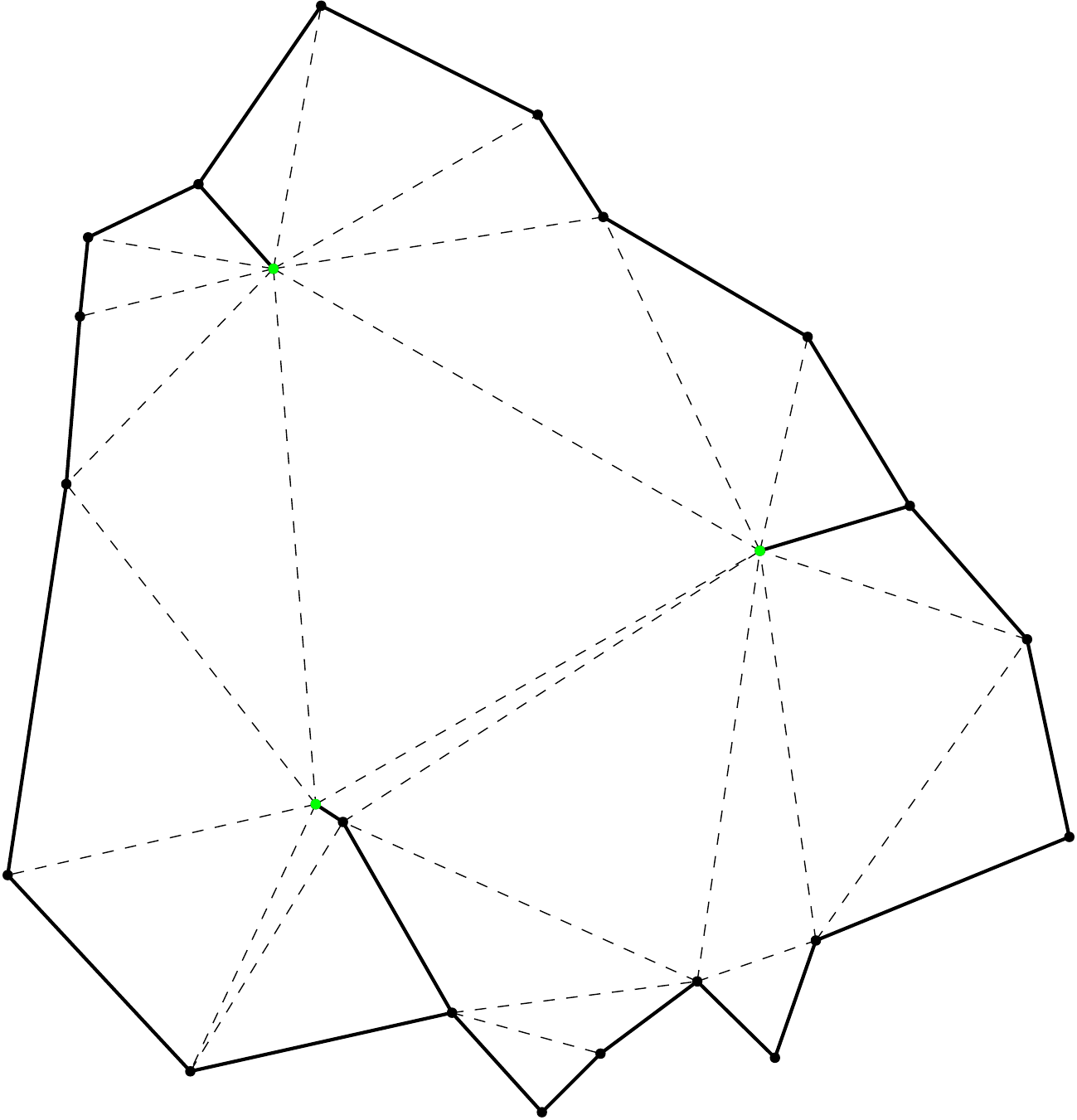}}\hspace{0.5cm}
\subfigure[]{\label{fig:besplited}\includegraphics[width=0.3\textwidth]{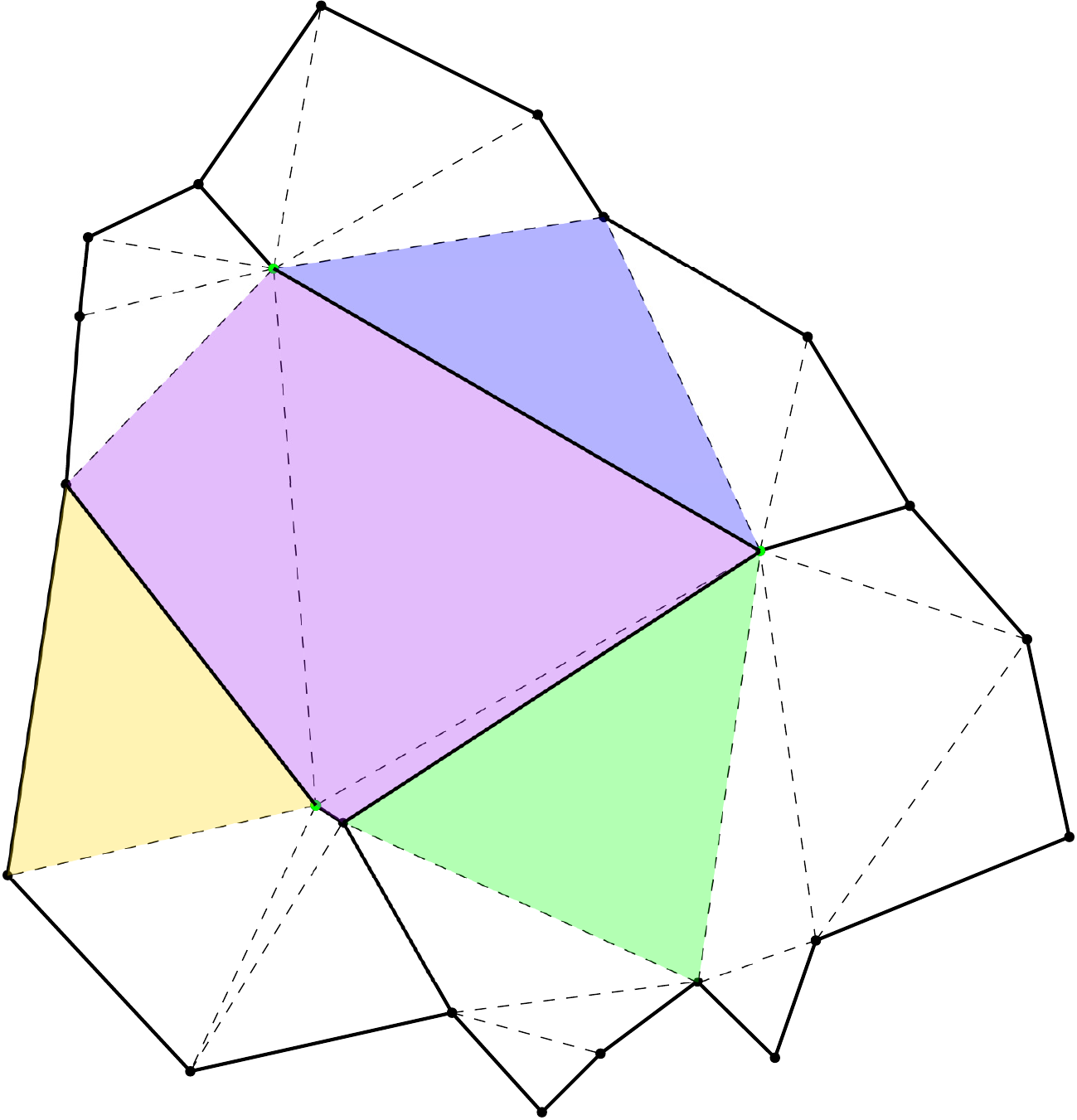}}\hspace{0.5cm}
\subfigure[]{\label{fig:benewpolygons}\includegraphics[width=0.3\textwidth]{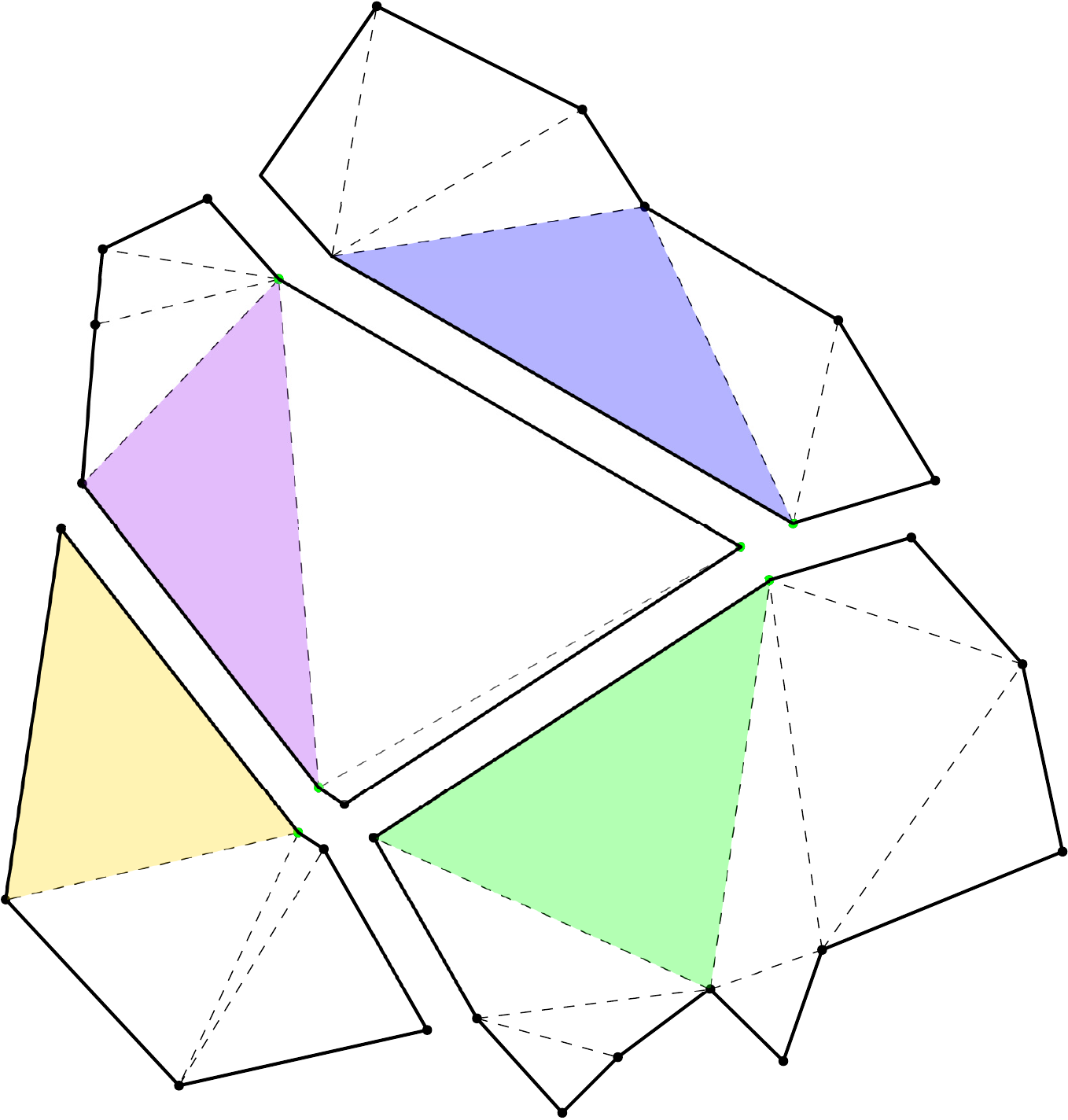}}
}
\caption{Example of a non-simple polygon split using barrier-edge tips. (a) Non-simple polygon. (b) Middle edges incident to barrier-edge tips labeled  as frontier-edges (solid lines) and seed triangles (colorful triangles) stored in the seed list $L_p$ and marked as 1 in the $A$ bit array. (c) The four new polygons without barrier-edge tips.  }
\label{figs:splitmid} 
\end{figure}

\subsection{Computational Complexity Analysis}
\label{sub:Complexity Analysis}

In this section, we analyze the computational complexity of  the whole algorithm. Let  $n$ be the initial number of points, $m$ the number of triangles of the input triangulation, $m_i$ the number of triangles of the terminal-edge region $R_i$ used to generate the simple or non-simple  polygon $P_i$

\begin{enumerate}
    \item \textbf{Label Phase.} This phase uses three iterations of cost $O(m)$, one to label longest-edges, one to identify  terminal-edges and another to store seed triangles.
    \item \textbf{Traversal Phase.} This phase calls Algorithm \ref{algo:travelphase} to build one polygon simple or not simple. The construction of polygon $P_i$ has cost $O(m_i)$. Each triangle is visited  three times by Lemma \ref{l:lemmatravelphase} in the worst case (when a triangle does not have frontier-edges). As each terminal-edge region covers the whole domain without overlapping, then this phase  costs $O(m)$.

    \item {\bf Non-simple Polygon Reparation phase}. This phase calls Algorithm~\ref{algo:reparationphase} to divide one non-simple polygon $P_i$ into simple ones. 
    By Lemma \ref{lemma:degReparation},  each triangle $t \in R_i$ is visited at most $3$ times, so the computation of  $deg(b)$, $\forall b \in R_i$ is  $O(m_i)$. The partition of $P_i$ into simple polygons has then a cost $O(m_i)$ because the new simple polygons  do not overlap and the algorithm only travels inside their triangles. Finally, the cost to repair all non-simple polygons is bounded $O(m)$.


\end{enumerate}

\noindent
\noindent
As result, the computational complexity of the whole algorithm is  $O(n)$ ($m = O(n)$).
With respect to the memory usage, the cost is $O(n)$ too. The Vertex array has size $2n$, the Triangle and Neighbor arrays have size $3n$. The number of final polygons in the mesh is less than the number of triangles in the input triangulation, i.e., bounded by $m$.

\subsection{Impact of the initial triangulation}

As mentioned above, the polygons inside a Polylla mesh depend on the initial triangulation. Any triangulation can be used as input.  Fig~\ref{figs:diff_triangulation} shows three different triangulations for the same point distribution generated using  algorithms available in the {\em 2D  Triangulations} package of the CGAL library~\cite{cgal:y-t2-21b}. Fig.~\ref{fig:DIFFrandomPolylla} is a Polylla mesh built from the triangulation generated by the incremental algorithm 
without  edge-flips (Fig~\ref{fig:DIFFrandomTriangulation}). The Polylla mesh in Fig~\ref{fig:DIFFregularPolylla} was built from a regular triangulation~\cite{BOISSONNAT20025} (a weighted triangulation) with random weights shown in Fig~\ref{fig:DIFFregularTriangulation}, and Fig.~\ref{fig:DIFFdelaunayPolylla} shows a Polylla mesh generated from the Delaunay triangulation drawn in Fig~\ref{fig:DIFFrandomdelaunay}. As expected, the shape of the resulting polygons is notoriously different. In order to see the differences, Table \ref{table:qualitycomp} shows for each mesh, the minimum and maximum interior angle of both the input triangulation and the generated polygon meshes, the number of polygons and the average number of edges per polygon. It can be observed that polygon meshes have a minimum interior angle greater than the minimum angle of the corresponding triangulation. Since the Polylla algorithm joins triangles to generate  polygons, the lower bound for the minimum interior angle of any polygon is the minimum angle of the input triangulation. In the case of maximum angles, the maximum angle of a Polylla mesh is usually larger than the maximum angle of the input triangulation because a Polylla's polygons can be nonconvex polygons. It seems that  Polylla meshes generated from  Delaunay triangulations tends to contain less polygons and  polygons defined with more edges than the ones obtained from other triangulations. Further research is necessary to probe these findings.

\begin{figure}[!h]
\centering     
\subfigure[
]{\label{fig:DIFFrandomTriangulation}\includegraphics[width=0.25\textwidth]{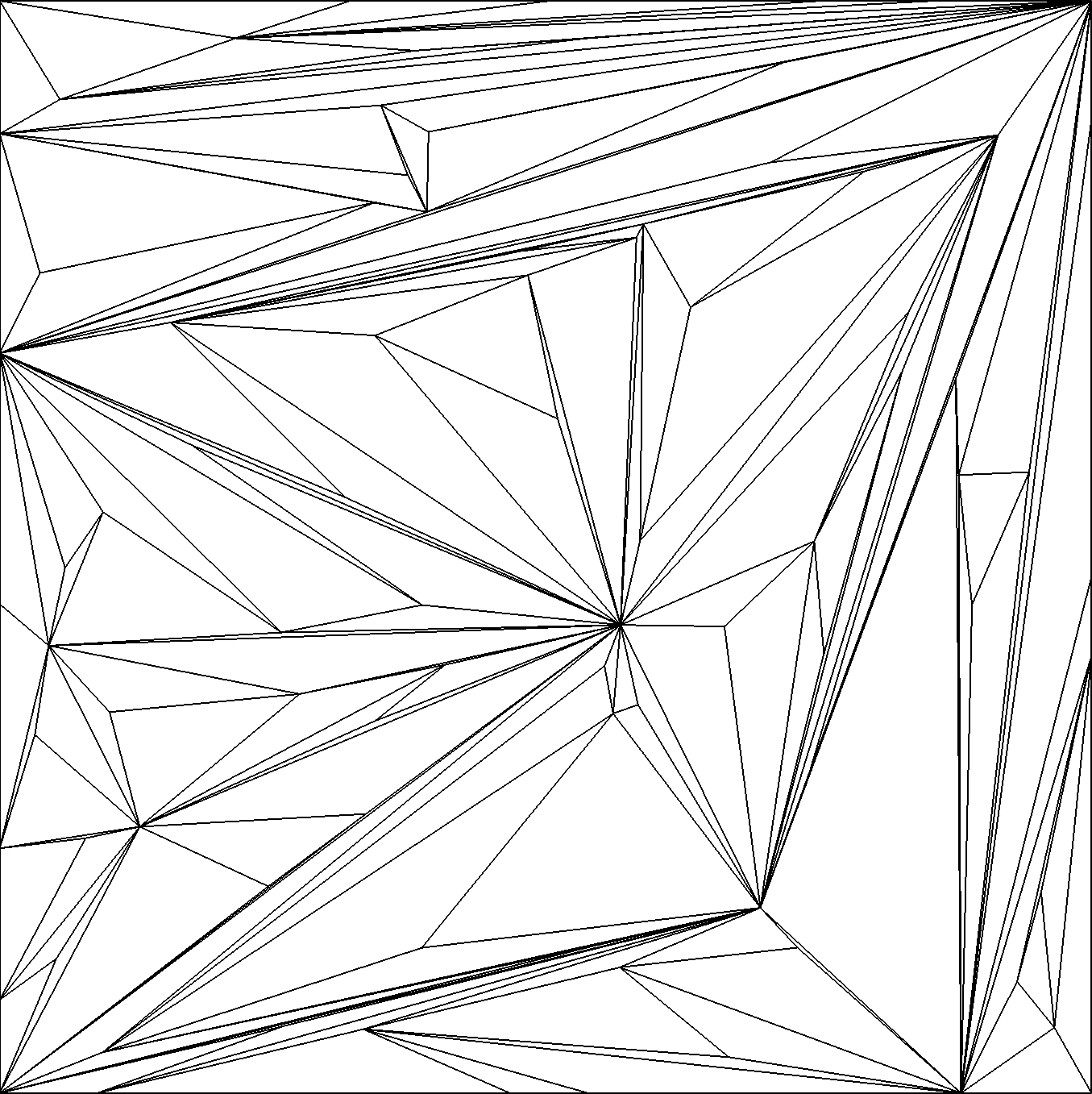}} 
\subfigure[
]{\label{fig:DIFFregularTriangulation}\includegraphics[width=0.25\textwidth]{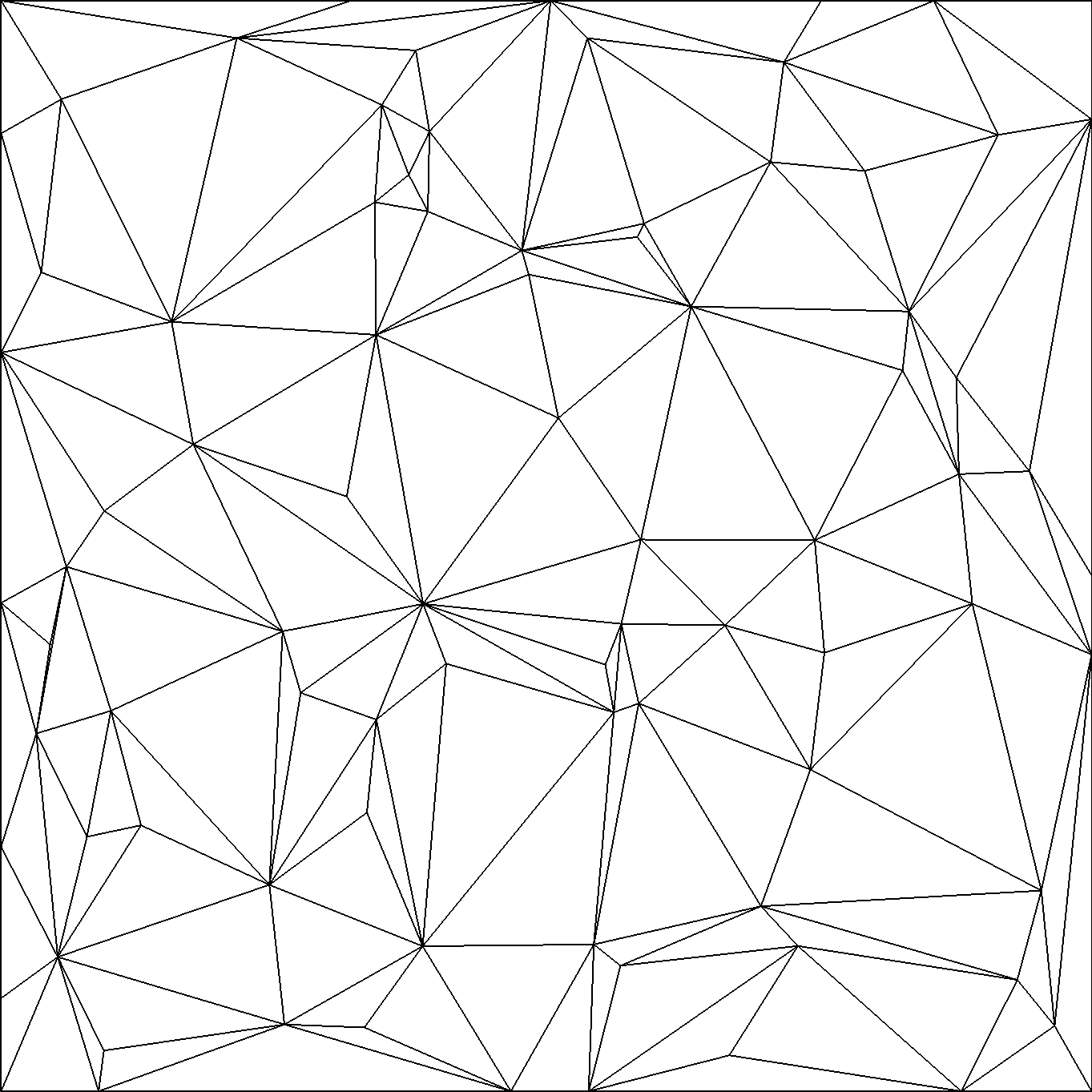}} 
\subfigure[
]{\label{fig:DIFFrandomdelaunay}\includegraphics[width=0.25\textwidth]{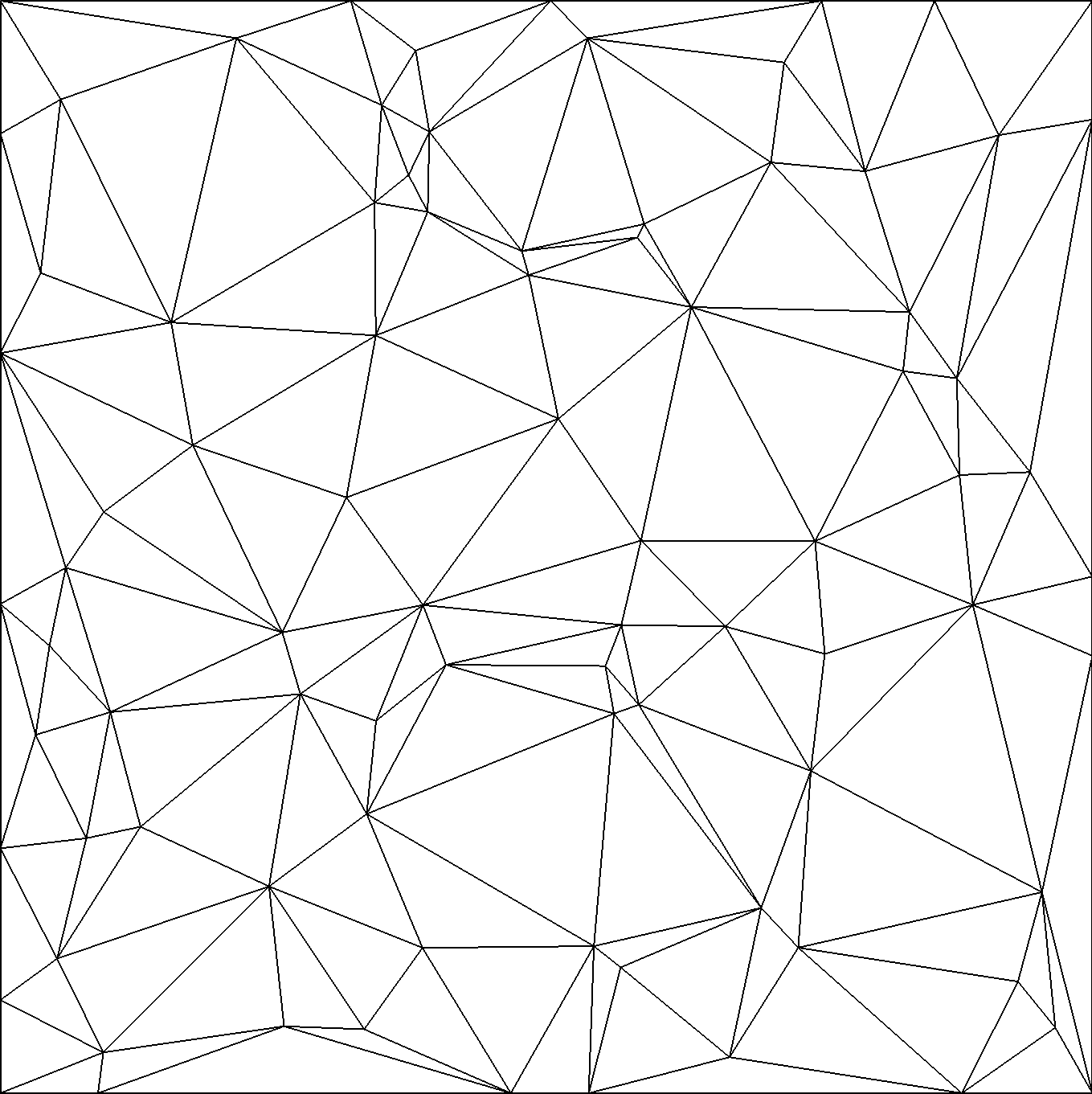}} 
\subfigure[
]{\label{fig:DIFFrandomPolylla}\includegraphics[width=0.25\textwidth]{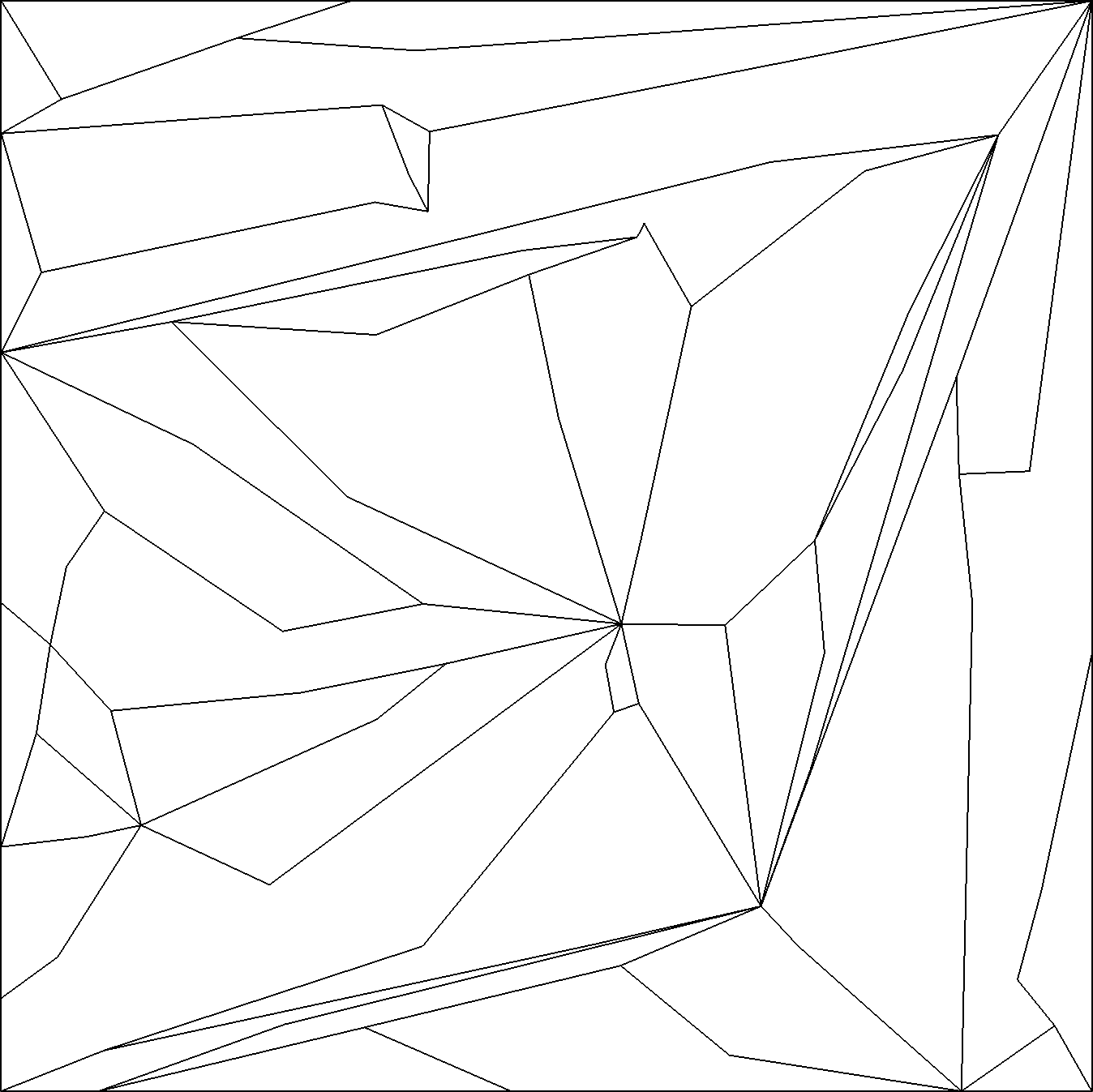}} 
\subfigure[
]{\label{fig:DIFFregularPolylla}\includegraphics[width=0.25\textwidth]{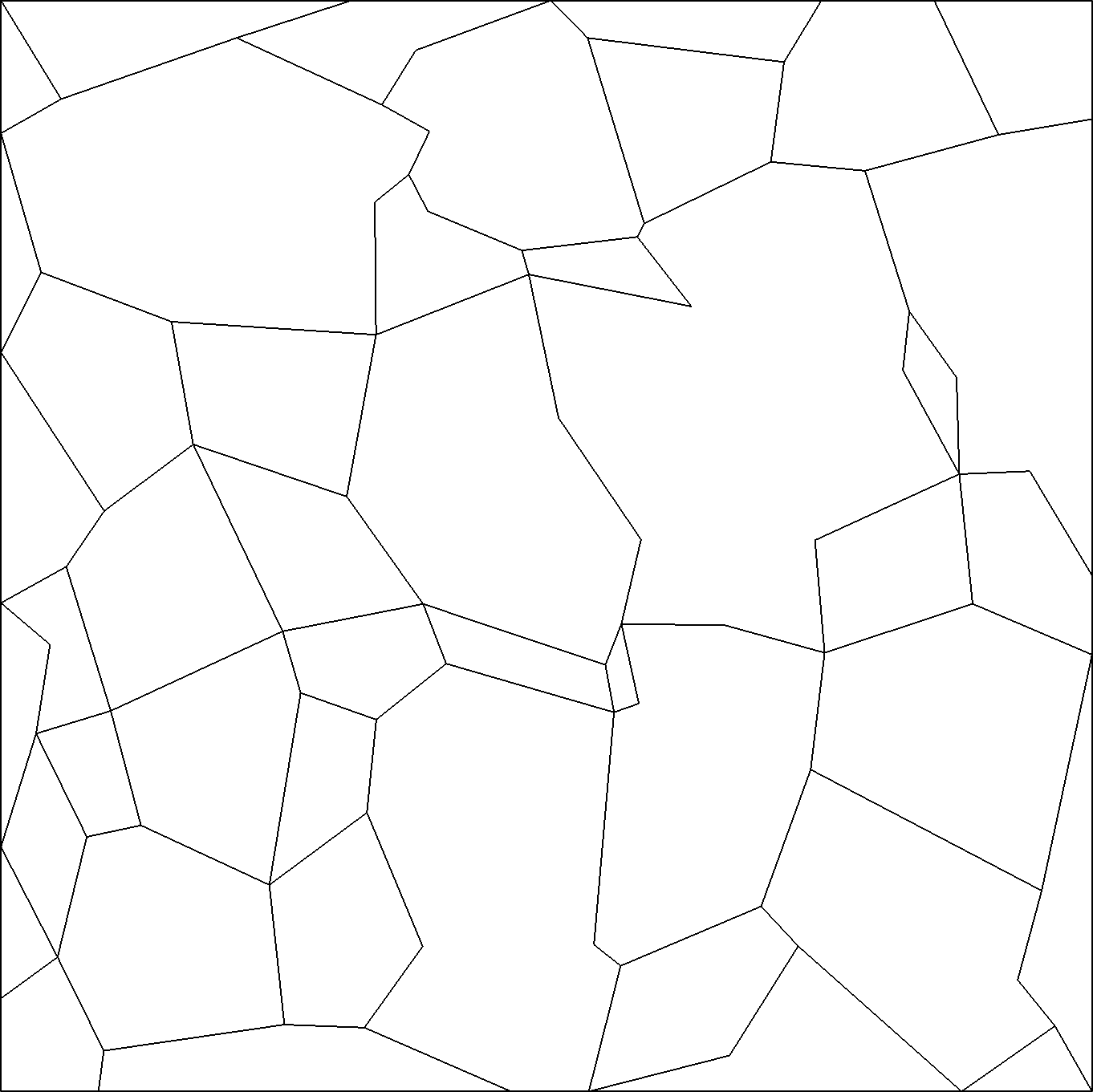}} 
\subfigure[
]{\label{fig:DIFFdelaunayPolylla}\includegraphics[width=0.26\textwidth]{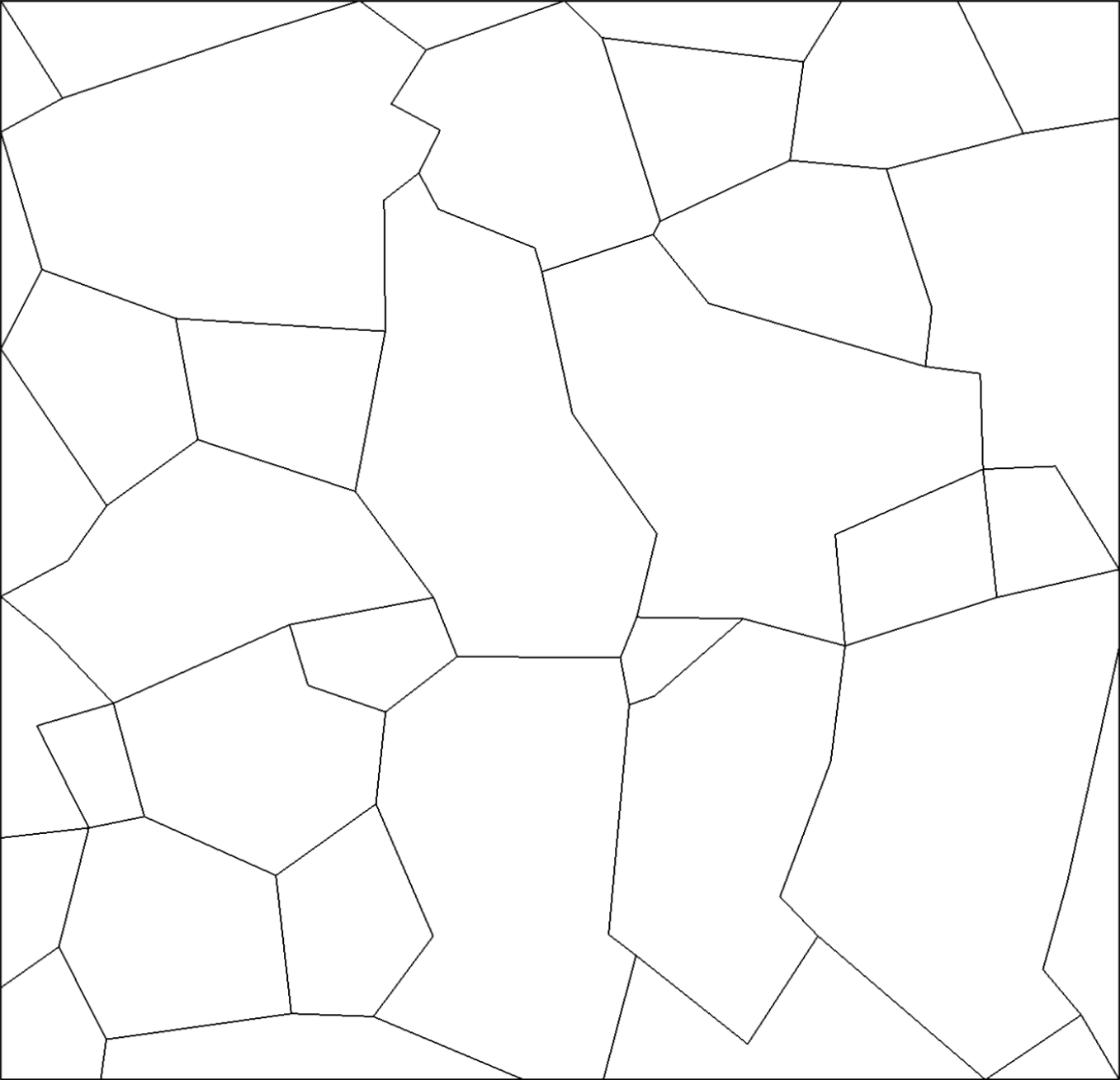}} 
    \caption{Polylla meshes generated from different kind of initial triangulations over the same input points. Triangulations were generated using the {\em 2D  Triangulations}  package of CGAL library~\cite{cgal:y-t2-21b}. (a) Triangulation generated by the incremental algorithm, (b) regular triangulation generated using random weights, (c) Delaunay triangulation, (d) Polylla mesh generated from the triangulation (a), (e) Polylla mesh generated from the triangulation (b), and (f) Polylla mesh generated from the triangulation (c).}

\label{figs:diff_triangulation} 
\end{figure}

\begin{table}[!h]
\centering
\begin{tabular}{l|rr|rrrr|}
\cmidrule{2-7} \addlinespace[-.6em]
& \multicolumn{2}{c|}{\begin{tabular}[c]{@{}c@{}}Delaunay\\ Triangulation\end{tabular}} & \multicolumn{4}{c|}{Polylla Mesh}\\ \vspace{-1.5em}\\  \cmidrule{2-7}\addlinespace[-0.6em] & \multicolumn{1}{c|}{\begin{tabular}[c]{@{}c@{}}Min\\ Angle\end{tabular}} & \multicolumn{1}{c|}{\begin{tabular}[c]{@{}c@{}}Max\\ Angle\end{tabular}} & \multicolumn{1}{c|}{\begin{tabular}[c]{@{}c@{}}Polylla's\\ Polygons\end{tabular}} & \multicolumn{1}{c|}{\begin{tabular}[c]{@{}c@{}}Min\\ Angle\end{tabular}} & \multicolumn{1}{c|}{\begin{tabular}[c]{@{}c@{}}Max\\ Angle\end{tabular}} & \multicolumn{1}{c|}{\begin{tabular}[c]{@{}c@{}}Edges per\\  Polygon\end{tabular}} \\ \hline
\multicolumn{1}{|l|}{\begin{tabular}[c]{@{}l@{}}Incremental\\ Triangulation\end{tabular}} & \multicolumn{1}{r|}{0.027}                                               & 179.87                                                                   & \multicolumn{1}{r|}{41}                                                           & \multicolumn{1}{r|}{0.36}                                                & \multicolumn{1}{r|}{303.16}                                              & 5.65                                                                              \\ \hline
\multicolumn{1}{|l|}{\begin{tabular}[c]{@{}l@{}}Regular\\ Triangulation\end{tabular}}     & \multicolumn{1}{r|}{1.36}                                                & 177.09                                                                   & \multicolumn{1}{r|}{45}                                                           & \multicolumn{1}{r|}{12.04}                                               & \multicolumn{1}{r|}{318.98}                                              & 5.33                                                                              \\ \hline
\multicolumn{1}{|l|}{\begin{tabular}[c]{@{}l@{}}Delaunay\\ Triangulation\end{tabular}}    & \multicolumn{1}{r|}{5.75}                                                & 158.61                                                                   & \multicolumn{1}{r|}{36}                                                           & \multicolumn{1}{r|}{12.04}                                               & \multicolumn{1}{r|}{279.22}                                              & 6.16                                                                              \\ \hline
\end{tabular}
\caption{Geometric information of  the Polylla meshes generated from different triangulations from  the same point set (150 points).}
\label{table:qualitycomp}
\end{table}

It is worth to mention that since the Polylla algorithm takes as input  a triangulation, it can process  any geometry domain that can be triangulated. This includes complex geometries with holes as  shown in Fig~\ref{figs:facePSLG}. In particular, Fig.~\ref{fig:PSLGface} shows a geometry  specified as a planar straight line graph (PLSG) obtained from~\cite{triangle2d}. Fig.~\ref{fig:PSLGface} is a Polylla mesh generated from a constrained Delaunay triangulation of \ref{fig:PSLGface26}. Fig.~\ref{fig:PSLGface220} shows a Polylla mesh from a refined  conforming Delaunay triangulation of~\ref{fig:PSLGface26}. Grey polygons are holes. The algorithm considers the edges of the holes  as border edges. Meshes generated from complex geometries representing   Chilean geographic regions can be seen in Appendix~\ref{Appendix:complexgeometries}.

\begin{figure}[!h]
\centering     
\subfigure[
]{\label{fig:PSLGface}\includegraphics[width=0.3\textwidth]{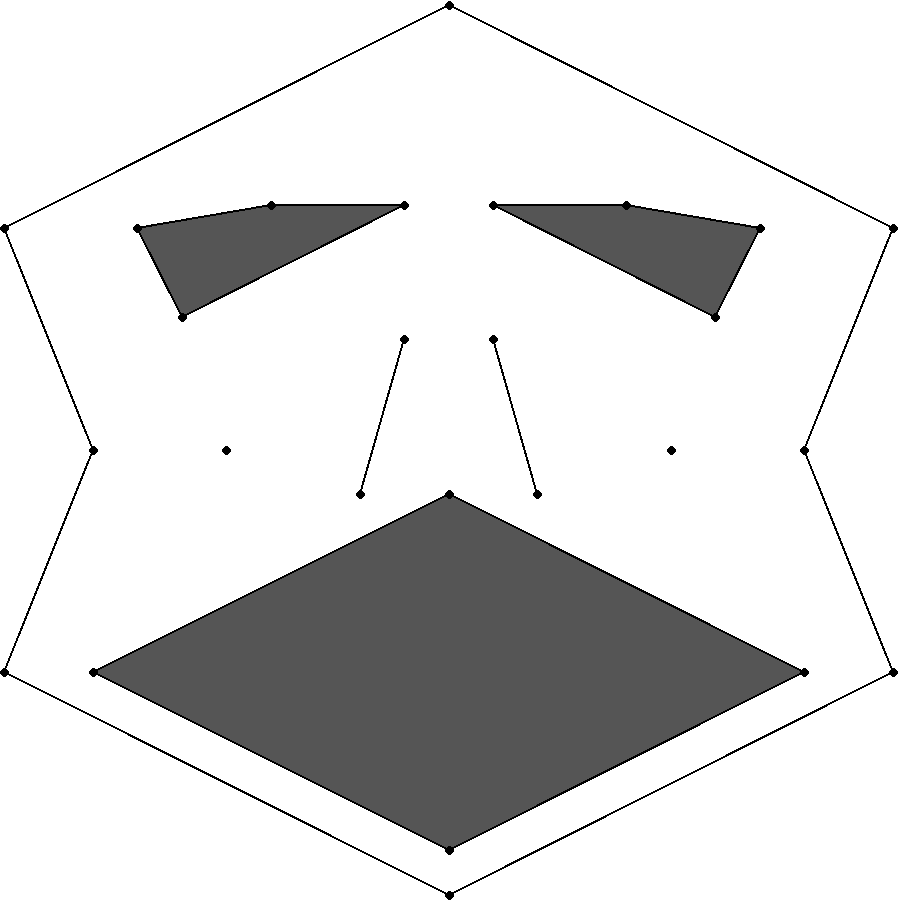}} \hspace{0.1cm}
\subfigure[
]{\label{fig:PSLGface26}\includegraphics[width=0.3\textwidth]{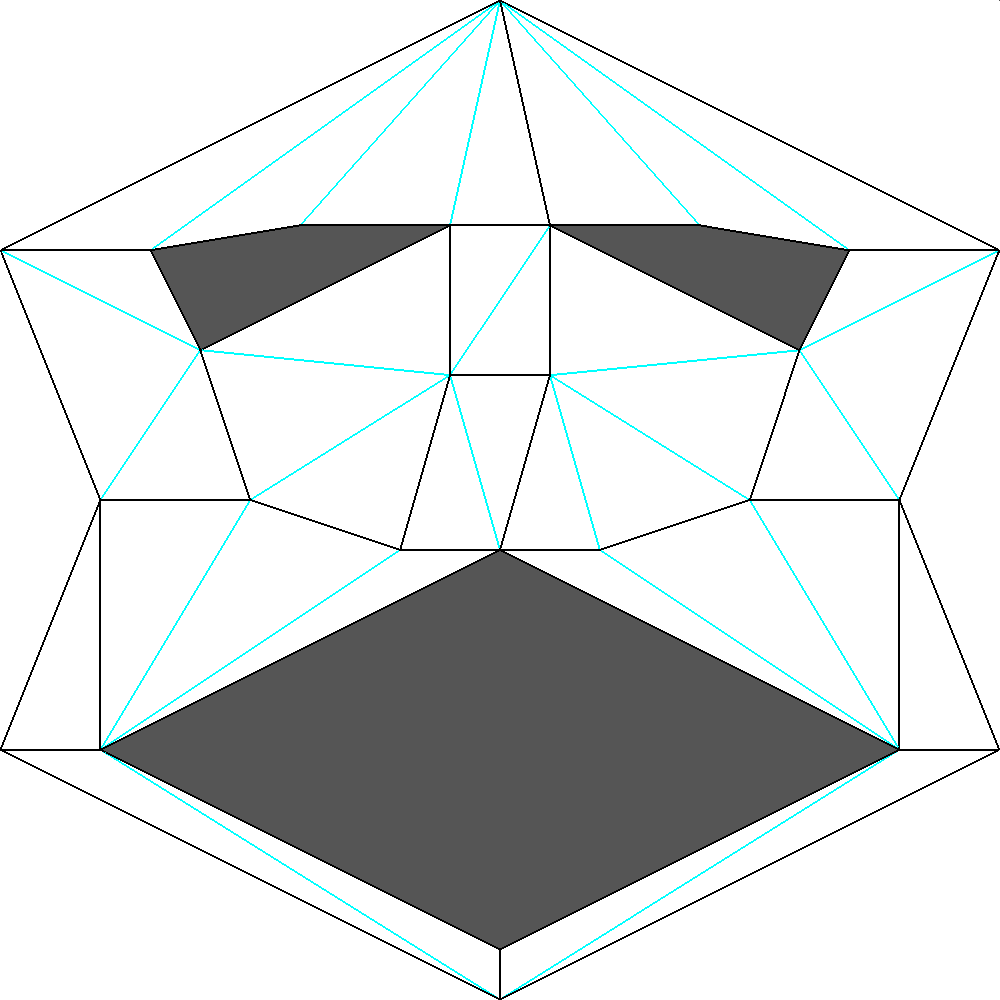}}\hspace{0.1cm}
\subfigure[
]{\label{fig:PSLGface220}\includegraphics[width=0.3\textwidth]{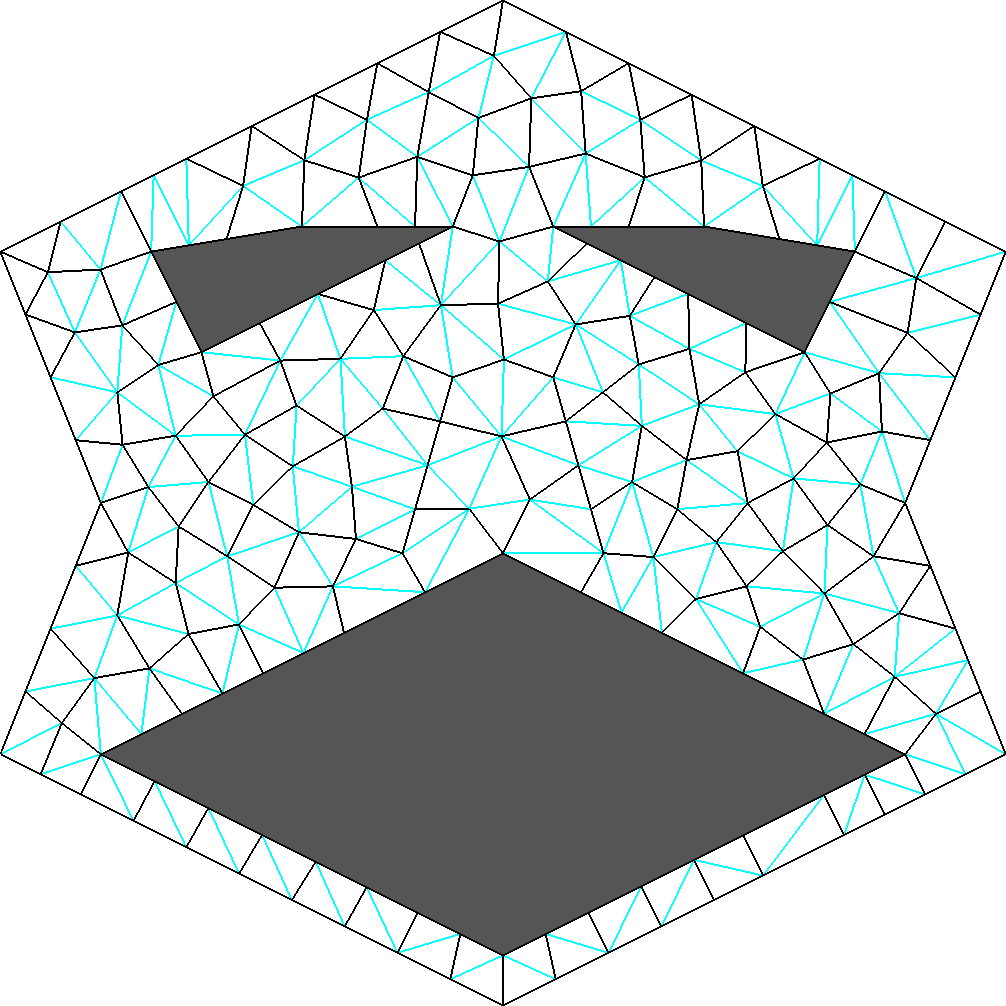}}
\caption{ Polylla mesh generated from the PLSG description of the Face  taken from~\cite{triangle2d}.  The edges of the initial triangulation are drawn using black and cyan colors.
(b) Polygonal  mesh  obtained from the constrained Delaunay triangulation of the input geometry (26 vertices). (c) Polygonal mesh  generated from  a  refined Delaunay triangulation with 220 vertices.}
\label{figs:facePSLG} 
\end{figure}

\section{Polylla meshes vs Voronoi based meshes}
\label{sec:experimental_evaluation}

Currently, if polygonal meshes different from triangulations or quad meshes are required to solve some scientific or engineering problems,  constrained Voronoi meshes are chosen. We believe that the polygonal meshes generated by Polylla  can be an alternative to Voronoi meshes when researchers and engineers would like to model domains using arbitrary polygon shapes. That is why in this section we show standard statistics of these new polygonal meshes  in contrast to constrained  Voronoi meshes. The C++ implementation of Polylla used for the experiments can be downloaded from this repository\footnote{\url{https://github.com/ssalinasfe/Polylla-Mesh}}.

The experiment was designed as follows: the domain is a square with  a set of random points in its interior varying from $10^2$ to $10^6$. A tolerance $\pm \gamma$ is defined and used in case of  a point $p$ is too close to one square side; in that case the point $p$ is inserted in the border edge.  The statistics of the polygonal meshes generated by Polylla starting from Delaunay triangulations generated  by using Detri2 \cite{Detri2} are summarized in Table \ref{table:results}.  In the case constrained Voronoi meshes, Deldir \cite{deldir} was used to generate the meshes and to obtain their mesh  information. The statistics for the same initial points (sites) are presented in Table \ref{table:resultsVoronoi}.


\begin{table}
\centering
\resizebox{\columnwidth}{!}{%
\begin{tabular}{|r|r|r|r|r|r|r|r|}
\hline
\multicolumn{1}{|c|}{\begin{tabular}[c]{@{}c@{}}Input\\ points\end{tabular}} & \multicolumn{1}{c|}{\begin{tabular}[c]{@{}c@{}}Triangles\end{tabular}} & \multicolumn{1}{c|}{\begin{tabular}[c]{@{}c@{}}Terminal-edge\\ regions\end{tabular}} & \multicolumn{1}{c|}{\begin{tabular}[c]{@{}c@{}}Polylla\\ polygons\end{tabular}} & \multicolumn{1}{c|}{\begin{tabular}[c]{@{}c@{}}Maximum barrier-edge  \\ tips in non-simple  polygons\end{tabular}} & \multicolumn{1}{c|}{\begin{tabular}[c]{@{}c@{}}Barrier-edge \\ tips\end{tabular}}  & \multicolumn{1}{c|}{\begin{tabular}[c]{@{}c@{}}Average triangles \\ per polygon\end{tabular}} & \multicolumn{1}{c|}{\begin{tabular}[c]{@{}c@{}}Average polygon\\ vertices \end{tabular}} \\ \hline
10      & 9       & 3      & 3      & 0 & 0     & 3.00 & 6.67 \\ \hline
$10^2$     & 180     & 35     & 38     & 1 & 3     & 4.74 & 6.89 \\ \hline
$10^3$    & 1911    & 281    & 305    & 2 & 24    & 6.27 & 8.30 \\ \hline
$10^4$   & 19768   & 3034   & 3228   & 4 & 194   & 6.12 & 8.13 \\ \hline
$10^5$  & 199163  & 30288  & 32271  & 4 & 1984  & 6.17 & 8.17 \\ \hline
$10^6$ & 1997497 & 302944 & 322561 & 4 & 19647 & 6.19 & 8.19 \\ \hline
\end{tabular}
}
\caption{Geometric information of Polylla mesh }
\label{table:results}
\end{table}

\begin{table}
\centering
\begin{tabular}{|r|r|r|r|r|}
\hline
\multicolumn{1}{|c|}{\begin{tabular}[c]{@{}c@{}}Input \\ Points\end{tabular}} & \multicolumn{1}{c|}{\begin{tabular}[c]{@{}c@{}}Voronoi \\ vertices\end{tabular}} & \multicolumn{1}{c|}{\begin{tabular}[c]{@{}c@{}}Voronoi\\ Regions\end{tabular}} & \multicolumn{1}{c|}{\begin{tabular}[c]{@{}c@{}}Voronoi\\ Edges\end{tabular}} & \multicolumn{1}{c|}{\begin{tabular}[c]{@{}c@{}}Edges\\ per Region\end{tabular}} \\ \hline
10                                                                            & 22                                                                               & 10                                                                              & 31                                                                           & 4.90                                                                             \\ \hline
$10^2$                                                                        & 202                                                                              & 100                                                                             & 301                                                                          & 5.75                                                                            \\ \hline
$10^3$                                                                        & 2002                                                                             & 1000                                                                            & 3001                                                                         & 5.90                                                                             \\ \hline
$10^4$                                                                        & 20001                                                                            & 10000                                                                           & 30001                                                                        & 5.96                                                                          \\ \hline
\end{tabular}
\caption{Geometric information of constrained Voronoi diagram }
\label{table:resultsVoronoi}
\end{table}

In Table \ref{table:results}, we can observe that over $10^4$ input points, the number of initial triangles per polygon is $6.5$ and the number of vertices per polygon is $8.5$, both values  on average. The number of barrier-edge tips is less than $3\%$ of the number of points, so the reparation phase just adds less than $10\%$ of polygons to the mesh. If we compare Polylla meshes with the meshes generated from the Voronoi diagram, the constrained  Voronoi meshes contain 3 times more polygons than our meshes. Each Voronoi region based polygon is delimited  in average by  6 edges and  Polylla polygons by 8 edges. 
Moreover, the Polylla meshes use only the points given as input; in contrast, the Voronoi based meshes use new points, the Voronoi points, one per each triangle of the Delaunay mesh. Since the number of triangles is greater than the number of input points, the size of the Voronoi mesh is greater than the size of the Polylla  mesh not only in terms of polygons, but also in terms of mesh points.

It is worth mentioning that a Polylla mesh does not need to insert extra-points at the boundary to fit the domain geometry; in contrast the constrained  Voronoi mesh includes new points at the boundary/interfaces inserted while cutting the Voronoi regions that go outside the domain.

 We have also evaluated the time performance of Polylla. The results  shown in Fig. \ref{fig:timecomp} were made on  Patagón, a computer with two Intel Xeon Platinum 8260 of 2.4Ghz, located at the  Universidad Austral de Chile \cite{patagon-uach}. The results are the average time of  five program executions with same data-set in a different core.  
 The time of each phase  is included together with time to generate the initial triangulation. 
 Of the three phases described in section \ref{sec:the_algorithm}, the  phase that takes less time is the non-simple polygon reparation phase. The main reason is that the number of terminal-edge regions with barrier edges  is around 1\% of the total  number  terminal-edges regions formed from random point sets  as  can be seen in Fig. \ref{fig:timecomp}.
 The time  of the Label phase is higher than the time of the traversal phase. This can be explained due to the use of floating-point arithmetic  to calculate the length of each edge  and so to assign the proper label. 


In order to compare the CPU time required to generate constrained Voronoi and Polylla meshes,  we looked for free and open source tools. The CGAL library provides  {\em  2D Voronoi Diagram Adaptor} package~\cite{cgal:k-vda2-21b} to generate Voronoi diagrams, but this package does not provide a method to cut Voronoi regions against the domain boundary. So we  implemented  a function to do this process. In contrast, Detri2~\cite{Detri2} offers a robust method to generate constrained Voronoi meshes.   Table~\ref{table:compvoropoylla} shows the cpu-times in seconds needed to generate polygonal meshes using Polylla, Detri2D and CGAL from 100000, 500000 and 1000000 points over a L-shape domain. Experiments were run in a CPU Intel(R) Core(TM) i5-9600K of 3.70GHz. Detri2d computes the CDT first and then the CVD from the CDT. That is why we included the time to generate the CDT separated from the generation of the CVD. 
The Polylla time includes only the time spent in  the three phases that processes the input triangulation to generate the polygon mesh. The CGAl algorithm to generate the VD includes generation of the Delaunay triangulation first and from this triangulation computes the Voronoi Diagram. The CVD cost in CGAL is the sum of VD and CV. This preliminary comparison shows that the time to generate a CDT using Detri2d or CGAL plus the time Polylla takes to generate the polygonal mesh is much less than the time needed for generating a CVD either using Detri2 or the CGAl library. The main reason is that constrained Voronoi based meshing algorithms require to compute and insert new points (Voronoi points), to cut Voronoi regions  and insert new boundary points. In contrast, the Polylla algorithm  just needs  to process the vertices, edges and triangles of the input triangulation.  The domain boundary is already represented by some triangle edges.

\begin{table}[]
\centering
\begin{tabular}{lr|rr|rllr|}
\cmidrule{3-8} \addlinespace[-.6em]
                               & \multicolumn{1}{c|}{}        & \multicolumn{2}{c|}{Detri2}                           & \multicolumn{4}{c|}{CGAL}                                                                                          \\ \hline
\multicolumn{1}{|l|}{Vertices} & \multicolumn{1}{c|}{Polylla} & \multicolumn{1}{c|}{CDT}   & \multicolumn{1}{c|}{CVD} & \multicolumn{1}{c|}{CDT}    & \multicolumn{1}{l|}{VD}     & \multicolumn{1}{l|}{CV}     & \multicolumn{1}{c|}{CVD} \\ \hline
\multicolumn{1}{|l|}{100000}   & 0.193                        & \multicolumn{1}{r|}{0.326} & 229.489                  & \multicolumn{1}{r|}{1.712}  & \multicolumn{1}{l|}{28.39}  & \multicolumn{1}{l|}{9.68}   & 38.07                    \\ \hline
\multicolumn{1}{|l|}{500000}   & 0.7684                       & \multicolumn{1}{r|}{1.703} & 5736.473                 & \multicolumn{1}{r|}{9.006}  & \multicolumn{1}{l|}{335.26} & \multicolumn{1}{l|}{56.16}  & 391.42                   \\ \hline
\multicolumn{1}{|l|}{1000000}  & 1.3294                       & \multicolumn{1}{r|}{3.346} & 21341.734                & \multicolumn{1}{r|}{16.082} & \multicolumn{1}{l|}{870.52} & \multicolumn{1}{l|}{120.52} & 991.04                   \\ \hline
\end{tabular}
\caption{Time comparison, in seconds, of Polylla vs Constrained Voronoi Diagram (CDV) using Detri2 and CGAL. CDT is the time to generate a Constrained Delaunay Triangulation, VD is the time of generate Voronoi Diagram, CV is the time need to cut all regions of Voronoi Diagram and CVD is the total time of generate the Constrained Voronoi Diagram. }
\label{table:compvoropoylla}
\end{table}

\begin{figure}
\centering     

\includegraphics[width=0.5\textwidth]{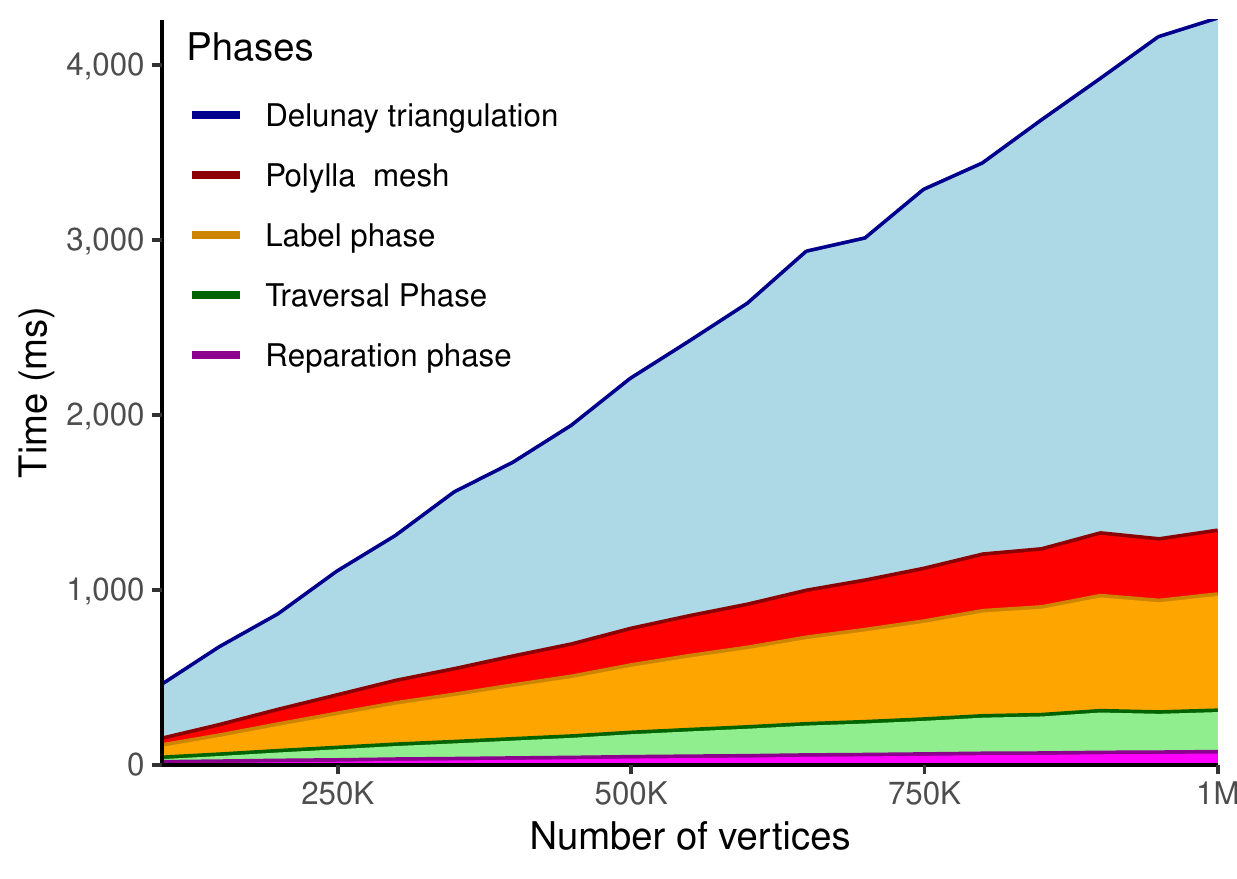}
\caption{Time performance of the  different algorithm phases. The label Polylla mesh indicates  the time of the three phases.
}
\label{fig:timecomp} 
\end{figure}


\begin{figure}[]
\centering     
\subfigure[
]{\label{fig:polymesh}\includegraphics[width=0.4\textwidth]{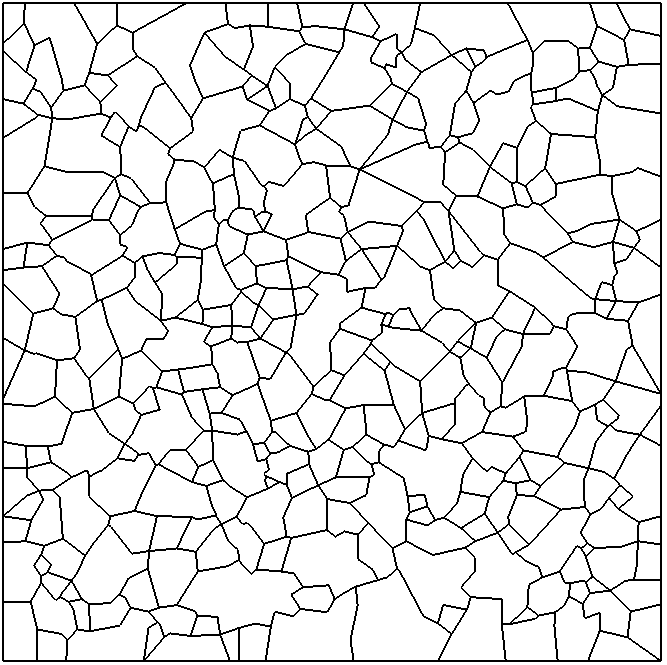}} 
\subfigure[
]{\label{fig:voromesh10000}\includegraphics[width=0.4\textwidth]{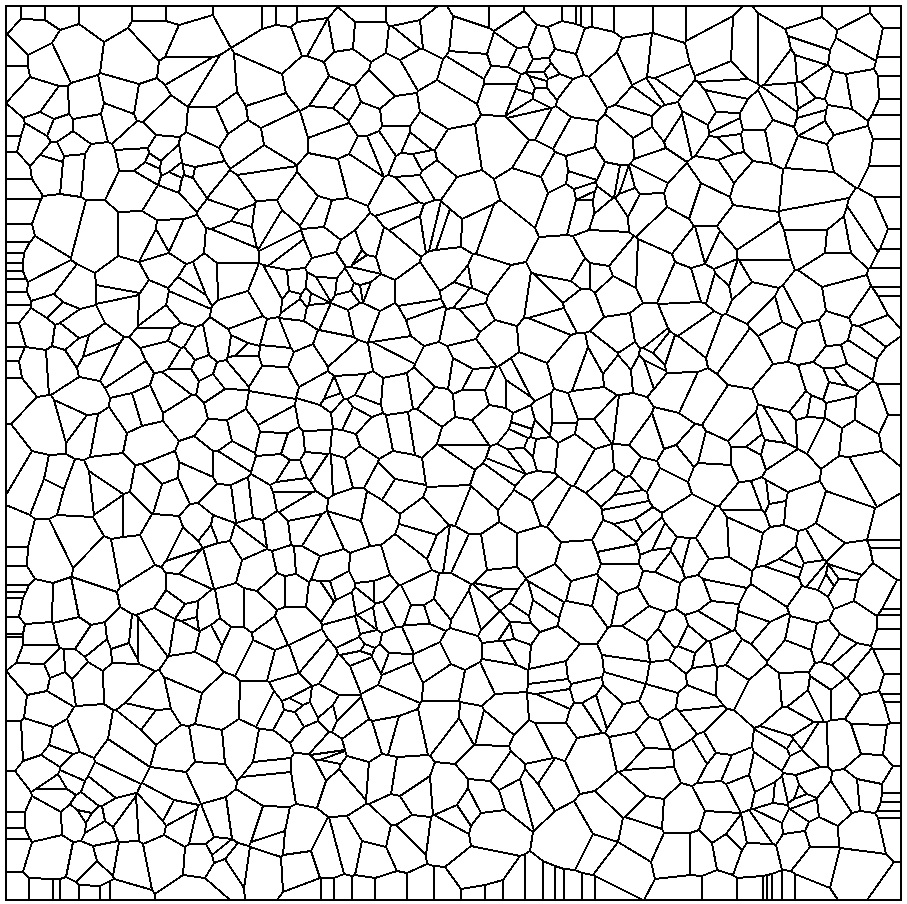}}
\caption{Meshes generated from 1000 random points. (a) Polylla mesh.\textbf{(b)} Constrained Voronoi mesh generated with Detri2qt \cite{Detri2}. }
\label{figs:voro_comp} 
\end{figure}

\begin{figure}[]
\centering     
\subfigure[
]{\label{fig:PSLGUnicornTriangulation}\includegraphics[width=0.25\textwidth]{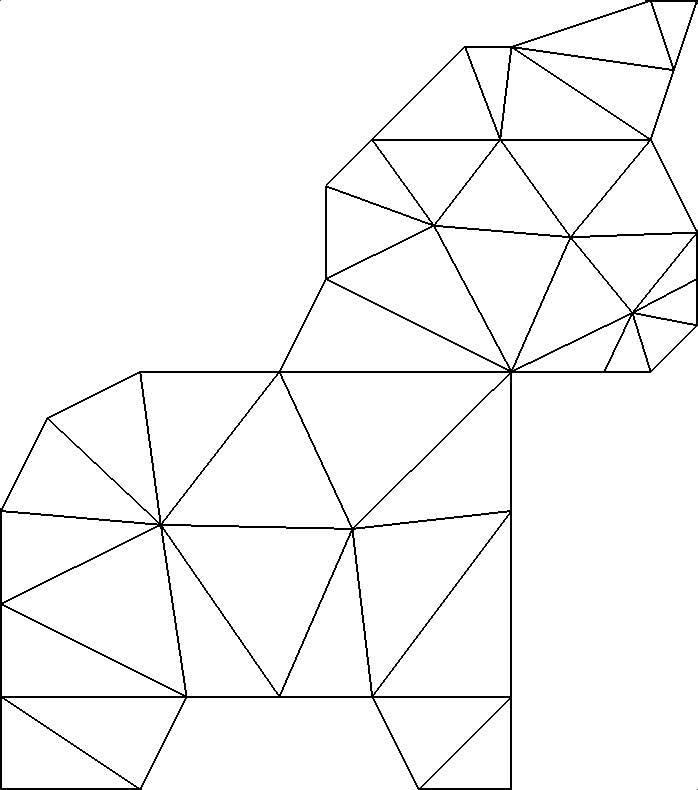}}%
\subfigure[
]{\label{fig:PSLGUnicorn}\includegraphics[width=0.25\textwidth]{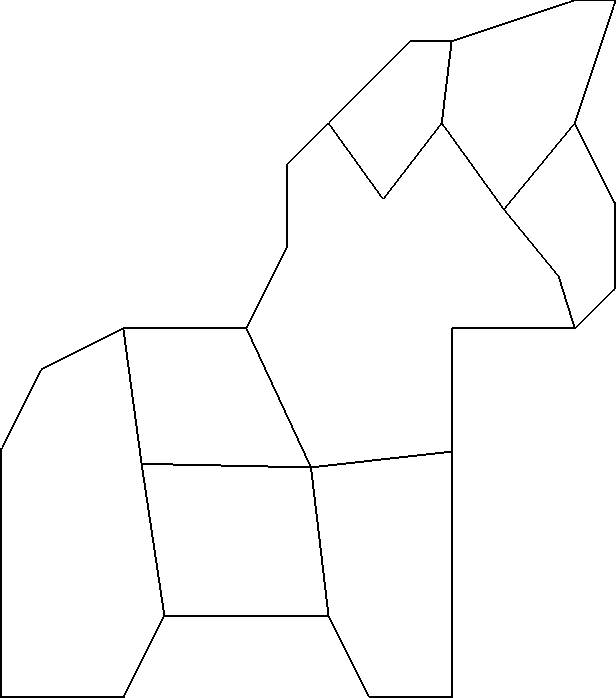}}%
\subfigure[
]{\label{fig:PSLGUnicornVoronoi}\includegraphics[width=0.25\textwidth]{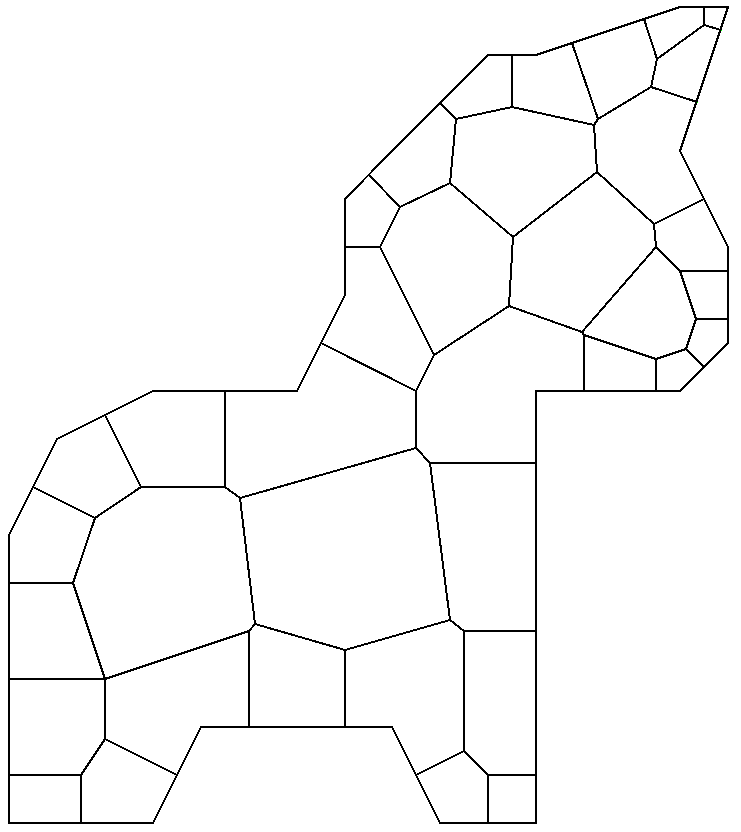}}%
\subfigure[
]{\label{fig:PSLGUnicorn100}\includegraphics[width=0.25\textwidth]{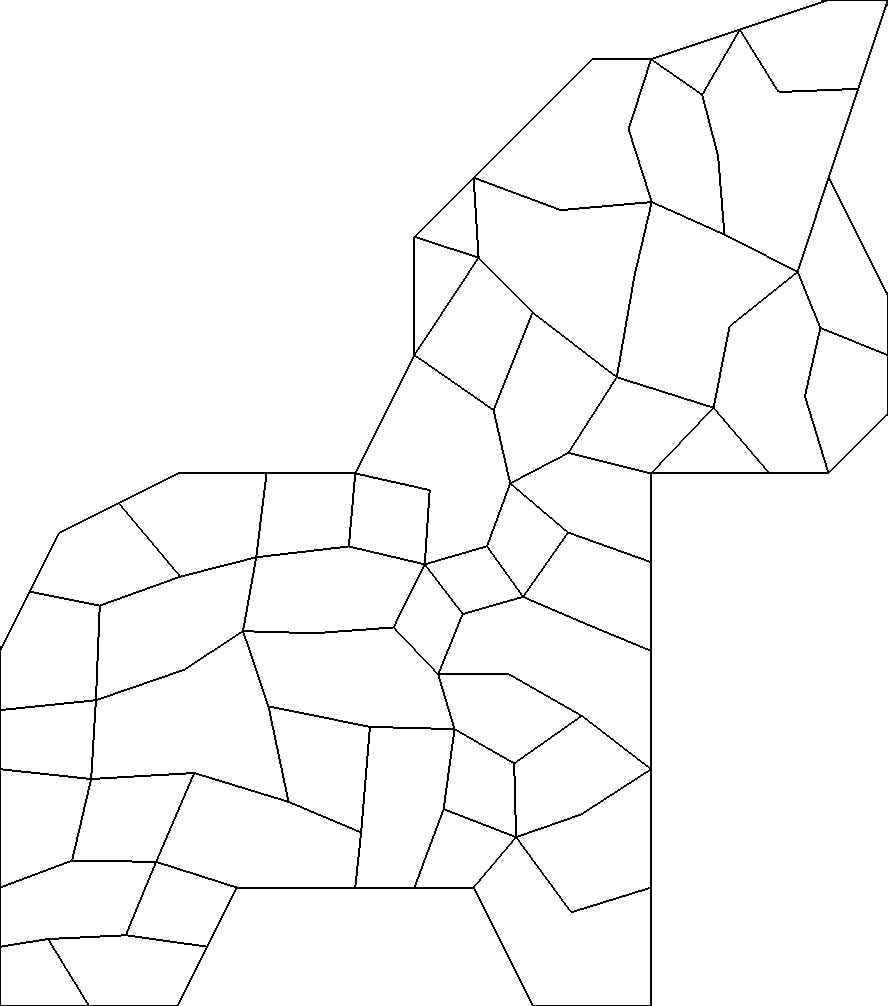}}

\caption{ Qualitative comparison using an Unicorn PLSG as input \cite{AlejandroNUA2019} (a) Triangulation of unicorn PLSG with 36 vertices. (b) Polylla mesh includes exactly the 36 input vertices. (c) The constrained Voronoi mesh requires 100 vertices to model the same input. (d) Polylla mesh from a refined  Delaunay triangulation with 100 vertices.}
\label{figs:univornPSLG} 
\end{figure}

In Fig.~\ref{figs:voro_comp} a qualitative comparison between a Polylla mesh and a constrained Voronoi  mesh can be observed. Both meshes were generated from the same initial random point set. Fig.~\ref{fig:polymesh} shows the polygonal mesh generated by Polylla and Fig.~\ref{fig:voromesh10000} the constrained Voronoi mesh generated by Detri2qt. Some other examples of polygonal meshes generated for non convex domains are shown in Fig. \ref{figs:facePSLG} and Fig. \ref{figs:univornPSLG}.

 Fig. \ref{figs:univornPSLG} shows polygonal meshes for an Unicorn example generated by both  the Polylla  (Fig. \ref{fig:PSLGUnicorn}) and  Detri2qt (Fig. \ref{figs:univornPSLG}) tools. The initial triangulation and interior points are shown  in Fig. \ref{fig:PSLGUnicornTriangulation}. As expected, the  constrained Voronoi  mesh contains more points (100 vertices) and polygons than the Polylla mesh as shown in Fig. \ref{figs:univornPSLG}. Just for a qualitative comparison, we used Polylla to generate a polygonal mesh from 100 vertices and this mesh is shown in Fig. \ref{fig:PSLGUnicorn100}. 




\section{Preliminary Simulation Results}
\label{sec:simulation_results}
In this section, we assess the Polylla meshes using the virtual element method (VEM)~\cite{Basisprinciples}. To this end, an L-shaped domain is considered. Fig.~\ref{figs:LshapedPolylla} shows this domain meshed with a random and a semiuniform Polylla sample mesh. For comparison purposes, we also consider random and semiuniform Voronoi meshes (sample meshes are depicted in Fig.~\ref{figs:LshapedVoronoi}). The chosen problem is governed by the Laplace equation and its exact solution is given by~\cite{MITCHELL2013350}
\begin{equation*}
u(x_1,x_2)=r^{2/3} \sin(2/3\,\theta), \quad r=\sqrt{x_1^2+x_2^2}, \quad \theta(x_1,x_2)=\arctan(x_2/x_1).
\end{equation*}

The boundary conditions are of Dirichlet type with the exact solution imposed on the entire domain boundary. The re-entrant corner of the L-shaped domain introduces a singularity in the solution that manifests  itself as unbounded derivatives of $u$ at the origin. The numerical solution (denoted by $u_h$) is assessed  through its convergence with mesh refinements. Figs.~\ref{figs:NormsLshapedRandom} and 
\ref{figs:NormsLshapedSemiuniform} present the $L^2$ norm and the $H^1$ seminorm of the error, where it is shown that the VEM on Polylla (random and semiuniform) and Voronoi (random and semiuniform) meshes delivers accurate solutions with optimal convergence rates of 2 (for the $L^2$ norm) and 1 (for the $H^1$ seminorm).

The performance of VEM using Polylla (random and semiuniform) and Voronoi (random and semiuniform) meshes are compared in Fig.~\ref{figs:PerformanceLshaped}, where the $H^1$ seminorm of the error and the normalized CPU time are each plotted as a function of the number of degrees of freedom (DOF). The normalized CPU time is defined as the ratio of the CPU time of a particular simulation to the maximum CPU time found for any of the simulations that were run. Each of the four set of meshes (random Polylla, semiuniform Polylla, random Voronoi and semiuniform Voronoi) consists of four meshes of increasing number of degrees of freedom. Therefore, in total, there are sixteen meshes in the study. 

The CPU time is measured from the reading of the mesh until the solution of the system of equations is ended. Each mesh is run ten times and the CPU time recorded is the average CPU time. From Fig.~\ref{figs:PerformanceLshaped} it is observed that for equal number of degrees of freedom similar accuracy and computational cost are obtained for the four set of meshes.

Finally, for completeness of the presented numerical results, contour plots of the VEM solution on Polylla meshes are shown in Figs.~\ref{figs:ContourPlots} and~\ref{figs:ContourPlotsGrad} for the $u_h$ and $\bm{\nabla}u_h$ fields, respectively.

\begin{figure}[!bth]
\centering     
\subfigure[]{\label{fig:LshapedPolyllaRandom}\includegraphics[width=0.4\textwidth]{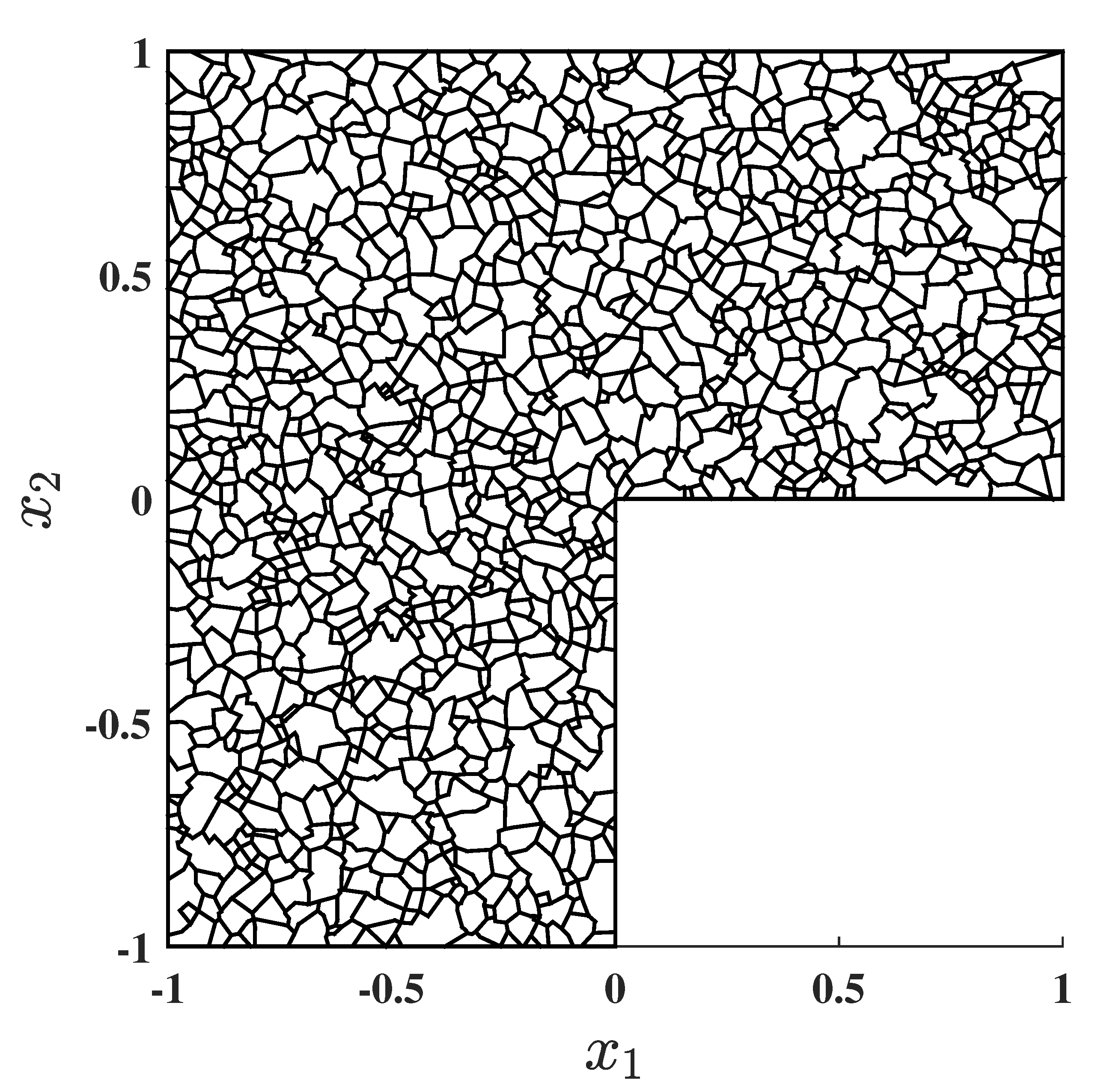}} \hspace{0.1cm}
\subfigure[]{\label{fig:LshapedPolyllaSemiuniform}\includegraphics[width=0.4\textwidth]{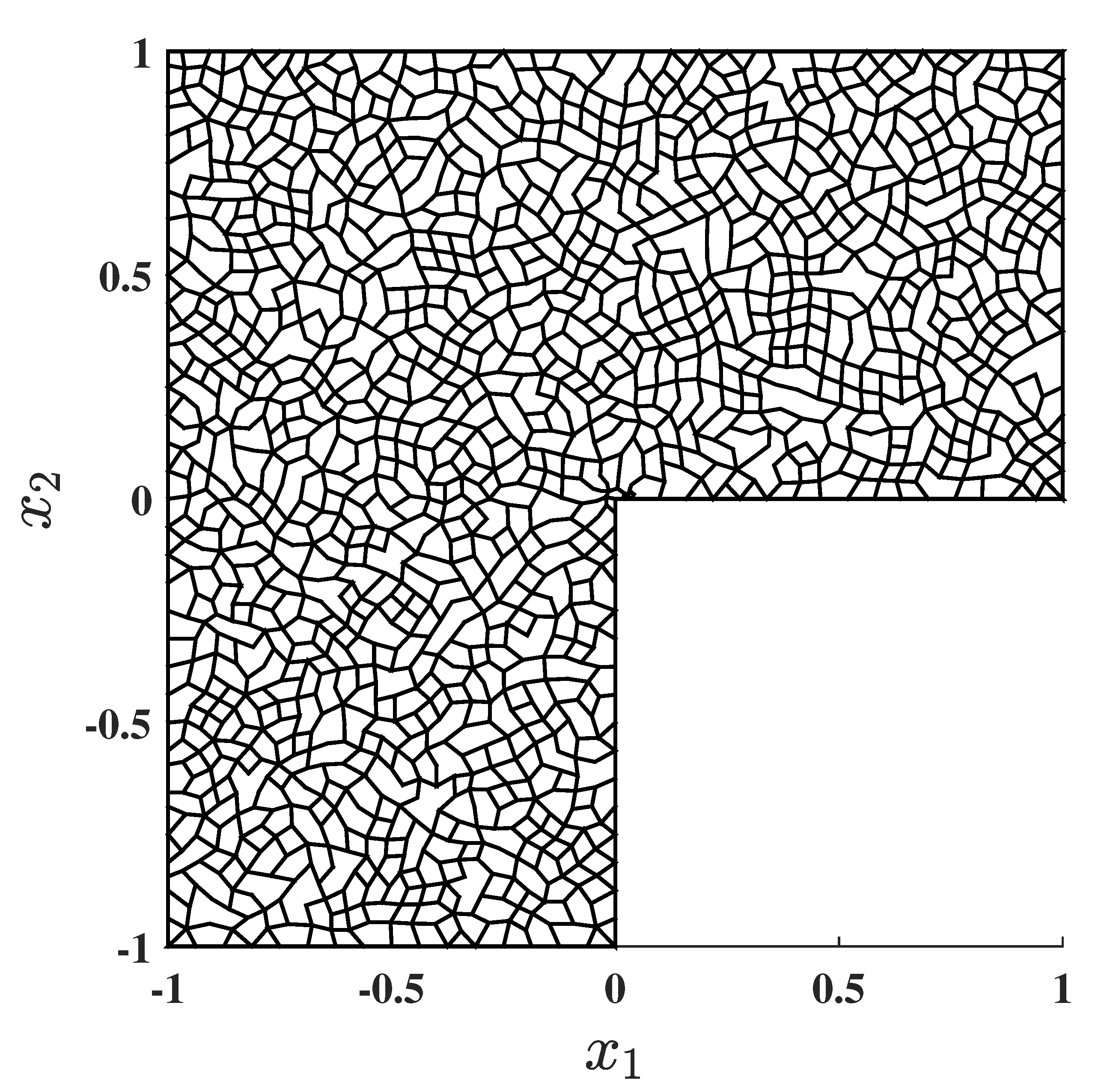}}
\caption{L-shaped domain meshed with \textbf{(a)} a random and \textbf{(b)} a semiuniform Polylla mesh.}
\label{figs:LshapedPolylla} 
\end{figure}

\begin{figure}[!bth]
\centering     
\subfigure[]{\label{fig:LshapedVoronoiRandom}\includegraphics[width=0.4\textwidth]{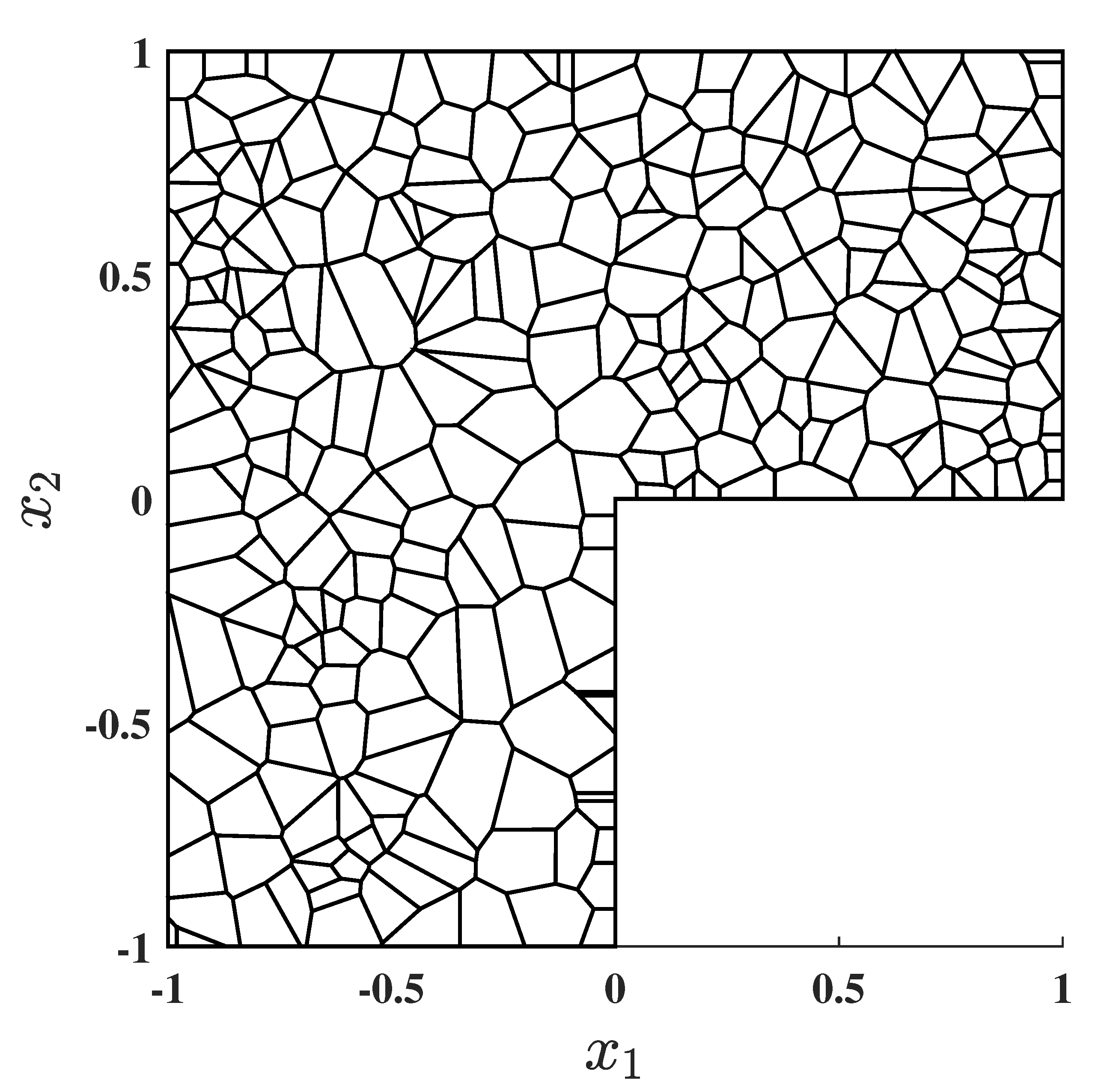}} \hspace{0.1cm}
\subfigure[]{\label{fig:LshapedVoronoiSemiuniform}\includegraphics[width=0.4\textwidth]{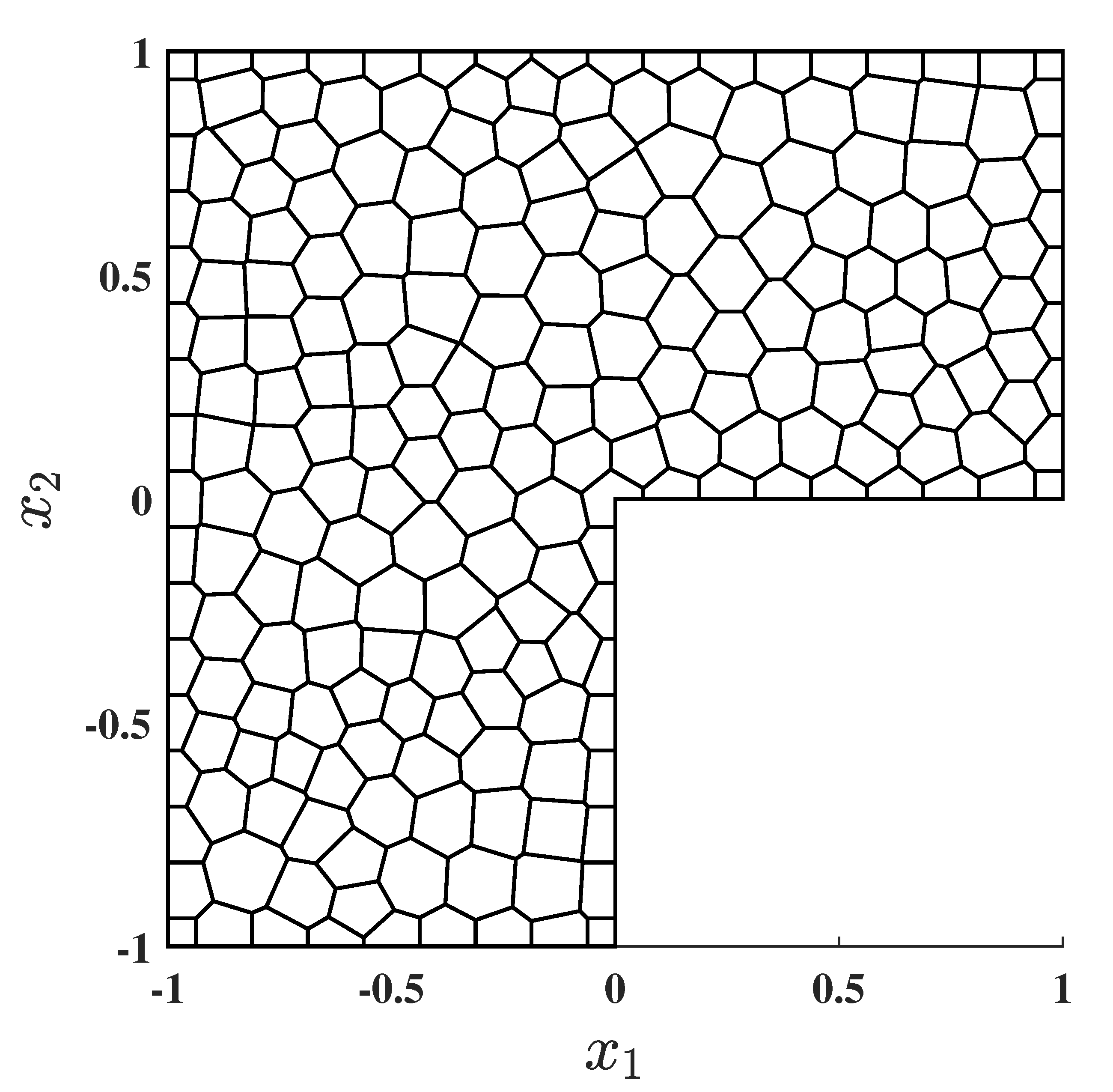}}
\caption{L-shaped domain meshed with \textbf{(a)} a random and \textbf{(b)} a semiuniform Voronoi mesh.}
\label{figs:LshapedVoronoi} 
\end{figure}

Note that random Polylla meshes (Fig~\ref{fig:LshapedPolyllaRandom}) were generated from constrained Delaunay triangulation built over  a L-shaped domain with points randomly generated in its interior. Semiuniform Polylla (Fig~\ref{fig:LshapedPolyllaSemiuniform}) meshes were generated  from a conforming Delaunay triangulation using pygalmesh~\cite{Schlomer_pygalmesh_Python_interface}. To generate the same number of triangles as the random mesh, an arbitrary maximum edge size constraint was given as input.  The random (Fig~\ref{fig:LshapedVoronoiRandom}) and semiuniform (Fig~\ref{fig:LshapedVoronoiSemiuniform}) Voronoi meshes were generated from  the points (sites) of both the  constrained Delaunay triangulation and the conforming Delaunay triangulation, respectively. Next, each infinite Voronoi region was intersected with the L-shaped domain to generate the constrained Voronoi mesh.

\begin{figure}[!bth]
\centering     
\mbox{
\subfigure[]{\label{fig:L2LshapedPolyllaRandom}\includegraphics[width=0.4\textwidth]{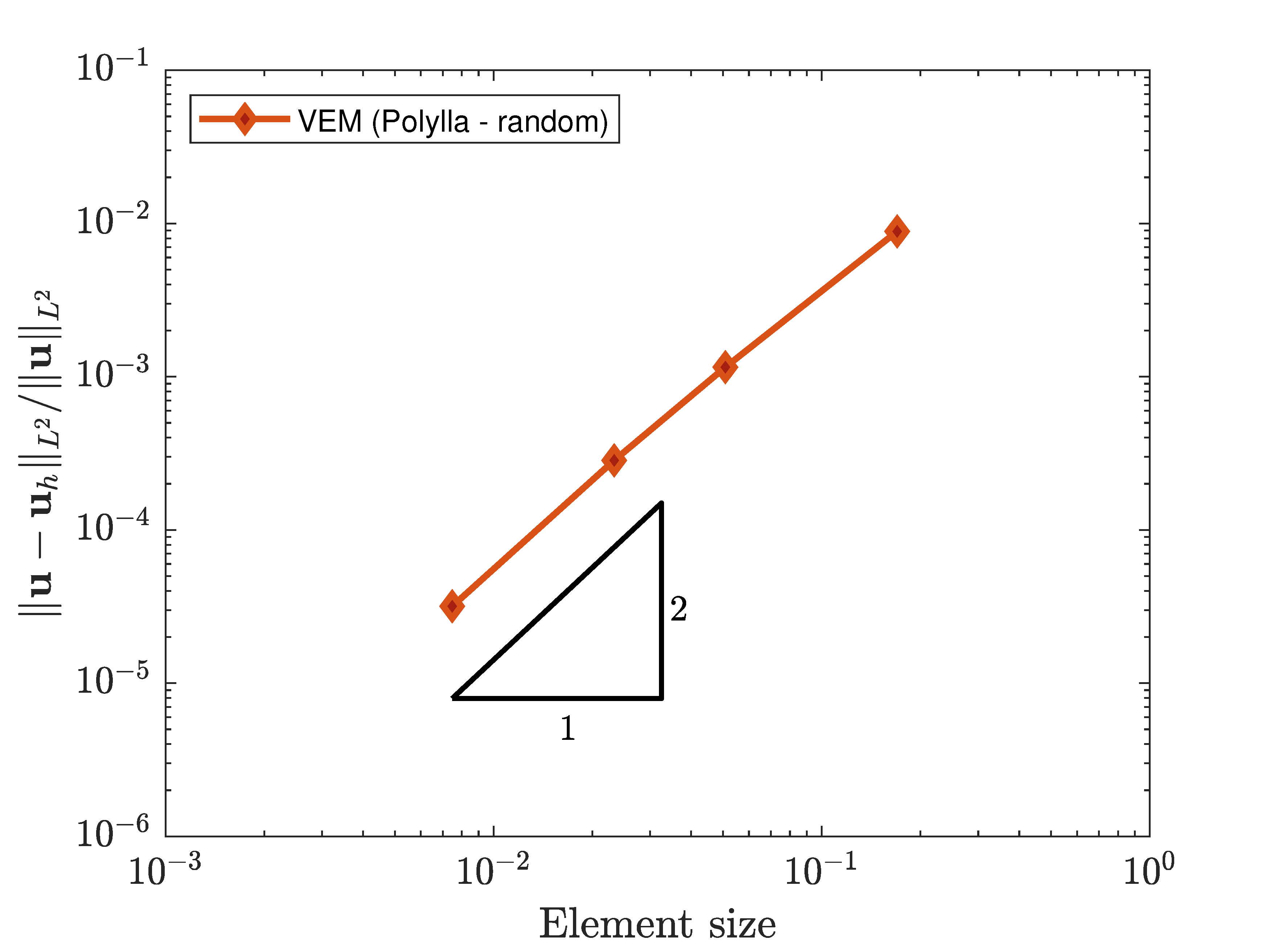}} \hspace{0.1cm}
\subfigure[]{\label{fig:H1LshapedPolyllaRandom}\includegraphics[width=0.4\textwidth]{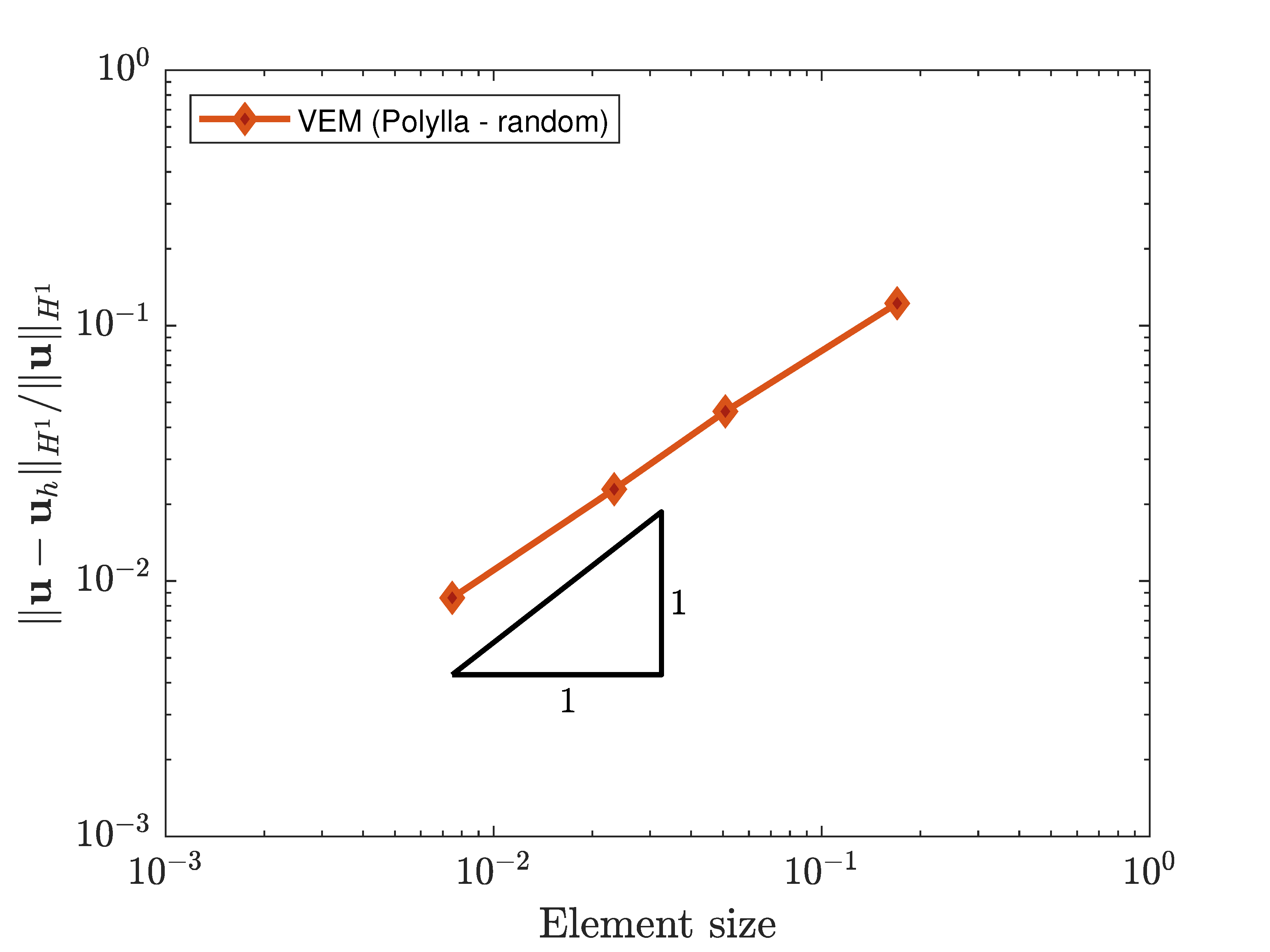}}
}
\mbox{
\subfigure[]{\label{fig:L2LshapedVoronoiRandom}\includegraphics[width=0.4\textwidth]{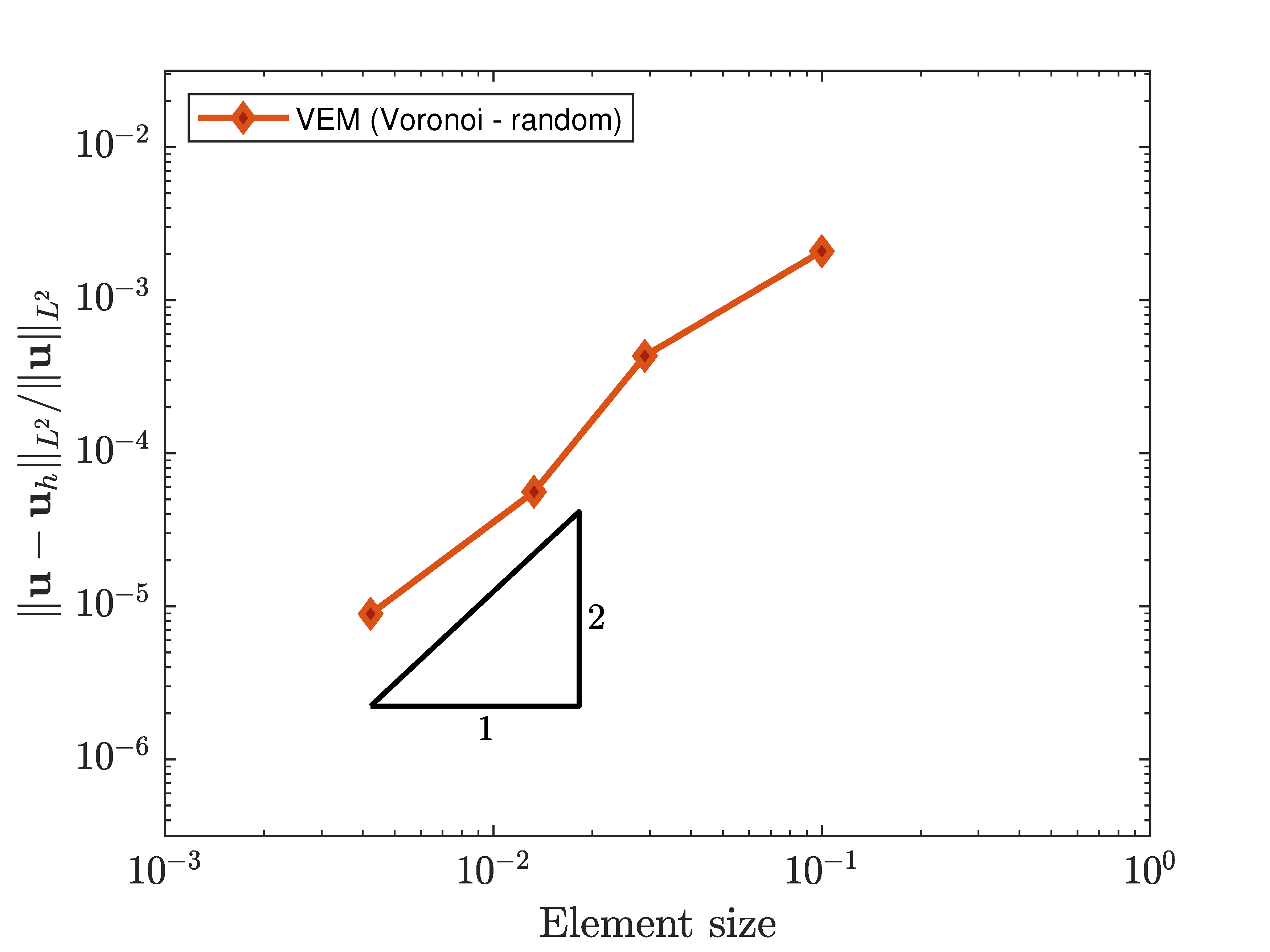}} \hspace{0.1cm}
\subfigure[]{\label{fig:H1LshapedVoronoiRandom}\includegraphics[width=0.4\textwidth]{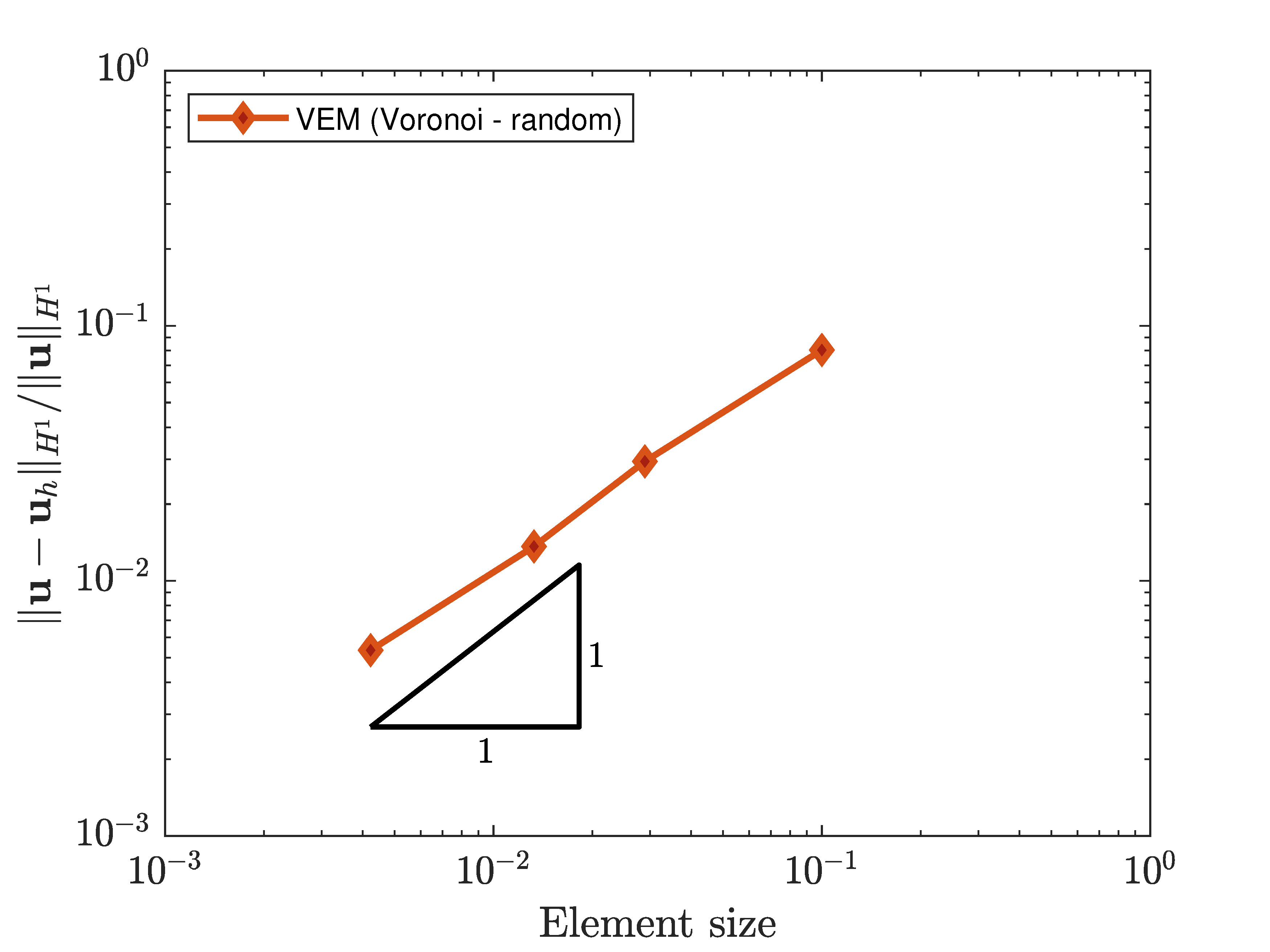}}
}
\caption{$L^2$ norm and $H^1$ seminorm of the error using the VEM on random Polylla meshes (\textbf{(a)} and \textbf{(b)}, 
respectively), and on random Voronoi meshes (\textbf{(c)} and \textbf{(d)}, respectively).}
\label{figs:NormsLshapedRandom} 
\end{figure}

\begin{figure}[!bth]
\centering     
\mbox{
\subfigure[]{\label{fig:L2LshapedPolyllaSemiuniform}\includegraphics[width=0.4\textwidth]{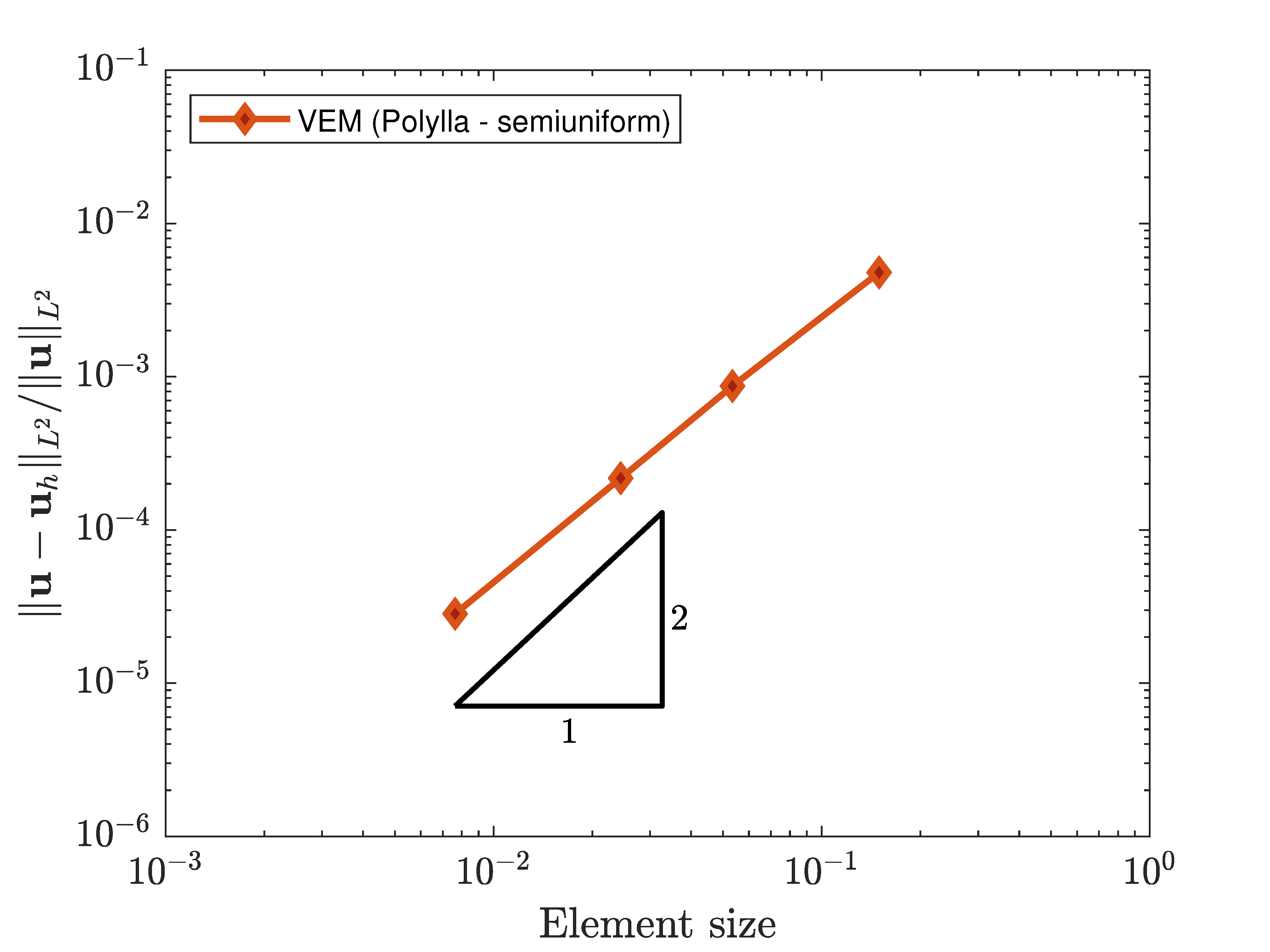}} \hspace{0.1cm}
\subfigure[]{\label{fig:H1LshapedPolyllaSemiuniform}\includegraphics[width=0.4\textwidth]{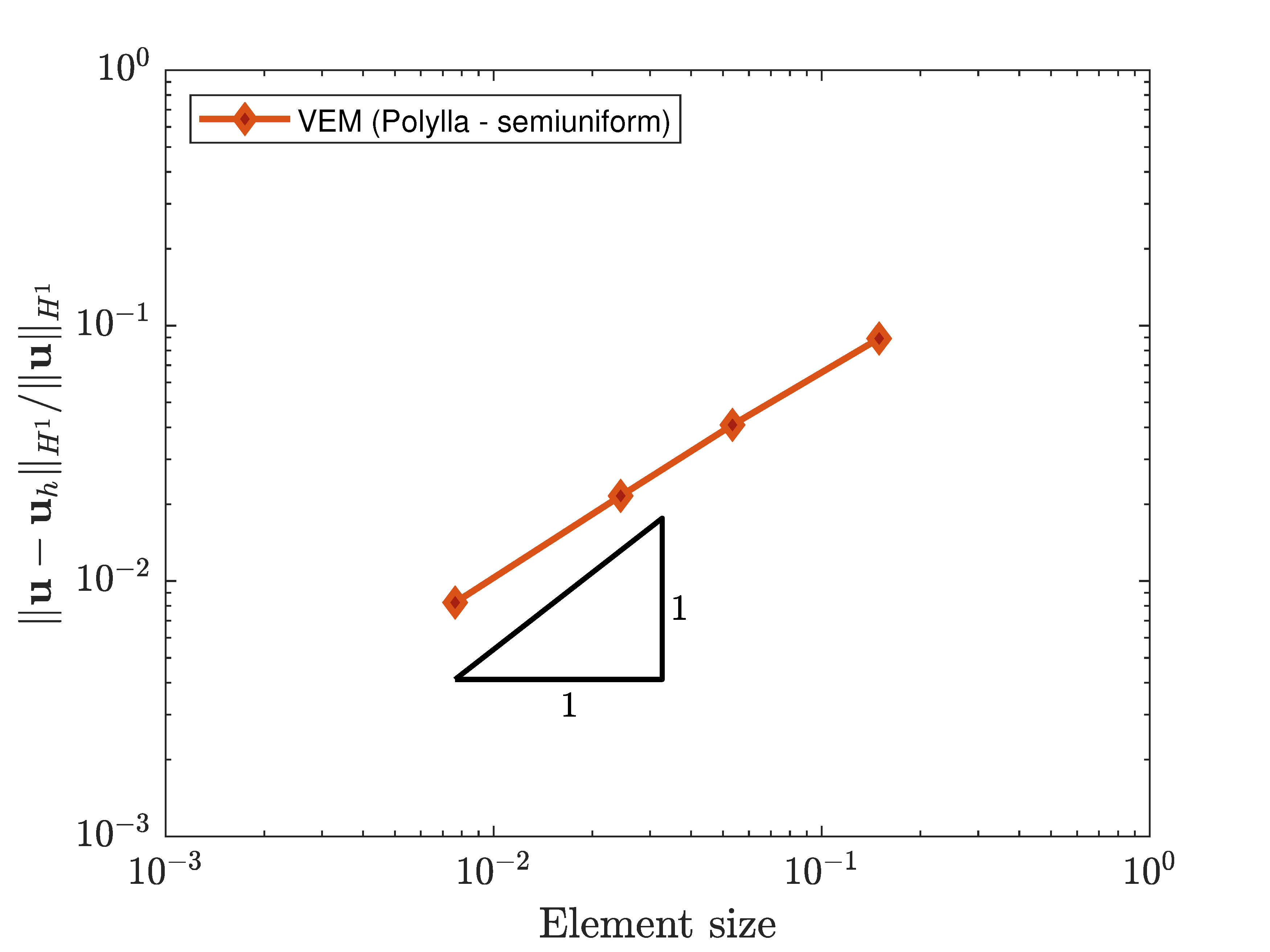}}
}
\mbox{
\subfigure[]{\label{fig:L2LshapedVoronoiSemiuniform}\includegraphics[width=0.4\textwidth]{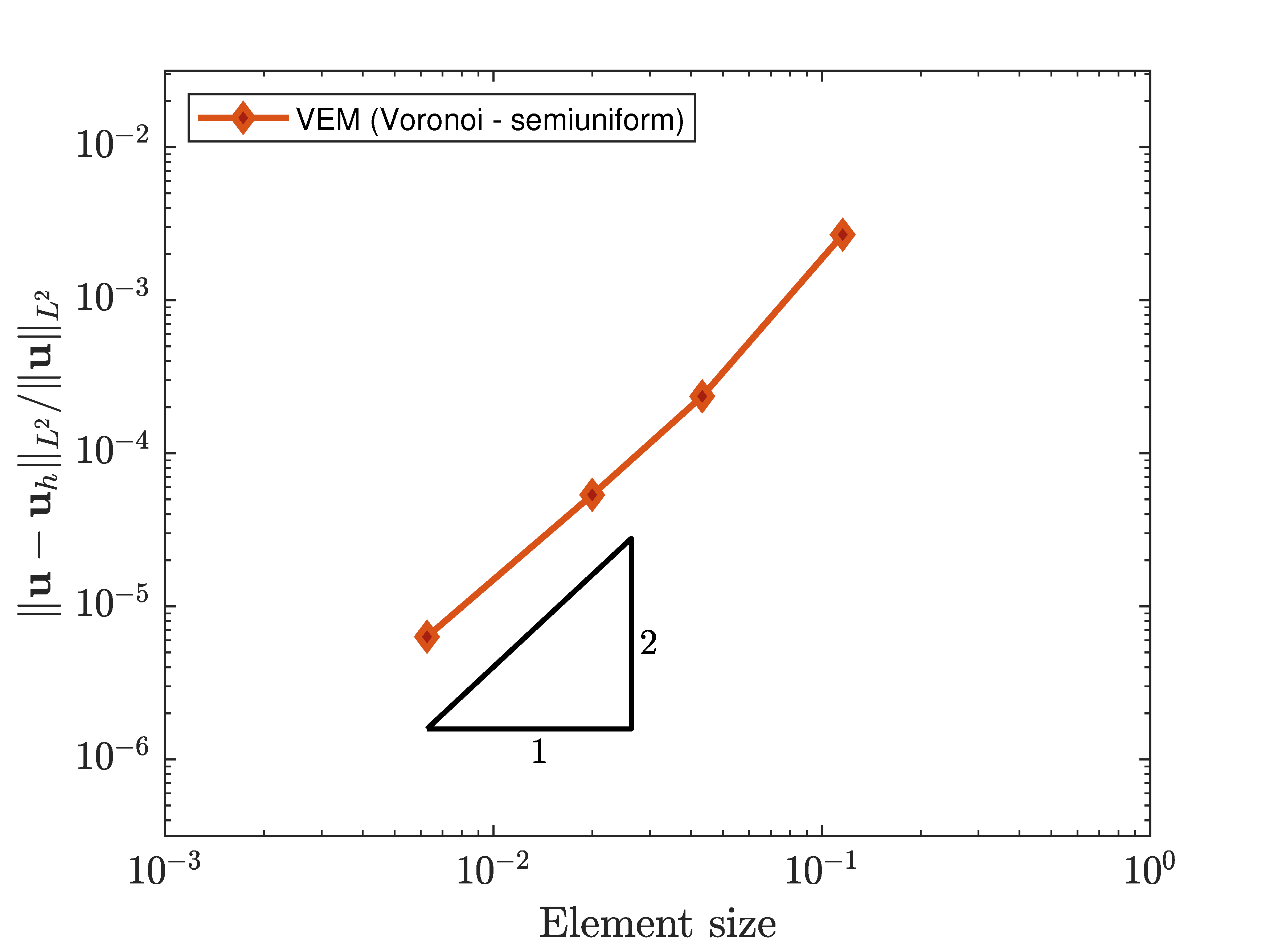}} \hspace{0.1cm}
\subfigure[]{\label{fig:H1LshapedVoronoiSemiuniform}\includegraphics[width=0.4\textwidth]{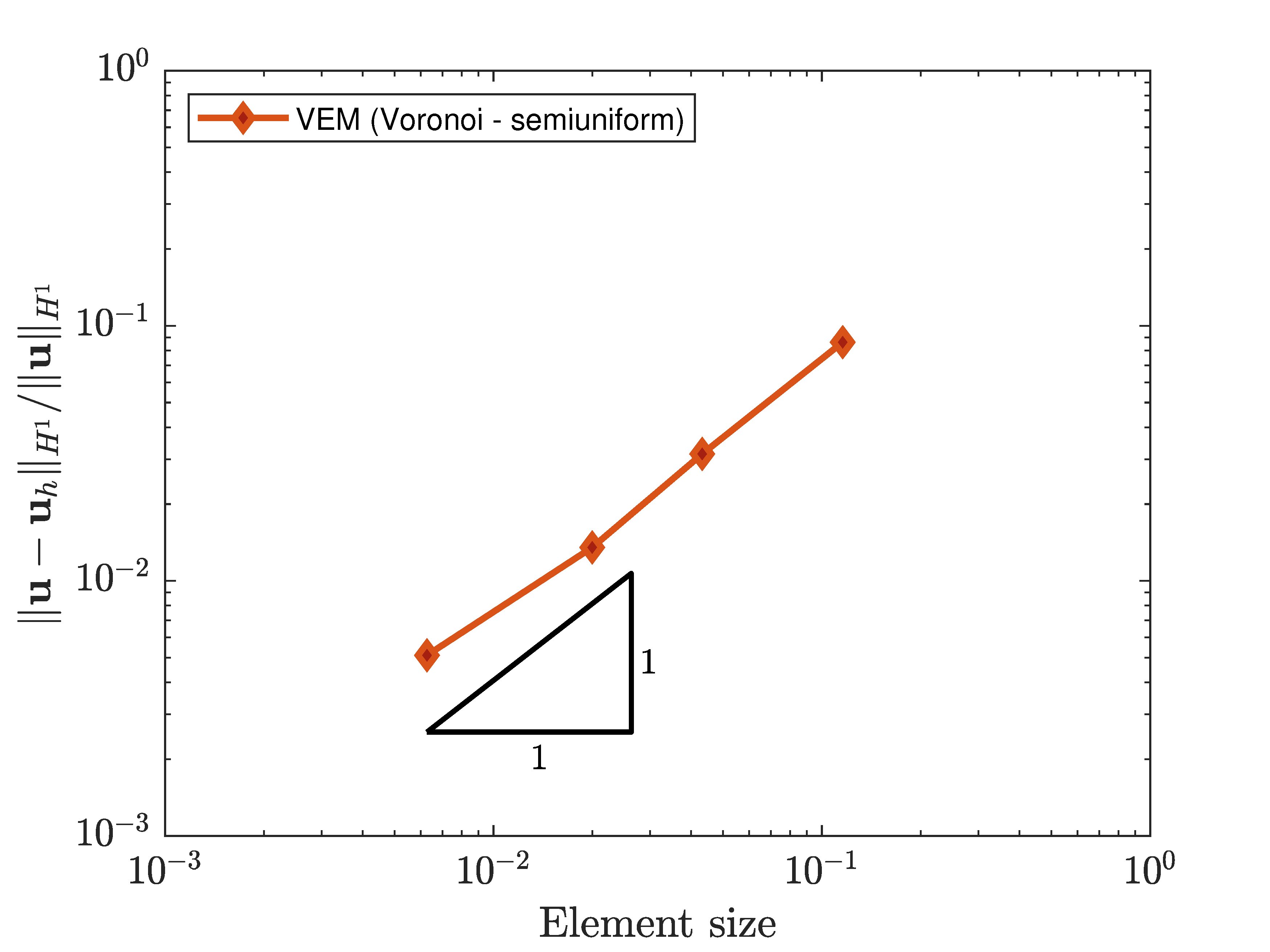}}
}
\caption{$L^2$ norm and $H^1$ seminorm of the error using the VEM on semiuniform Polylla meshes (\textbf{(a)} 
and \textbf{(b)}, respectively), and on semiuniform Voronoi meshes (\textbf{(c)} and \textbf{(d)}, respectively).}
\label{figs:NormsLshapedSemiuniform} 
\end{figure}

\begin{figure}[!bth]
\centering     
\mbox{
\subfigure[]{\label{fig:PerformanceLshapedH1}\includegraphics[width=0.4\textwidth]{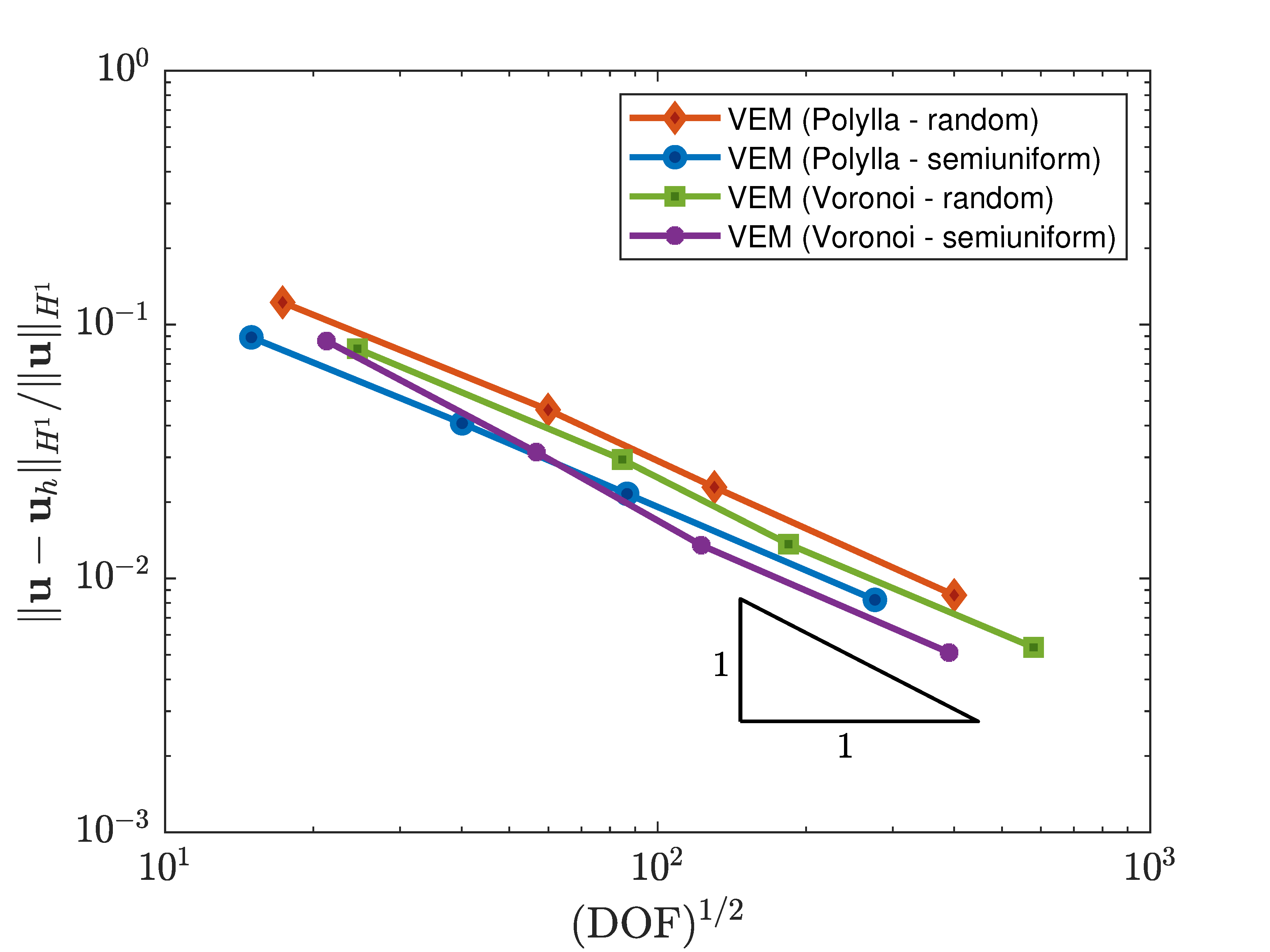}} \hspace{0.1cm}
\subfigure[]{\label{fig:PerformanceLshapedCPU}\includegraphics[width=0.4\textwidth]{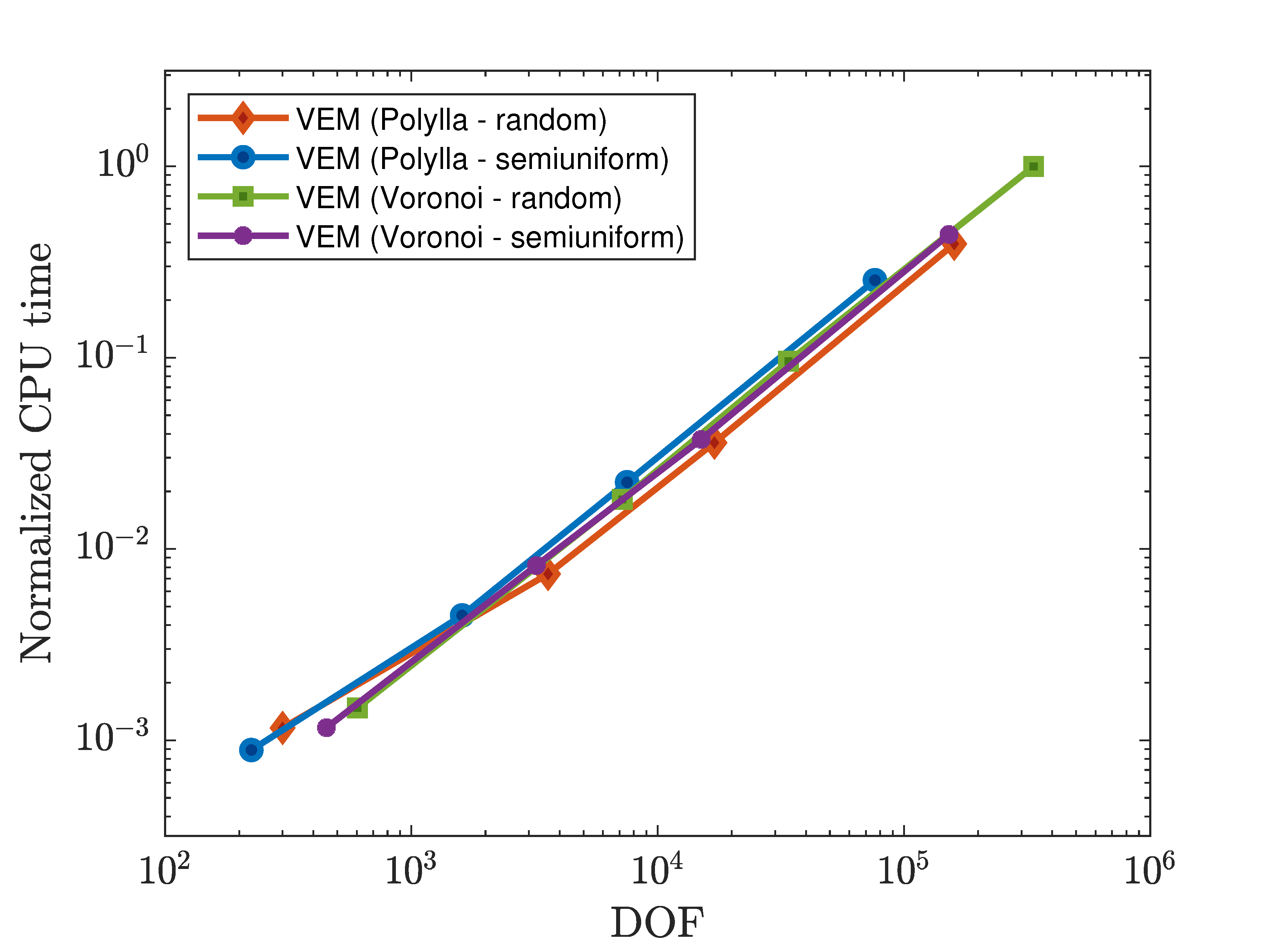}}
}
\caption{Performance of the VEM using Polylla and Voronoi meshes. \textbf{(a)} $H^1$ seminorm of the error and 
\textbf{(b)} CPU time as a function of the degrees of freedom (DOF).}
\label{figs:PerformanceLshaped} 
\end{figure}

\begin{figure}[!bth]
\centering     
\subfigure[]{\label{fig:ContourRandom}\includegraphics[width=0.4\textwidth]{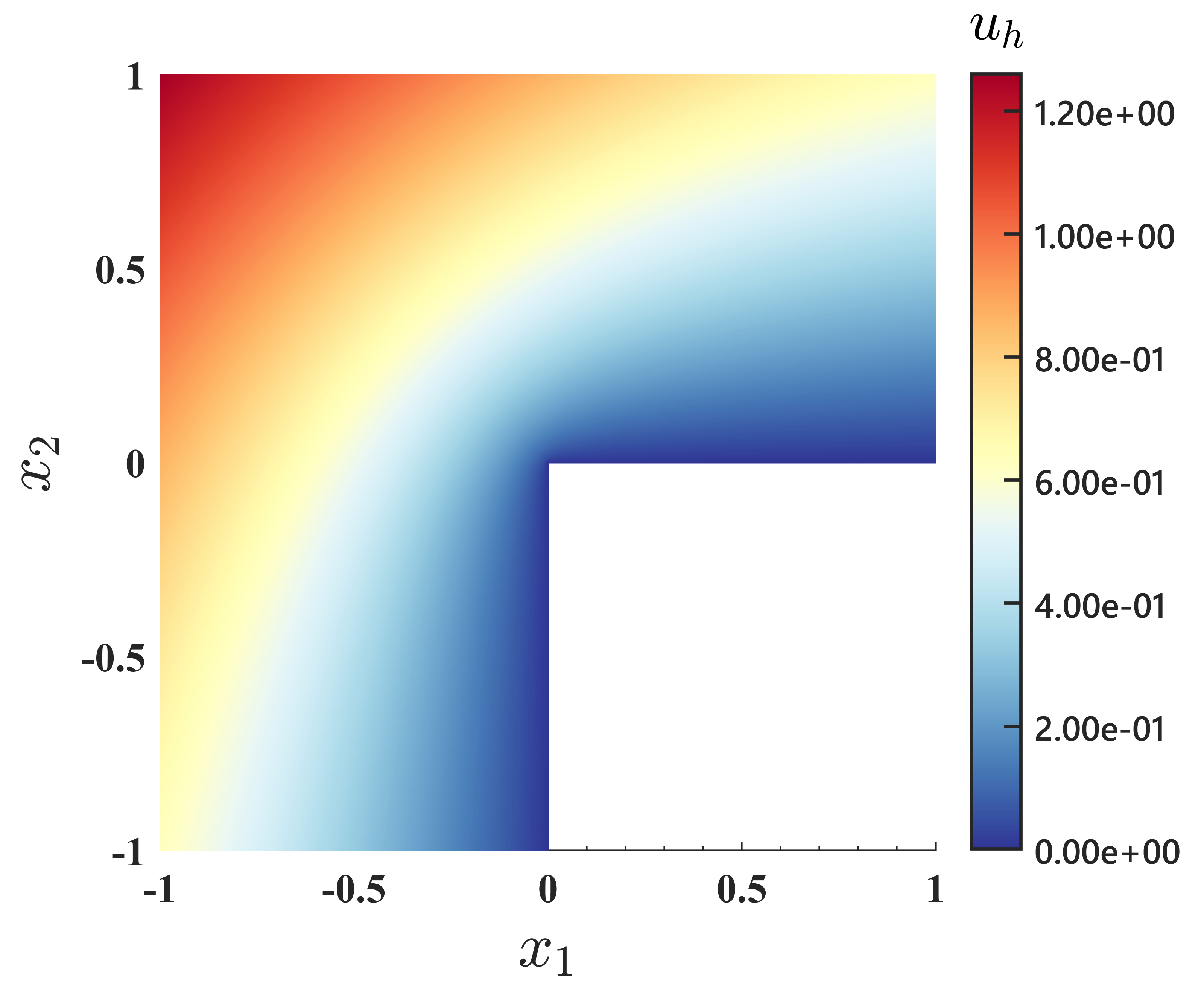}} \hspace{0.5cm}
\subfigure[]{\label{fig:ContourSemiuniform}\includegraphics[width=0.4\textwidth]{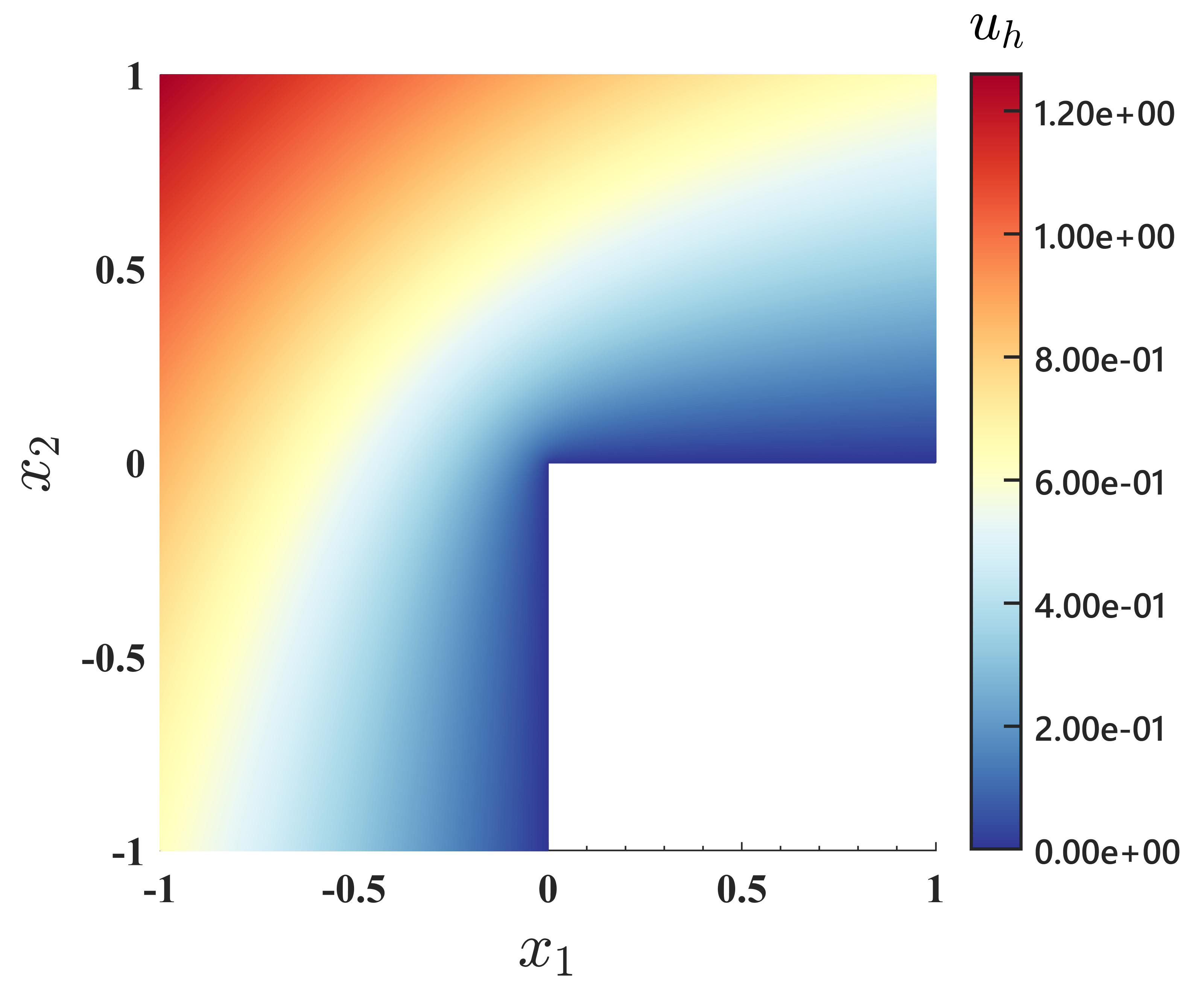}}
\caption{Contour plots of the VEM solution for the $u$ field. \textbf{(a)} Random and \textbf{(b)} semiuniform Polylla meshes.}
\label{figs:ContourPlots} 
\end{figure}

\begin{figure}[!bth]
\centering     
\subfigure[]{\label{fig:ContourRandomGrad}\includegraphics[width=0.4\textwidth]{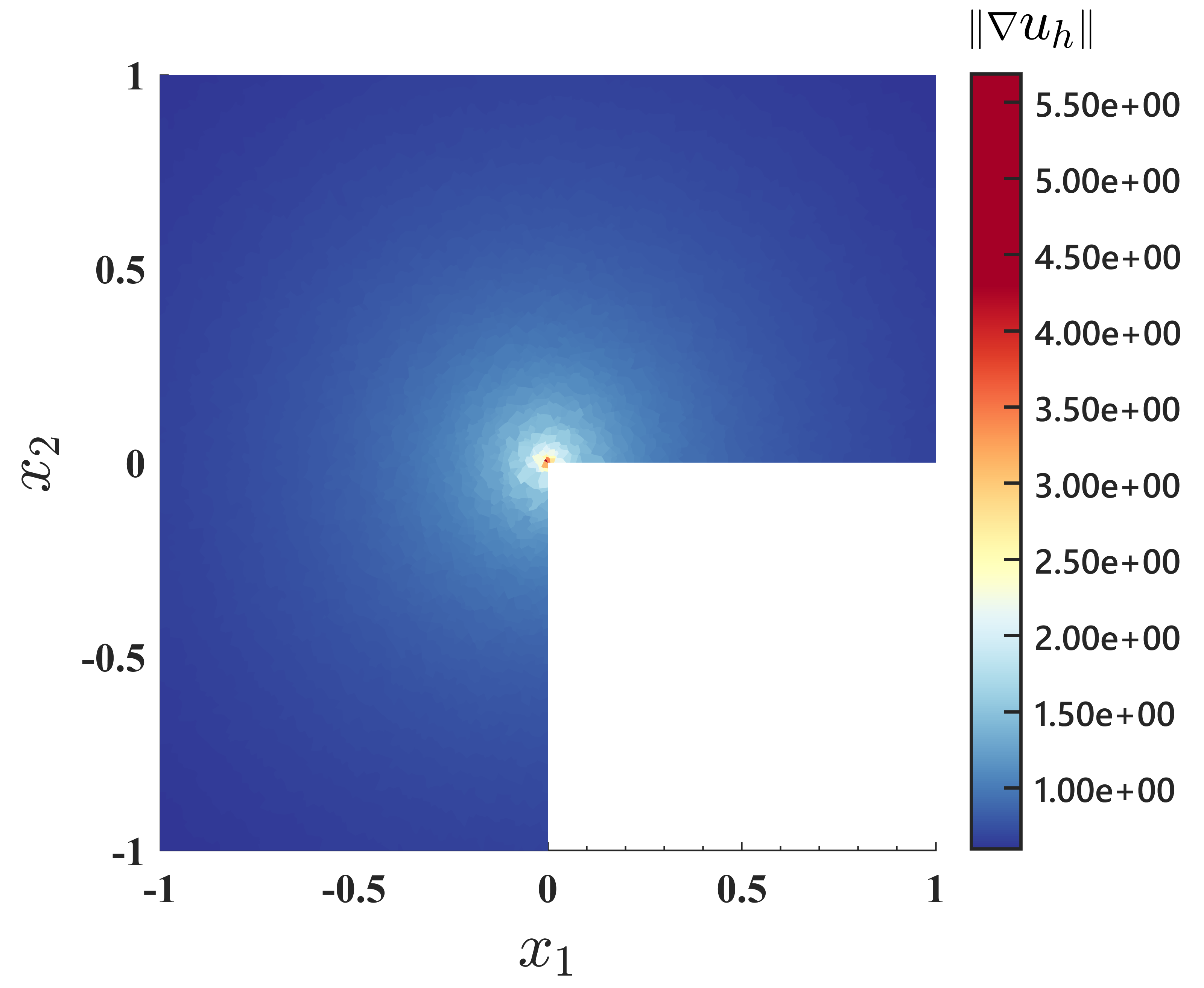}} \hspace{0.5cm}
\subfigure[]{\label{fig:ContourSemiuniformGrad}\includegraphics[width=0.4\textwidth]{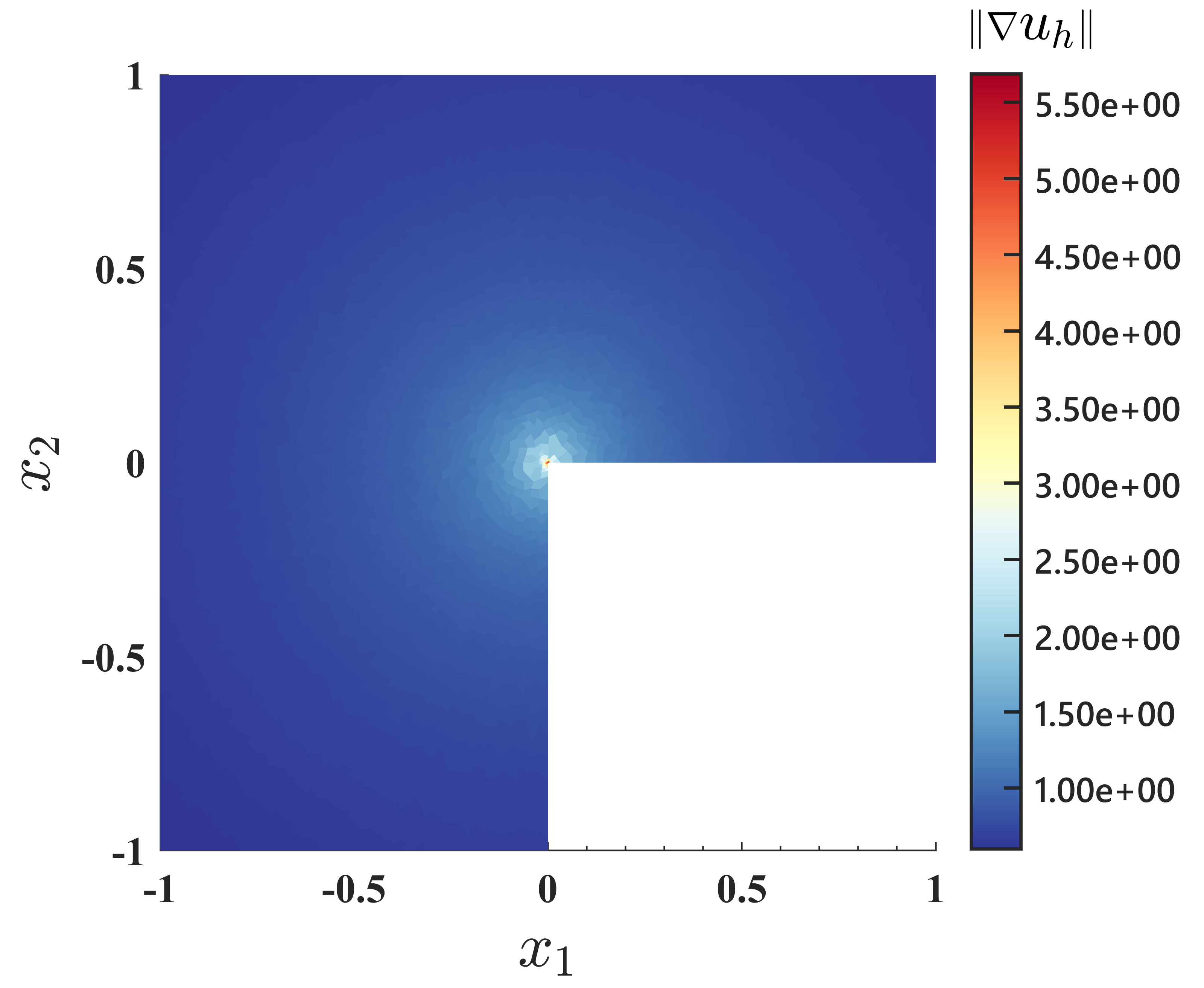}}
\caption{Contour plots of the VEM solution for the $\bm{\nabla}u$ field. \textbf{(a)} Random and \textbf{(b)} semiuniform Polylla meshes.}
\label{figs:ContourPlotsGrad} 
\end{figure}

%
%
\clearpage

\section{Conclusions and ongoing work}
\label{sec:conclusions}

We have presented Polylla, an algorithm to generate a new kind of polygonal mesh using terminal-edge regions. The proposed algorithm takes as input a triangulation with the required point density and generates a polygonal mesh without inserting new points. This generated mesh contains three times less polygons and half the number of points than the standard polygonal meshes based on the Voronoi diagram generated from the same input.

Some preliminary simulations were presented and the numerical results showed that the accuracy and computational cost of the VEM numerical solution using Polylla and Voronoi meshes are similar. We  think that the advantage of Polylla meshes over Voronoi meshes lies in the mesh generation process, where, once the initial triangulation is available, Polylla meshes are simpler and faster to generate since these meshes do not require any addition or removal of vertices.


The ongoing work is devoted to adding more capabilities to Polylla. In particular, since the algorithm retains the underlying triangulation during the polygon construction, we plan to allow further refinements inside the polygonal mesh in a forthcoming version of Polylla. We also are in the process of developing an extension of the Polylla algorithm to polyhedral meshes using  terminal edge-regions in 3D~\cite{HerviasHCF13}.


\section{Acknowledgements}
This research was supported by the Patagón supercomputer of Universidad Austral de Chile (FONDEQUIP EQM180042). The second author thanks to Fondecyt project No 1211484 and the first author to Anid doctoral scholarship 21202379.
\section*{Declarations}

Some journals require declarations to be submitted in a standardised format. Please check the Instructions for Authors of the journal to which you are submitting to see if you need to complete this section. If yes, your manuscript must contain the following sections under the heading `Declarations':

\begin{itemize}
\item Funding: This project was funded by Fondecyt Grant No 1211484 and Anid doctoral scholarship 21202379.
\item Conflict of interest/Competing interests (check journal-specific guidelines for which heading to use): The authors report no conflict of interest.
\item Ethics approval : Not applicable
\item Consent to participate : Not applicable
\item Consent for publication : Not applicable
\item Availability of data and materials : \url{https://github.com/ssalinasfe/Polylla-Mesh}
\item Code availability: Not applicable \url{https://github.com/ssalinasfe/Polylla-Mesh}
\item Authors' contributions : Not applicable
\end{itemize}

\noindent
If any of the sections are not relevant to your manuscript, please include the heading and write `Not applicable' for that section. 

\bigskip
\begin{flushleft}%
    Editorial Policies for:

\bigskip\noindent
Springer journals and proceedings: \url{https://www.springer.com/gp/editorial-policies}

\bigskip\noindent
Nature Portfolio journals: \url{https://www.nature.com/nature-research/editorial-policies}

\bigskip\noindent
\textit{Scientific Reports}: \url{https://www.nature.com/srep/journal-policies/editorial-policies}

\bigskip\noindent
BMC journals: \url{https://www.biomedcentral.com/getpublished/editorial-policies}
\end{flushleft}

\begin{appendices}

\section{Complex geometries}\label{Appendix:complexgeometries}

This appendix presents some examples of Polylla meshes generated using cartography maps obtained from the Library of National Congress of Chile~\footnote{\url{https://www.bcn.cl/siit/mapas_vectoriales}}, and the library PyShp~\footnote{\url{https://pypi.org/project/pyshp/}}. This library was used to read the \texttt{.shp} files, generate the PSLGs and store the information in \texttt{.poly} files. Fig~\ref{fig:loslagos} is a Polylla mesh generated from the PSLG of the {\em Los Lagos} Region, Fig~\ref{fig:magallanes} is a mesh of the {\em Magallanes} Region and Fig~\ref{fig:budi} is a mesh of the {\em Budi} lake.

\begin{figure}[]
    \centering
    \includegraphics[width=0.87\textwidth]{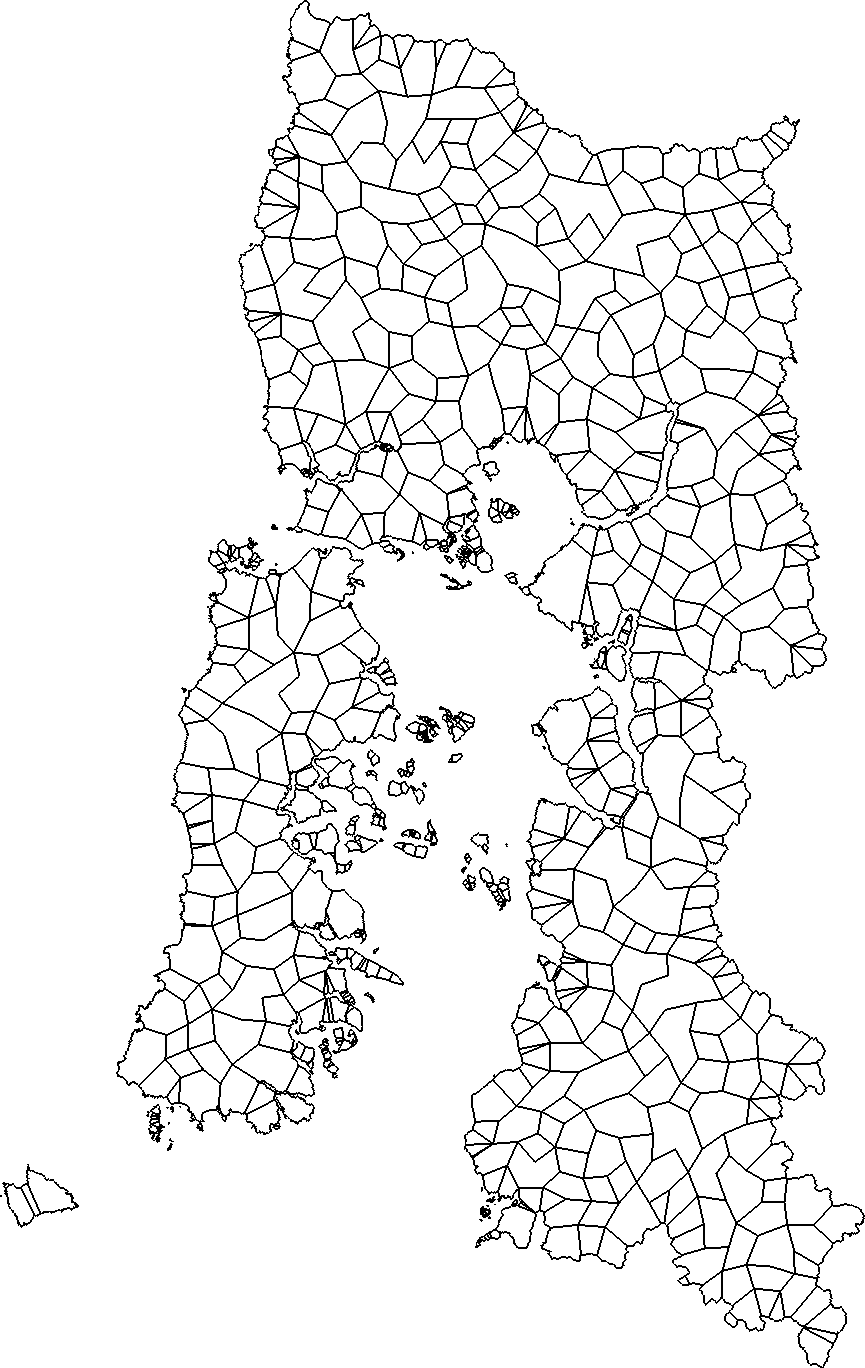}
    \caption{{\em Los Lagos} Region, Chile. The  .poly file contains $294540$ vertices and $294242$ segments. The initial triangulation was generated using Triangle~\cite{triangle2d} with maximum area size of $100000000$ as option.  The resulting triangulation was a conforming Delaunay triangulation with  $309264$ vertices, $309305$ triangles and $618272$ edges. The final Polylla mesh contains $309264$ vertices, $12578$ polygons and $321534$ edges}
    \label{fig:loslagos}
\end{figure}

\begin{figure}[]
    \centering
    \includegraphics[width=\textwidth]{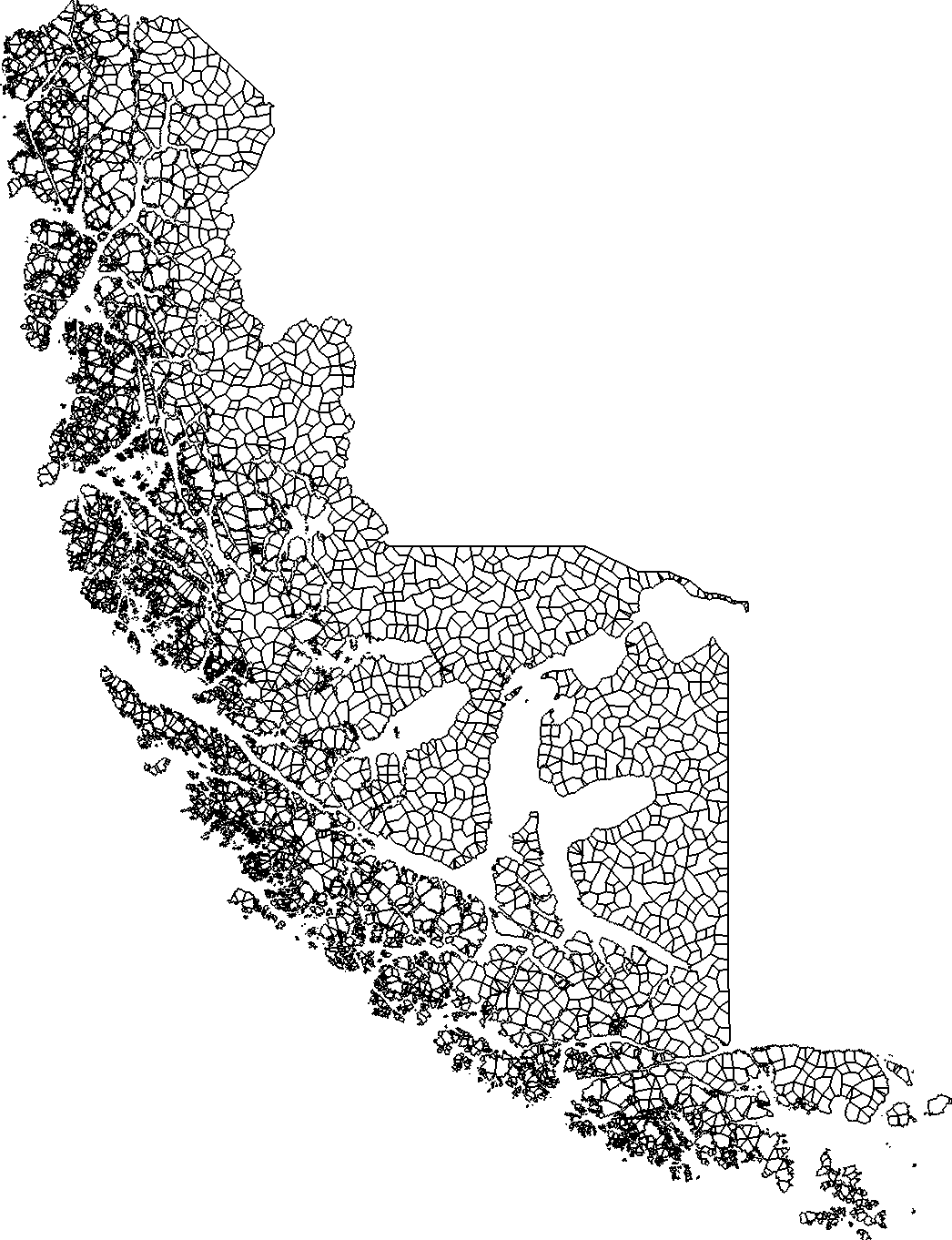}
    \caption{{\em Magallanes} Region, Chile. The .poly file contains $1130733$ vertices and $1130752$ segments. The initial triangulation was generated using Triangle~\cite{triangle2d} with maximum area size of $100000000$ as option. The resulting conforming Delaunay triangulation contains $1148594$ vertices, $1182133$ triangles and $2324968$ edges. The final Polylla mesh contains $1148594$ vertices,  $45300$ polygons and $1186451$ edges.}
    \label{fig:magallanes}
\end{figure}

\begin{figure}[]
    \centering
    \includegraphics[width=\textwidth]{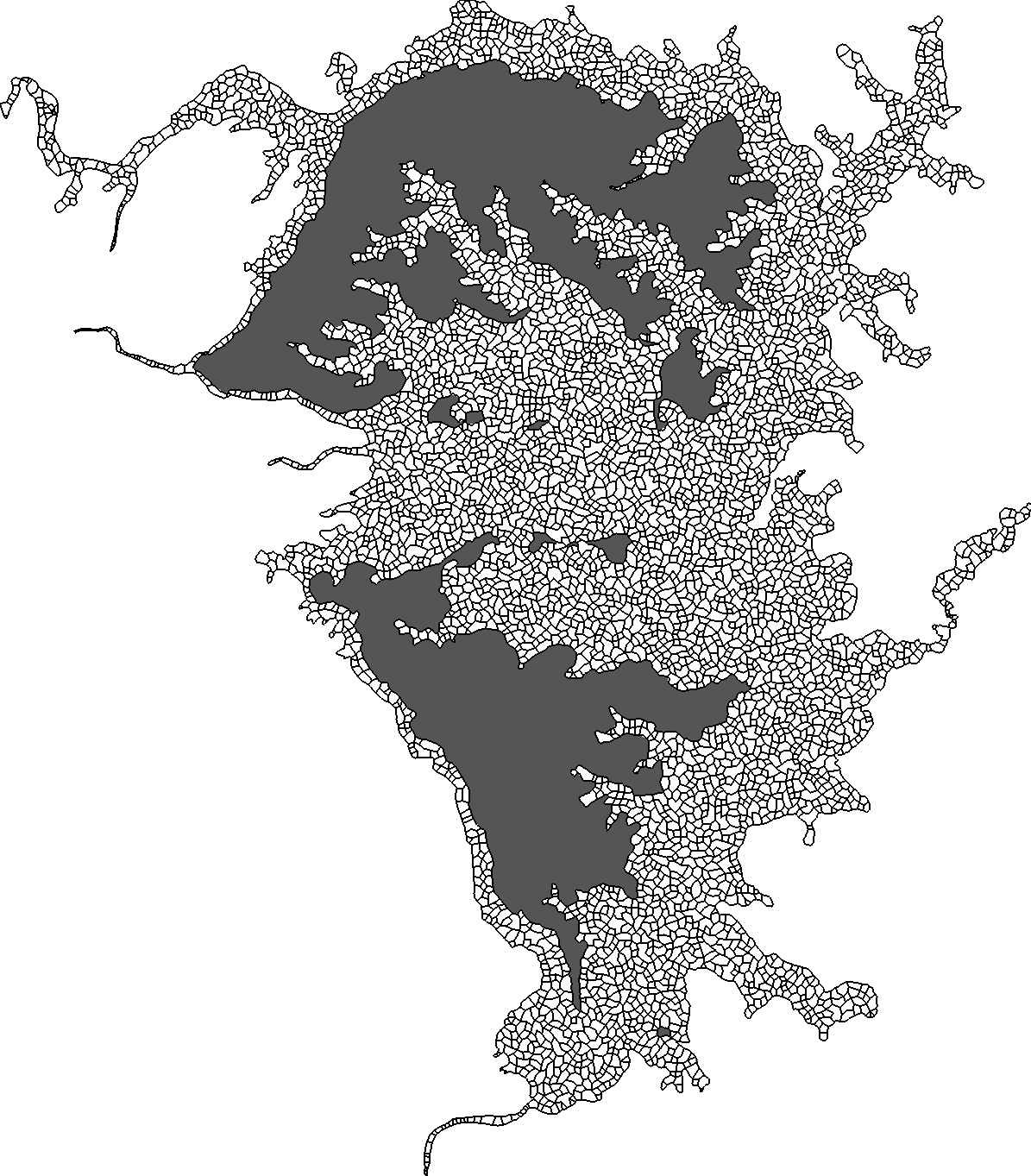}
    \caption{{\em Budi} Lake, Araucanía Region, Chile. The  .poly file contains $5794$ vertices, $5794$ constrained edges and $10$ holes. The initial triangulation was generated using Triangle~\cite{triangle2d} with maximum area size of $10000$ as option. The resulting conforming Delaunay triangulation contains $12608$ vertices, $19406$ triangles and $32023$ edges. The final Polylla mesh contains $12608$ vertices, $5169$ polygons, $17768$ edges and $10$ holes (Islands of the lake). Grey polygons are holes.}
    \label{fig:budi}
\end{figure}

\end{appendices}

\clearpage


\end{document}